\documentclass[conference,svgnames]{IEEEtran}
\IEEEoverridecommandlockouts

\usepackage{amsmath}
\usepackage{amssymb} 
\usepackage{amsthm}
\usepackage{mathtools}
\usepackage{latexsym}

\usepackage{xifthen}
\usepackage{xspace}
\usepackage{xargs}

\usepackage{hyperref}
\usepackage{color}

\usepackage{paralist}
\usepackage{enumitem}
\setlist{leftmargin=5.5mm}

\usepackage{thm-restate}
\usepackage{thmtools}

\usepackage{todonotes}

\usepackage{bussproofs}
\usepackage{orcidlink}

\usepackage[plain,linesnumbered]{algorithm2e}
\DontPrintSemicolon

\usepackage{multirow}

\usepackage{cleveref}
\crefname{algocf}{alg.}{algs.}
\Crefname{algocf}{Algorithm}{Algorithms}
\crefalias{AlgoLine}{line}

\usepackage{tikz}
\usetikzlibrary{arrows,arrows.meta}

\newtheorem*{remark*}{Remark}
\newtheorem{definition}{Definition}
\newtheorem{lemma}{Lemma}
\newtheorem{corollary}{Corollary}
\newtheorem{example}{Example}

\definecolor{orangethema}{RGB}{255,200,150}
\definecolor{bluethema}{RGB}{186,225,246}
\definecolor{pinkthema}{RGB}{204,0,102}

\newcounter{mycommentc}

\usepackage{etoolbox} 

\makeatletter
\patchcmd{\algocf@makecaption@ruled}{\hsize}{\textwidth}{}{} 
\patchcmd{\@algocf@start}{-1.5em}{0em}{}{} 
\makeatother

\newcommand{\multiset}[1]{\{\!\!\{ #1 \}\!\!\}}		                  

\newcommand{\XOR}{\mathrm{XOR}}
\newcommand{\XORh}{\mathrm{XORh}}

\newcommand{\AG}{\mathcal{AG}}
\newcommand{\AGh}{\mathcal{AG}\mathrm{h}}

\newcommand{\X}{\mathcal{X}}                            
\newcommand{\M}{\mathcal{M}}                            
\newcommand{\N}{\mathcal{N}}                            
\newcommand{\V}{\mathcal{V}}                            

\newcommand{\F}{\mathcal{F}}		                    

\newcommand{\T}{\mathcal{T}}		                    
\newcommand{\Pos}[1]{\mathcal{P}os(#1)}                 
\newcommand{\Pset}{\mathcal{P}}                         
\newcommand{\vars}[1]{\mathit{vars}(#1)}                
\newcommand{\names}[1]{\mathit{names}(#1)}              

\newcommand{\dom}[1]{\mathit{dom}(#1)}                  
\newcommand{\img}[1]{\mathit{img}(#1)}                  

\newcommand{\mgu}[2]{\mathrm{mgu}_{#1}(#2)}             
\newcommand{\eqE}[1]{=_{#1}}                            
\newcommand{\eqES}[1]{=^s_{#1}}                         

\newcommand{\R}{\mathcal{R}}				            

\newcommand{\rwstep}[2][]{                              
    \ifthenelse{\isempty{#1}}{\rightarrow_{#2}}{\rightarrow_{#2,#1}}
}
\newcommand{\rwstepC}[2][]{                             
    \ifthenelse{\isempty{#1}}{\rightarrow_{#2}}{\rightarrow_{#2/#1}}
}

\newcommand{\rwsteps}[2][]{                             
    \ifthenelse{\isempty{#1}}{\rightarrow_{#2}}{\rightarrow^*_{#2,#1}}
}
          
\newcommand{\norm}[2][]{{                               
    \ifthenelse{\isempty{#1}}{\downarrow_{#2}}{\downarrow_{#2,#1}}
}}

\newcommand{\minOrd}[2][]{{
    \ifthenelse{\isempty{#1}}{\min_{#2}}{\min_{#2,#1}}  
}}

\newcommand{\symbeval}[1]{\overline{#1}}                
\newcommand{\removeeval}[1]{\underline{#1}}             

\newcommand{\Alg}{\mathcal{A}}                          

\newcommand{\Rn}{{\R_\downarrow}}

\newcommand{\En}{{E_\downarrow}}

\newcommand{\Ea}{{E_\Alg}}

\newcommand{\nf}[3]{\mathrm{nf}_{#1,#2}(#3)}

\newcommand{\st}[1]{\mathrm{st}(#1)}
\newcommand{\sst}[1]{\mathrm{st}^s(#1)}
\newcommand{\OEval}[1]{\Downarrow'_{#1}}

\newcommand{\Eval}[1]{\Downarrow_{#1}}

\newcommand{\sizeEval}[1]{|#1|_{TE}}

\newcommand{\overlapset}[3]{#1 \mathrel{\overset{#2}{\rhd}} #3}
\newcommand{\overlapseteq}[3]{#1 \mathrel{\overset{#2}{\rhd\mspace{-4mu}\lhd}} #3}

\newcommandx{\rightrwlabel}[5][5=]{(\ifthenelse{\isempty{#5}}{}{#5,}#1,#2:\,#3 \rightarrow #4)}
\newcommandx{\leftrwlabel}[5][5=]{(\ifthenelse{\isempty{#5}}{}{#5,}#1,#2:\,#4 \leftarrow #3)}
\newcommandx{\leftrightrwlabel}[6][1=\sim,6=]{(\ifthenelse{\isempty{#6}}{}{#6,}#2,#3:\,#4 #1 #5)}

\newcommand{\rightrwstep}[5][]{\xrightarrow[\ifthenelse{\isempty{#1}}{}{#1,} #2,#3]{#4\ \rightarrow\ #5}}
\newcommand{\leftrwstep}[5][]{\xleftarrow[#2,#3\ifthenelse{\isempty{#1}}{}{,#1} ]{#5\ \leftarrow\ #4}}

\newcommand{\Rightrwstep}[1]{\xRightarrow{#1}}
\newcommand{\Leftrwstep}[1]{\xLeftarrow{#1}}
\newcommand{\Leftrightrwstep}[1]{\xLeftrightarrow{#1}}

\newcommand{\reverse}[1]{#1^{-1}}          
\newcommand{\pos}{\mathit{pos}}                     

\newcommand{\lessRwLbl}[1]{\prec_{#1}}           
\newcommand{\eqRwLbl}[1]{\backsimeq_{#1}}      
\newcommand{\leqRwLbl}[1]{\preccurlyeq_{#1}}           

\newcommand{\measureF}{\mathsf{rf}}
\newcommand{\measure}[1]{\mathsf{rf}(#1)}

\newcommand{\measureTerm}[1]{\mathsf{rf}_2(#1)}

\newcommand{\RelPara}{\approx_{\para}}              

\newcommand{\tikzrightrwstep}[7][]{
    \draw[->] (#2) edge node[auto,sloped] {\tiny$#5 \rightarrow #6$} node[auto,sloped,below] {\tiny$\ifthenelse{\isempty{#1}}{}{#1,} #3, #4$} (#7);
}
\newcommand{\tikzleftrwstep}[7][]{
    \draw[<-] (#2) edge node[auto,sloped] {\tiny$#6 \leftarrow #5$} node[auto,sloped,below] {\tiny$#3, #4\ifthenelse{\isempty{#1}}{}{, #1}$} (#7);
}

\newcommand{\tikzRightrwstep}[3]{
    \draw (#1) edge[double equal sign distance,-Implies] node[auto,sloped] {\tiny$#2$} (#3);
}
\newcommand{\tikzLeftrwstep}[3]{
    \draw (#1) edge[double equal sign distance,Implies-] node[auto,sloped] {\tiny$#2$} (#3);
}

\newcommand{\rwstepOrd}[5][]{
    \ifthenelse{\isempty{#1}}{
        \xrightarrow{#4\,:\,#2 \rightarrow #3}_{#5}
    }{
        \xrightarrow{#1,#4\,:\,#2 \rightarrow #3}_{#5}
    }
}

\newcommand{\rwLeftRightlabel}[7][]{                            
    #6,#7,#2,#3:\,#4 \ifthenelse{\isempty{#1}}{\sim}{#1} #5
}
\newcommand{\rwLeftrightlabel}[7][]{                            
    #6,#7,#2,#3:\,#4 \mathrel{\ifthenelse{\isempty{#1}}{\sim}{#1}} #5
}
\newcommand{\rwLeftrightstep}[7][]{                             
    \xleftrightarrow[\ifthenelse{\isempty{#6}}{}{#6,} #2,#3 \ifthenelse{\isempty{#7}}{}{,#7}]{#4\ \ifthenelse{\isempty{#1}}{\sim}{#1}\ #5}
}
\newcommand{\rwRightstep}[5][]{                             
    \xrightarrow[\ifthenelse{\isempty{#1}}{}{#1,} #2,#3]{#4\ \rightarrow\ #5}
}

\newcommand{\rwstepRightFull}[6]{
    \xrightarrow[\ifthenelse{\isempty{#5}}{}{#5,}#1,#2 \ifthenelse{\isempty{#6}}{}{,#6}]{\,#3 \rightarrow #4}
}

\newcommand{\rwLeftrightStep}[1]{\xLeftrightarrow{#1}}
\newcommand{\rwLeftStep}[1]{\xLeftarrow{#1}}
\newcommand{\rwRightStep}[1]{\xRightarrow{#1}}

\newcommand{\rwstepLeftFull}[6]{
    \xleftarrow[\ifthenelse{\isempty{#6}}{}{#6,}#1,#2 \ifthenelse{\isempty{#5}}{}{,#5}]{\,#4 \leftarrow #3}
}

\newcommand{\rwstepOrdLeft}[5][]{
    \ifthenelse{\isempty{#1}}{
        \prescript{}{#5}{\xleftarrow{#4\,:\,#2 \rightarrow #3}}
    }{
        \prescript{}{#5}{\xleftarrow{#1,#4\,:\,#2 \rightarrow #3}}
    }
}
\newcommand{\rwstepOrdEq}[5][]{
    \ifthenelse{\isempty{#1}}{
        \xleftrightarrow{#4\,:\,#2 \rightarrow #3}_{#5}
    }{
        \xleftrightarrow{#1,#4\,:\,#2 \rightarrow #3}_{#5}
    }
}

\newcommand{\normstep}[1][]{\rightsquigarrow_{#1}}

\newcommand{\PosX}[1]{\mathcal{P}os_{\setminus \X}(#1)}
\newcommand{\para}{\mathrel{||}}

\newcommand{\amin}{\mathsf{a}_{\mathrm{min}}}
\newcommand{\bin}{\mathsf{b}}
\newcommand{\tr}{\mathsf{tr}}

\newcommand{\RNormR}{\textsc{NormR}\xspace}
\newcommand{\RNormL}{\textsc{NormL}\xspace}
\newcommand{\REq}{\textsc{Eq}\xspace}
\newcommand{\RSubsume}{\textsc{Sub}\xspace}
\newcommand{\RVar}{\textsc{Var}\xspace}
\newcommand{\ROrd}{\textsc{Ord}\xspace}

\SetKwProg{Fn}{Function}{}{}
\SetKwFunction{GenExtended}{generate\_rw\_th}
\SetKwFunction{GenExtendedPlus}{generate\_rw\_th+}
\SetKwFunction{Normalize}{normalise}
\SetKwFunction{Cleanup}{cleanup}
\SetAlCapSkip{0.3cm}

\usepackage{amsmath}

\title{%
  Automatic verification of Finite Variant Property beyond
  convergent equational theories
\thanks{This work received funding from the France 2030 program
		managed by the French National Research Agency under grant
		agreement No. ANR-22-PECY-0006.}}

\author{
  \IEEEauthorblockN{Vincent Cheval\,\orcidlink{0000-0002-3622-2129}}
  \IEEEauthorblockA{
    University of Oxford, United Kingdom}
  \and
  \IEEEauthorblockN{Caroline Fontaine\,\orcidlink{0000-0001-8184-7366}}
  \IEEEauthorblockA{
    Universit\'e Paris-Saclay, CNRS, ENS Paris-Saclay, \\
    Laboratoire M\'ethodes Formelles, France}
}

\begin{document}

\maketitle


\pagestyle{plain} 
\thispagestyle{plain}

\begin{abstract}

Computer-aided analysis of security protocols heavily relies on equational theories to model cryptographic primitives. Most automated verifiers for security protocols focus on equational theories that satisfy the \emph{Finite Variant Property (FVP)}, for which solving unification is decidable. However, they either require to prove FVP by hand or at least to provide a representation as an $E$-convergent rewrite system, usually $E$ being at most the equational theory for an associative and commutative function symbol (AC). 
The verifier ProVerif is probably the only exception amongst these tools as it automatically proves FVP without requiring a representation, but on a small class of equational theories.
In this work, we propose a novel semi-decision procedure for proving FVP, without the need for a specific representation, and for a class of theories that goes beyond the ones expressed by an $E$-convergent rewrite system. We implemented a prototype and successfully applied it on several theories from the literature. 
\end{abstract}


\section{Introduction}
\label{sec:intro}

Cryptographic protocols are distributed programs designed to secure communications. They are used in diverse critical applications, such as electronic voting, cryptocurrencies, secure messaging, online payment, etc. Each protocol comes with its own list of security claims, some of which are common to all applications (e.g., secrecy and authentication) and others of which are tailored to specific domains (e.g., coercion resistance in electronic voting).

The design of such protocols being notoriously error prone, it is now standard practice to provide formal proofs that the target security properties are really satisfied. Over the years, there have been many success stories on the symbolic analysis of impactful real-life security protocols, for example TLS~\cite{DBLP:conf/sp/BhargavanBK17,DBLP:conf/sp/CremersHSM16,DBLP:conf/ccs/CremersHHSM17}, ECH~\cite{DBLP:conf/ccs/BhargavanCW22}, Signal~\cite{DBLP:conf/eurosp/KobeissiBB17}, 5G-AKA~\cite{DBLP:conf/ccs/BasinDHRSS18}, Noise~\cite{DBLP:conf/eurosp/KobeissiNB19}, EMV~\cite{DBLP:conf/sp/BasinST21} and IKEv2~\cite{DBLP:conf/acsac/GazdagGGHL21}. 
These successes were partly due to advances in the capabilities of automatic verifiers for cryptographic properties, as ProVerif~\cite{manual-proverif,DBLP:conf/sp/BlanchetCC22}, Tamarin~\cite{manual-tamarin,DBLP:journals/ieeesp/BasinCDS22}, Maude-NPA~\cite{manual-maude-npa,DBLP:journals/tcs/EscobarMM06}, DeepSec~\cite{DBLP:conf/sp/ChevalKR18,DBLP:conf/cav/ChevalKR18} and AKiSs~\cite{DBLP:journals/tocl/ChadhaCCK16,DBLP:conf/csfw/BaeldeDGK17}.

All the above-mentioned tools rely on similar underlying symbolic models. In particular, messages exchanged over the network are expressed by terms constructed out of function symbols (representing cryptographic primitives), names (representing large numbers such as keys), and variables.  Here, a term $\mathsf{enc}(m,k)$ represents the symmetric encryption of the plain text $m$ with the key $k$ using the encryption algorithm $\mathsf{enc}$. The algebraic properties of these cryptographic primitives are expressed by means of an \emph{equational theory}, comprising of equations between two terms that indicate that they
correspond to the same message. In our example, an encryption algorithm comes with an associated decryption algorithm, denoted here $\mathsf{dec}$. The link between $\mathsf{enc}$ and $\mathsf{dec}$ can be expressed by the equation $\mathsf{dec}(\mathsf{enc}(x,y),y) = x$. A set $E$ of such equations is an equational theory and induces an equivalence relation $=_E$ which represents the terms \emph{equal modulo $E$}. For example, $\mathsf{dec}(\mathsf{enc}(\mathsf{h}(a),k),k) =_E \mathsf{h}(a)$ but $\mathsf{dec}(\mathsf{enc}(\mathsf{h}(a),k),k') \neq_E \mathsf{h}(a)$ as the decryption key $k'$ is different from the encryption key $k$.

Checking equality of two terms modulo an equational theory $E$ is the most basic problem that these automatic tools must decide. However, as they consider an active attacker that controls the network, they need to verify the security claims for every possible message crafted by the attacker that would be accepted by honest participants.
Among other things, this entails checking  whether two terms can \emph{unify modulo the equational theory $E$}, that is, for terms $u$ and $v$, checking whether there exists a substitution $\sigma$, i.e. a mapping from variables to terms, such that $u\sigma =_E v\sigma$. Such a $\sigma$ is called a \emph{unifier}. Not only do the tools need to
decide whether such a unifier exists,
they also need the ability to compute the set of \emph{most general unifiers}.
This is a computationally hard problem that heavily depends on the equations in $E$ and that is undecidable in general.

The problem of unification modulo an equational theory has been extensively studied~\cite{DBLP:books/el/RV01/BaaderS01,DBLP:conf/ccl/BoudedC94,KIRCHNER1989171,DBLP:conf/fscd/Ayala-RinconFSS22,DBLP:journals/jsc/BaaderS96}. A very successful technique consists in finding a representation of $E$ as a \emph{convergent rewrite system} and applying \emph{basic narrowing}~\cite{DBLP:conf/rta/RetyKKL85,DBLP:journals/jsc/NuttRS89}. A rewrite system $\R$ is a set of oriented equations also called \emph{rules}, which allows rewriting a term $u$ to another term $v$ through a \emph{rewrite step}, denoted $u \rwsteps{\R} v$. For example, when $\R = \{ \mathsf{dec}(\mathsf{enc}(x,y),y) \rightarrow x \}$, we have $\mathsf{g}(\mathsf{dec}(\mathsf{enc}(\mathsf{h}(a),k),k),b) \rwsteps{\R} \mathsf{g}(\mathsf{h}(a),b)$.  Convergence implies that the rewrite system is \emph{terminating}, i.e., it does not allow infinite sequences of rewrite steps, and that it is \emph{confluent}, i.e., all sequences of rewrite steps from an initial term $t$ end up  in the same term.  Basic narrowing does not always terminate, but it is a generic technique that covers a large class of cryptographic primitives. This is for instance the underlying technique used by the verifier DeepSec~\cite{DBLP:conf/sp/ChevalKR18,DBLP:conf/cav/ChevalKR18}.  However, basic narrowing fails when the equational theory contains function symbols that are associative and commutative (AC), such as Exclusive-Or (XOR).
Not only does basic narrowing fail as such equational theory cannot be represented by a terminating rewrite system, but even basic AC-narrowing (that is basic narrowing modulo AC) was shown to be incomplete~\cite{DBLP:conf/rta/Comon-LundhD05}.  

Many works have thus developed ad-hoc algorithms solving the unification problem for specific equational theories: AC~\cite{DBLP:conf/ccl/BoudedC94,DBLP:journals/jar/Boudet93,DBLP:journals/jsc/Fages87,DBLP:journals/jacm/Stickel81,DBLP:conf/lics/BoudetCD90}, Abelian Groups ($\mathcal{AG}$)~\cite{10.1007/3-540-44881-0_13,MR0749246}, AC plus unit, idempotency (ACUI) and distributivity~\cite{DBLP:journals/jar/AnantharamanNR04}, Exclusive-Or with homomorphism (XORh)~\cite{DBLP:conf/cade/LiuL11}, $\mathcal{AG}$ with homomorphism ($\mathcal{AG}h$)~\cite{DBLP:journals/iandc/LiuL14}, etc. However, in the context of the automated verification of cryptographic protocols, ad-hoc procedures are not sufficient. The ever-increasing complexity of protocols and their cryptographic primitives require tools that support user-defined equational theories and primitives. 

In 2005, Comon and Delaune introduced the \emph{finite variant property} (FVP)~\cite{DBLP:conf/rta/Comon-LundhD05} for equational theories. They present a general technique, showing that when an equational theory $E$ satisfies the FVP modulo some smaller equational theory $E'$, the problem of unification modulo $E$ can be reduced to the problem of unification modulo $E'$, using what is now known as \emph{$E$-variant narrowing}. This is why FVP is of particular interest in the case of protocol analyses, as many user-defined equational theories can be reduced in practice to unification modulo AC. Since 2005, and although the definition of FVP is quite general, most efforts in the literature for proving the FVP have focused on cases where $E$ can be represented by a $E'$-convergent rewrite system either by hand~\cite{DBLP:conf/rta/Comon-LundhD05,DBLP:conf/wrla/Meseguer20} or automatically~\cite{DBLP:conf/rta/EscobarMS08,cholewa2014variants}.

It is important to notice that 
using the FVP to solve the unification problem is usually less efficient than using a bespoke algorithm. This is partly due to the fact that unification modulo AC is extremely costly, even on small problems, leading to significant decreases in performance, whereas unification modulo XOR and $\mathcal{AG}$ are much more efficient. As an example, the unification of $x + x + x + x$ and $x_1 + x_2 + x_3 + x_4$ modulo AC has more than 34 billion most general unifiers \cite{DBLP:journals/jar/Boudet93} whereas the unification modulo XOR has only 57. Nevertheless, using the FVP for solving unification remains in our opinion the most successful technique in practice.

The same year Comon and Delaune introduced the FVP, Blanchet, Abadi and Fournet published a paper \cite{DBLP:journals/jlp/BlanchetAF08} on an extension of the ProVerif tool in which they handled the equational theories by introducing the notion of \emph{extended signature modelling an equational theory}. In~\cite{DBLP:journals/jlp/BlanchetAF08}, they described two procedures to create such extended signatures from an equational theory given as input: one for $E'$-convergent equational theories with $E' = \emptyset$, and the other for \emph{linear equational theories} (i.e., those for which in all equations  each variable occurs at most once in the left-hand side and at most once in the
right-hand side) that cannot be expressed as a convergent rewrite system. 
Interestingly, their notion of extended signatures modelling an equational theory~\cite{DBLP:journals/jlp/BlanchetAF08} can be shown to coincide with the FVP when $E' = \emptyset$. 

As mentioned, most past works on the FVP have been focusing on equational theories $E$ that can be represented by an $E'$-convergent rewrite system $\R$ where $E' = AC$. Furthermore, they also require the rewrite system $\R$ to be provided in order to either compute $E$-variants or check the FVP. This is the case for the Maude-NPA and Maude tools. The Tamarin tool is even further restrictive as it only allows four groups of built-in function symbols with associative-commutative properties: Exclusive-OR, Diffie-Hellman groups (Abelian Group with an exponentiation operator), Bilinear-pairing and multiset (simple AC function symbol). In particular, 
\begin{inparaenum}[(i)]
    \item the function symbols defined in these equational theories cannot be used in user-defined equational theories; 
    \item user-defined primitives cannot have function symbols with associative-commutative properties.
\end{inparaenum}

The limitations in Tamarin partly come from the fact that creating an $AC$-convergent rewrite system $\R$ that satisfies the FVP can be very tricky. In the case of $\mathcal{AG}$, the only known representation showing the FVP is a peculiar rewrite system containing 10 rewrite rules, first proposed by Lankford~\cite{hullot1980catalogue,DBLP:conf/rta/Comon-LundhD05}.
In contrast, we make the following observations on the algorithms proposed in~\cite{DBLP:journals/jlp/BlanchetAF08} and currently used in ProVerif:
\begin{itemize}
\item they can be used to effectively compute variants of terms and most general unifiers;
\item they do not require any input other than the equational theory $E$ itself (no mandatory rewrite system requested);
\item they handle some theories that are not convergent (e.g., they handle the linear theory $\mathrm{exp}(\mathrm{exp}(g,x),y) = \mathrm{exp}(\mathrm{exp}(g,y),x)$);
\item they terminate only if $E$ has the FVP modulo the empty theory. In particular, they cannot handle theory modulo AC. 
\end{itemize}

\subsubsection*{Our contributions} This present paper generalises and improves upon the framework and algorithms of~\cite{DBLP:journals/jlp/BlanchetAF08,DBLP:journals/jlp/BlanchetAF08}. Our main contributions are as follows:
\begin{itemize}
    \item we introduce a new notion of \emph{Rewrite Theory mimicking an equational theory $E$} that implies the FVP modulo an equational theory beyond AC, e.g. $\XOR$, $\XORh$, $\AG$, $\AGh$; thereby avoiding the costly unification modulo AC;
    \item we show under which conditions our framework coincides with the FVP;
    \item we provide a new semi-decision procedure that can automatically transform an equational theory $E$ into a rewrite theory that mimics $E$. We lift the convergent or linear restrictions from~\cite{DBLP:journals/jlp/BlanchetAF08} as our procedure requires no initial condition on $E$;
    \item we consider some optimisation techniques for both scope and efficiency;
    \item we have implemented in a prototype our algorithm restricted to FVP modulo AC and successfully tested it on several theories with known FVP modulo AC from the literature (e.g. $\XOR$, $\AG$, Diffie-Hellman bilinear pairing). We also showed that many of these equational theories augmented with a homomorphic symbol with different group operators also satisfy the FVP. This is for example the case for $\XOR$, $\AG$, ElGamal and (a)symmetric encryption. Finally, we also provide several toy examples showing that the prototype also works on equational theory without AC-convergent rewrite system representation.
    The code and examples are available at~\cite{prototypeFVP}. 
    \end{itemize}
We believe that our work is generic enough to be ported to most automatic verifiers of security protocols, leading to a significant improvement in their scope, usability and efficiency. 

\subsubsection*{Outline} In \Cref{sec:prelim}, we provide some preliminary definitions, including the finite variant property. \Cref{sec:extended-signature} presents the framework with the new notion of rewrite theory and \Cref{sec:variant} explains how it relates to the finite variant property. In \Cref{sec:generation-extended-signature}, we shall describe our semi-decision procedure for generating rewrite theories and present a detailed overview of its proof of correctness. \Cref{sec:optimisations} focuses on three possible optimisations to our procedure. Finally, in \Cref{sec:discussion}, we present experimental results on a subset of equational theories that can be handled by our prototype and discuss some adjacent properties of our procedure and its limitations.


\section{Preliminaries}
\label{sec:prelim}


\subsection{Terms} 

We use classical notation and terminology from \cite{DBLP:conf/rta/Comon-LundhD05} on terms, unification, rewrite systems. Let $\X$ be an infinite set of variables and $\N$ be an infinite enumerable set of names (also called \emph{free constant} in the rewriting community). 
The set $\T(\F,\X \cup \N)$ consists  of all terms built over the finite ranked alphabet $\F$ of function symbols, variables from  $\X$ and names from $\N$. A term $t$ is \emph{ground} when $t \in \T(\F,\N)$.
The set of positions of a term $t$ is written $\Pos{t}$.
The subterm of $t \in \T(\F,\X \cup \N)$ at position $p \in \Pos{t}$ is written $t|_p$. We denote $\st{t}$ (resp. $\sst{t}$) the set of all (resp. strict) subterms of $t$.
The term obtained by replacing $t|_p$ with $u$ is denoted $t[u]_p$. The set of variables occurring in $t$ is denoted $\vars{t}$. The set of names occurring in $t$ is denoted $\names{t}$.

A context $C[\_]$ is a term where one of its subterms is replaced by a hole denoted by $\_$. Given a term $u$, $C[u]$ denotes the term obtained from $C[\_]$ by replacing the hole $\_$ with $u$.

A \emph{substitution} $\sigma$ is a mapping from $\X$ to $\T(\F,\X \cup \N)$. The \emph{domain} of $\sigma$, denoted $\dom{\sigma}$, is the set of variables $x$ such that $x\neq x\sigma$. The \emph{image} of $\sigma$, denoted $\img{\sigma}$, is the set of terms $x\sigma$ where $x\in \dom{\sigma}$. Given $\mathsf{N},\mathsf{N}' \subseteq \N$, a bijective mapping $\rho$ from $\mathsf{N}$ to $\mathsf{N}'$ is called a \emph{renaming}. We will denote $\dom{\rho} = \mathsf{N}$ and $\img{\rho} = \mathsf{N}'$.


\subsection{Equational theories} 

An \emph{equational theory} $E$ is a set of equations (unordered pairs of terms) from $\T(\F,\X)$, meaning that they should not contain names. It induces the relation  $=_E$ which is the least congruence on $\T(\F,\X \cup \N)$ such that $u\sigma =_E v\sigma$ for all pairs $u = v \in E$ and substitutions $\sigma$.
$E$ is \textit{regular} if for all equations $t_1 = t_2 \in E$, $vars(t_1) = vars(t_2)$. $E$ is \emph{trivial} if for all terms $s, t \in \T(\F,\X \cup \N)$, $s =_E t$.

Two terms $s, t \in \T(\F,\X \cup \N)$ are \emph{$E$-unifiable} if there exists a substitution $\sigma$, called \emph{$E$-unifier}, such that $s\sigma =_E t\sigma$. For example, using the infix notation for $+$, the set $E_\oplus$ of equations for XOR is as follows:
\[
    E_\oplus = \left\{ \begin{array}{c}x + y = y + x, x + (y + z) = (x + y) + z,\\
    x + x = 0 , x + 0 = x \end{array}\right\}
\]
The equations for Abelian Groups ($\mathcal{AG}$) include $x * 1 = x$ and $x * inv(x) = 1$ with AC equations for $*$.
We say that a finite set of substitutions $S = \{\sigma_1,\ldots,\sigma_n\}$ is a \emph{complete set of $E$-unifiers of $s,t$} if for all $E$-unifiers $\sigma$ of $s,t$, there exist $i \in \{1,\ldots,n\}$ and a substitution $\theta$ such that for all $x \in \vars{s,t}$, $x\sigma =_E x\sigma_i\theta$. We say that $S$ is a \emph{complete set of most general $E$-unifiers of $s,t$}, denoted by $\mgu{E}{s,t}$, when it is a complete set of $E$-unifiers and no substitution is an instance of another, i.e. for all $\sigma_1,\sigma_2 \in S$, there is no substitution $\alpha$ such that $\sigma_1 =_E \sigma_2\alpha$.
Note that it is easy to derive an algorithm that produces finite complete sets of most-general $E$-unifiers from an algorithm that produces finite complete sets of $E$-unifiers. 


\subsection{Rewriting}

A \emph{term rewrite system} (TRS) is a finite set of \textit{rewrite rules} $\ell \rightarrow r$ where $\ell, r \in \T(\F,\X)$. 
A term $s \in \T(\F,\X \cup \N)$ rewrites to $t$ by a TRS $\R$, denoted $s \rwstep{\R} t$, if there exist $\ell \rightarrow r$ in $\R$, $p \in \Pos{s}$ and a substitution $\sigma$ such that $s|_p = \ell\sigma$ and $t = s[r\sigma]_p$. 

Given a rewrite system $\R$ and an equational theory $E$, \emph{$s$ rewrites into $t$ by $\R$ modulo $E$}, denoted $s \rwstep[E]{\R} t$, iff there exist a position $p \in \Pos{s}$, a rule $\ell \rightarrow r \in \R$ and a substitution $\sigma$ such that $s|_p =_E \ell\sigma$ and $t = s[r\sigma]_p$. The relation $\rwstep[E]{\R}$ has usually been used~\cite{DBLP:conf/icalp/JouannaudKK83} to implement the larger relation ${\rwstepC[E]{\R}} = (=_E \circ \rwsteps[]{\R} \circ =_E)$.

A rewrite system $\R$ is \emph{$E$-confluent} if and only if for all $s,t \in \T(\F,\X \cup \N)$, if $s =_{(\R^=) \cup E} t$ then there exist $s',t'$ such that $s \rwsteps[E]{\R} s'$, $t \rwsteps[E]{\R} t'$ and $s' =_E t'$, where $\R^= = \{ \ell = r \mid \ell \rightarrow r \in \R\}$ and $\rwsteps[E]{\R}$ is the reflexive transitive closure of $\rwstep[E]{\R}$. 
We say that $\R$ is \emph{$E$-terminating} when the relation $\rwstepC[E]{\R}$ is well founded. $\R$ is said to be \emph{$E$-convergent} when it is both $E$-terminating and $E$-confluent. Finally, we also say that $\R$ is $E$-confluent (resp. convergent) for $E'$ when $\R$ is $E$-confluent (resp. convergent) and $(\R^=) \cup E$ is the same relation as $E'$, i.e. $s =_{(\R^=) \cup E} t$ iff $s =_{E'} t$.

\begin{example}
\label{ex:XOR}
Defining $\R = \{ x + 0 \rightarrow x; x + x \rightarrow 0 \}$ and $\R_{\oplus} = \R \cup \{ x + (x + y) \rightarrow y \}$, it is well known that $\R_{\oplus}$ is $AC$-convergent for XOR but not $\R$. Indeed, the term $t_1 = (a + a) + b \rwsteps[AC]{\R} b$ but the term $t_2 = a + (a + b)$ cannot be rewritten by $\rwstep[AC]{\R}$ since neither $a$ nor $(a + b)$ can be $AC$-matched with the left hand side of a rule in $\R$. In other words, $t_2 \not\rwstep[AC]{\R}$. Therefore, $\R$ is not $AC$-confluent.

This also illustrates the difference between $\rwstepC[AC]{\R}$ and $\rwstep[AC]{\R}$, as $t_2 \rwstepC[AC]{\R} 0 + b \rwstepC[AC]{\R} b$.
\end{example}

\subsection{Ordering terms}
A strict order on terms $>$ is said to be a \emph{rewrite order} when it is closed by application of contexts (\emph{i.e.}, $t > s$ implies $C[t] > C[s]$) and substitutions (i.e., $t > s$ implies $t\sigma > s\sigma$). When the rewrite order is well-founded, it is called a \emph{reduction order}. 
An order $>$ is \emph{$E$-compatible} if $s' =_E s > t =_E t'$ implies $s' > t'$. We say that $>$ is \emph{$E$-total on ground terms} ($E$-total for short) if for all $s,t \in \T(\F,\N)$, $s \neq_E t$ implies $s > t$ or $t > s$. Furthermore, $>$ is \emph{$E$-compatible with a rewrite system $R$} when it is $E$-compatible and for all $\ell \rightarrow r \in \R$, we have $\ell > r$. It is well known that a rewrite system $\R$ is $E$-terminating if and only if there exists a reduction order $E$-compatible with $\R$.
An $E$-compatible reduction order induces a well-defined order on the set of $=_E$-equivalence classes. Moreover, when $>$ is in addition $E$-total, $>$ becomes total on $\T(\F,\N)/_{=_E}$. 
In such a case, given another equational theory $E'$ such that $E \subseteq E'$, we define $\minOrd[>]{E'}(t)$ as the smallest element of $\T(\F,\N)/_{=_E}$ by $>$ such that $t =_{E'} \minOrd[>]{E'}(t)$. In other words, $\minOrd[>]{E'}(t)$ is the smallest term modulo $E$ by $>$ that is equal modulo $E'$ to $t$.

\begin{example}
\label{ex:reduction order count}
Consider the order as follows: $t >_\# s$ when
\begin{itemize}
\item the number of function symbols or names in $t$ is greater than in $s$
\item for all $x \in \X$, the number of occurrences of $x$ in $t$ is greater than in $s$
\item at least one of these inequalities is strict.
\end{itemize}
The order $>_\#$ is a reduction order AC-compatible but not AC-total.
\end{example}

\begin{remark*}
It was already shown in~\cite{DBLP:conf/lics/Baader97} that the existence of a non-empty $E$-compatible reduction order necessarily implies $E$ to be regular. In other words, when $E$ is not regular, no constraint system can be $E$-terminating. Indeed, consider an equation $u = C[x]$ in $E$ with $x$ a variable not in $u$, and a rewrite rule $\ell \rightarrow r$ in $\R$. We have $C[\ell] \rwstep{\R} C[r] =_E C[\ell]$. Hence, $=_E \circ \rwstep[]{\R} \circ =_E$ is not well founded.
\end{remark*}

There exist many ways to implement a reduction order, specially when $E = \emptyset$. The most commonly used in the literature are \emph{recursive path ordering} (RPO)~\cite{DBLP:journals/ipl/JouannaudL82} and \emph{lexicographic path ordering} (LPO)~\cite{DBLP:journals/jar/Lescanne90}. They enjoy nice properties such as the \emph{subterm property} (i.e. $C[s] > s$ for all non-empty contexts $C[\_]$) and being total on ground terms. When $E \neq \emptyset$, it is more complicated to create an $E$-compatible and $E$-total reduction order but previous works have been successful when considering theories such as AC~\cite{DBLP:conf/rta/RubioN93,DBLP:conf/rta/Rubio99,DBLP:journals/tcs/RubioN95,DBLP:journals/tplp/0002WHM16} or permutations~\cite{DBLP:conf/fscd/0001L21}. 

\begin{example}
\label{ex:order AC}
Consider $\F = \F_{AC} \cup \F_o$ where $\F_{AC}$ are the binary function symbols with associative and commutative properties, and $\F_o$ are all the other function symbols.

To define the order $>_{AC}$ from \cite{DBLP:conf/rta/Rubio99}, we first need to consider a strict order on the function symbols, denoted $>_\F$. Additionally, we need consider the terms in their \emph{flattened form}, e.g., $f(u_1,f(u_2,u_3))$ with $f \in \F_{AC}$ is represented as $f(u_1,u_2,u_3)$. This is a classic representation of terms when working with AC function symbols. Let us denote by $\mathsf{tf}(t)$ the term $t$ in flattened representation. We now define some preliminary notions as follows.

For all $s = f(s_1,\ldots,s_n)$ with $f \in \F_{AC}$, we define:
\begin{itemize}
\item $\mathrm{EmbSmall}(s) = \{ \mathsf{tf}(f(s_1,\ldots,s_{i-1},v_j,s_{i+1},\ldots,s_n)) \mid s_i = h(v_1,\ldots,v_r) \wedge f >_\F h \wedge j \in \{1,\ldots,r\}\}$;
\item $\mathrm{BigHead}(s) = \multiset{s_i \mid i \in \{1, \ldots,n\} \wedge top(s_i) >_\F f}$;
\item $\mathrm{NoSmallHead}(s) = \multiset{s_i \mid i \in \{1, \ldots,n\} \wedge f  \not>_\F top(s_i)}$;
\item $\#(f(s_1,\ldots,s_n)) = \#_v(s_1) + \ldots +  \#_v(s_n)$ with $\#_v(x) = x$ and $\#_v(t) = 1$ if $t$ is not a variable;
\end{itemize}
where $\multiset{u_1,u_2,\ldots,u_k}$ is a multiset.

The order $>_{AC}$ is defined as: $s = f(s_1,\ldots, s_n) >_{AC} g(t_1,\ldots,t_m)$ if and only if one of the following properties holds
\begin{enumerate}
    \item $s_i \geq_{AC} t$ for some $i \in \{1, \ldots, n\}$
    \item $f >_\F g$ and $s >_{AC} t_i$ for all $i \in \{1, \ldots, m\}$
    \item $f = g \in \F_o$ and $(s_1,\ldots,s_n) >^{lex}_{AC} (t_1,\ldots,t_m)$ and $s >_{AC} t_i$ for all $i \in \{1, \ldots, m\}$
    \item $f = g \in \F_{AC}$ and $\exists s' \in \mathrm{EmbSmall}(s)$ s.t. $s' \geq_{AC} t$
    \item $f = g \in \F_{AC}$ and $\forall t' \in \mathrm{EmbSmall}(t)$, $s >_{AC} t'$ and $\mathrm{NoSmallHead}(s) \geq^{mul}_{AC} \mathrm{NoSmallHead}(t)$ and either:
    \begin{itemize}
    \item $\mathrm{BigHead}(s) >^{mul}_{AC} \mathrm{BigHead}(t)$ or 
    \item $\#(s) > \#(t)$ or
    \item $\#(s) \geq \#(t)$ or $\multiset{ s_1, \ldots, s_n} >^{mul}_{AC} \multiset{ t_1, \ldots, t_m}$
    \end{itemize}
\end{enumerate} 
with $>^{lex}_{AC}$ the lexicographic extension of $>_{AC}$ on sequences and $>^{mul}_{AC}$ is the multiset extension of $>_{AC}$ on finite multisets. 

The order $>_{AC}$ was shown in \cite{DBLP:conf/rta/Rubio99} to be a AC-compatible and AC-total reduction order. It is also AC-compatible with $\R_{\oplus}$ from \Cref{ex:XOR}. This is one of the two reduction orders we implemented in our prototype.
\end{example}

Most orders in the literature do not consider names in terms. For the need of our algorithm, we will require our reduction order to be stable by application of renamings that preserve the order. Formally, a renaming $\rho$ \emph{preserves $>$} when for all $a,b \in \dom{\rho}$, $a > b$ iff $a\rho > b\rho$. We say that $>$ is \emph{stable by renaming} when for all terms $s,t$, for all renamings $\rho$ preserving $>$, if $\names{s,t} \subseteq \dom{\rho}$ then $s > t$ iff $s\rho > t\rho$. 

\subsection{Finite variant property}

We rely on the seminal notion of finite variant property which has been initially introduced by~\cite{DBLP:conf/rta/Comon-LundhD05} adapted to our notation.
Let $E \subseteq E'$ be two equational theories. Let $>$ be an $E$-total and $E$-compatible reduction order.
Given two terms $t,t' \in \T(\F,\X \cup \N)$, $(t',\theta)$ is a \emph{$E'$-variant} of $t$ when $t\theta =_{E'} t'$. A \emph{complete set of $E'$-variants modulo $E$ of $t$} is a set $S$ of $E'$-variants of $t$ such that for all substitutions $\sigma$ closing for $t$ (i.e. $t\sigma$ is ground), there exist $(t',\theta) \in S$ and a substitution $\alpha$ such that $\minOrd[>]{E'}(t\sigma) =_{E} t'\alpha$.


\begin{definition}[\cite{DBLP:conf/rta/Comon-LundhD05}]
\label{def:finite_variant}
Given two equational theories $E \subseteq E'$, $E'$ has the \emph{finite variant property modulo $E$} w.r.t. $>$ if for all terms $t$, there exists a finite complete set of $E'$-variants modulo $E$ of $t$.
\end{definition}

Intuitively, if $S$ is the complete set of $E'$-variants modulo $E$ of $t$ then the smallest term by $>$ equal modulo $E'$ to any closing instantiation of $t$ is an instantiation (modulo $E$) of a term in $S$. To draw a parallel with the notion of most general unifiers, one could say that the terms in $S$ of $t$ are the \emph{most general smallest $E$-instantiations of $t$ by $>$}. Note that even a simple term can have many variants.

\begin{example}
Coming back to the theory $E_{\oplus}$ of XOR, the complete set of $E_{\oplus}$-variants modulo $AC$ of $x_1 + x_2$ is composed of 7 variants:
\[
\begin{array}{@{}c@{}}
x_1, \{ x_2 \mapsto 0\} \quad x_2, \{ x_1 \mapsto 0\} \quad y, \{ x_1 \mapsto y + x_2\}\\
y_1 + y_2, \{ x_1 \mapsto y_1 + z, x_2 \mapsto y_2 + z\} \quad 0, \{ x_2 \mapsto x_1\}\\
y, \{ x_2 \mapsto x_1 + y\} \quad  x_1 + x_2, \emptyset
\end{array}
\]
For the Abelian Group theory, the term $x_1 + x_2$ has 47 variants.
\end{example}


In the rest of this paper, we will say that $>$ is a \emph{$E$-strong reduction order} when it is $E$-compatible, $E$-total and stable by renaming.

\section{Rewrite theory}
\label{sec:extended-signature}

In the vein of previous
works~\cite{DBLP:journals/jlp/BlanchetAF08,DBLP:conf/rta/EscobarMS08,DBLP:conf/rta/Comon-LundhD05,DBLP:journals/entcs/EscobarMS09,cholewa2014variants},
a rewrite theory $T$ can be seen as a decomposition of the main
equational theory $E$ into a rewrite system and a smaller equational
theory. However, as previously mentioned, the decomposition of an
equation theory $E$ into a rewrite system that is AC-convergent does
not necessarily lead to an efficient unification algorithm as the
unification modulo AC is notoriously slow. 
On the other hand, the problem of matching modulo AC, that is,
checking whether there exists $\sigma$ such that $u\sigma =_{AC} v$,
can be often solved efficiently in practice~\cite{10.1007/3-540-44881-0_3}, despite still being an NP-complete problem~\cite{BENANAV1987203}. 
This distinction between unification and matching modulo AC gave rise
to one of the key component of our rewrite theories: instead of decomposing the main equation theory $E$ into one rewrite system and one small equational theory, a rewrite theory $T$ will decompose $E$ into two sets of rewrite rules $\Rn$ and $\R$ and two equational theories $\En$ and $\Ea$ such that $\En \subseteq \Ea \subseteq E$.
Intuitively, we will assume the existence of
an efficient algorithm for solving unification modulo $\Ea$
(e.g. XOR, $\mathcal{AG}$) whereas we only require an efficient
algorithm for solving matching modulo $\En$ (e.g. AC). 
Similarly,
the set of rewrite rules $\R$ will be used to compute the variants of terms modulo $\Ea$, 
whereas the set of rewrite rules $\Rn$ will be used to normalise the terms
modulo $\En$. Finally, our rewrite theory $T$ will also include an
$\En$-strong reduction order compatible with $\Rn$. This
order will serve several purposes: first, it allows us to ensure
$\En$-termination of $\Rn$ (we do not require $\En$-convergence);
second, it will be used to show that $E$ has the finite variant property modulo $\Ea$; and it will play a crucial part in the proof of correctness of our main algorithm (see \Cref{sec:generation-extended-signature}). 


\begin{definition}
\label{def:extended_signature}
A \emph{rewrite theory} $T$ is a tuple $(>,\R,\Rn,\En,\Ea)$ where:
\begin{enumerate}[label=\textbf{S\arabic*}]
    \item\label{S:subset} $\En$ and $\Ea$ are equational theories with $\En \subseteq \Ea$
    \item\label{S:order} $>$ is a $\En$-strong reduction order compatible with $\Rn$
    \item\label{S:std} $\R$ is a rewrite system such that for all $f \in \F$, denoting $\ell = f(x_1,\ldots,x_n)$, we have $(\ell \rightarrow \ell) \in \R$, and for all $\ell' \rightarrow r' \in \R$, we have $\vars{r'} \subseteq \vars{\ell'}$
\end{enumerate}
\end{definition} 

In \Cref{S:std}, the presence of rules of the form $f(x_1,\allowbreak\ldots,x_n) \rightarrow f(x_1,\ldots,x_n)$ is used in the next section to homogenise the definition of computation of most general unifiers and variants. Intuitively, when computing the variants of a term $u$, we will use the rules in $\R$ to rewrite all symbols in $u$. Hence, the above rule represents the case when the symbol $f$ is actually not rewritten.

\begin{example}
\label{ex:signature}
Coming back to the theory $E_{\oplus}$, we can define two rewrite theories $T_1 = (>_{AC},\R_1,\Rn_1,AC,AC)$ and $T_2 = (>_{AC},\allowbreak \R_2,\Rn_2,AC,E_{\oplus})$ where $\R_1$ is the rewrite system that includes the following rules:
\[
\begin{array}{c}
x_1 + 0 \rightarrow x_1 \qquad  0 + x_2 \rightarrow x_2    \qquad (y + x_2) + x_2 \rightarrow y  \\
x_1 + x_1 \rightarrow 0 \qquad  (y_1 + z) + (y_2 + z) \rightarrow y_1 + y_2  \\
x_1 + x_2 \rightarrow x_1 + x_2 \qquad x_1 + (x_1 + y) \rightarrow y \qquad 0 \rightarrow 0
\end{array}
\] 
and $\R_2 = \{ x_1 + x_2 \rightarrow x_1 + x_2; 0 \rightarrow 0 \}$. 

Notice that $\R_1$ (resp. $\R_2$) corresponds to the $E_{\oplus}$-variants modulo $AC$ (resp. $E_{\oplus}$) of $x_1 + x_2$ and $0$. Though we will formally define the notion of a rewrite theory mimicking an equational theory later in \Cref{def:mimics}, $\R_1$ and $\R_2$ containing variants will explain why both rewrite theories $T_1$ and $T_2$ mimic $E_{\oplus}$. 

Hence, for a larger equational theory $E$ that contains $E_{\oplus}$, by taking $\Ea$ to be $E_{\oplus}$ instead of $AC$ in the rewrite theory, we also reduce the number of rewrite rules in the rewrite system $\R$.
Consider, as toy-example, $E_{\mathsf{ed}} = E_{\oplus} \cup \{ \mathsf{d}(\mathsf{e}(x + k_2,k_1),k_1,k_2) = x \}$.
The term $\mathsf{d}(\mathsf{e}(x + r,k),k,r)$ has 8 $E_{\mathsf{ed}}$-variants modulo AC whereas it only has 2 $E_{\mathsf{ed}}$-variants modulo $E_{\oplus}$, corresponding to the following two rules: $\mathsf{d}(\mathsf{e}(x + k_2,k_1),k_1,k_2) \rightarrow x$ and $\mathsf{d}(x,k_1,k_2) \rightarrow \mathsf{d}(x,k_1,k_2)$.

The rewrite systems $\Rn_1$ and $\Rn_2$ in the rewrite theories $T_1$ and $T_2$ are only used for normalisation. As such, their content can vary and will mostly be useful for the generation of rewrite theories (see \Cref{sec:generation-extended-signature}). For instance, on input $E_\oplus$, our prototype will start by building a rewrite theory with $\Rn_1 = \emptyset$ but gradually augment it to reach $\Rn_1 = \{ x + 0 \rightarrow x; x + x \rightarrow 0 \}$.
\end{example}

Although the rewrite system $\Rn$ in a rewrite theory $T$ is intuitively used for normalisation, \Cref{def:extended_signature} does not impose $\En$-convergence of $\Rn$. Therefore, we rely on $>$ being a $\En$-strong reduction order compatible with $\Rn$ to define a notion of normal form with respect to a rewrite theory $T$.


\begin{definition}
\label{def:normal form}
Let $T = (>,\R,\Rn,\En,\Ea)$ be a rewrite theory and $E$ an equational theory. Let $k$ be the smallest name in $\N$ by $>$. A set of terms $\M$ is \emph{in normal form}, denoted $\nf{T}{E}{\M}$, when $\Ea \subseteq E$ and there exists an injective substitution $\sigma$ such that  $\dom{\sigma} = \vars{\M}$, $\img{\sigma} \subseteq \N$ and:
\begin{itemize}
\item $\forall a \in \img{\sigma}$, $a \neq k$ and $\forall b \in \names{\M}$, $a > b$
\item $\forall t \in \M$, $t\sigma =_\En \minOrd[>]{E}(t\sigma)$.
\end{itemize}
\end{definition}

Intuitively, a term $t$ is in normal form when there is no other term equal to $t$ modulo $E$ that is strictly smaller than $t$ by $>$. Such intuition only holds on ground terms as the order $>$ is $\En$-total. When the term $t$ contains variables, it may not be ordered with other terms. Thus, to discuss the normal form of a term with variables, we will consider all variables of $t$ as names, i.e. we will close $t$ with some substitution $\sigma$ such that $\img{\sigma}$ only contains names. However, the choice of names in the substitution $\sigma$ may impact the minimality by $>$. In particular, to distinguish variables with names in $t$, it is important for $\sigma$ to select names not already in $t$. Similarly, when we need to consider multiple terms in normal forms, e.g. in a set of terms $\M$, 
we need to select names that are not already in any of the terms in $\M$. This is the purpose of the first item of \Cref{def:normal form} which guarantees that we do not confuse the names already in $\M$ with the one used to close $\M$, i.e. in $\img{\sigma}$.


\begin{example}
Consider two symbols $f/2$ and $h/1$, and an equational theory $E$ representing that $f$ is commutative. Take $\En = \emptyset$. Consider a lexicographic path ordering such that on ground terms, the smallest terms are $a < b < h(a) < c < \ldots$ where $a,b,c$ are names. By definition, $f(x,h(a)) =_E f(h(a),x)$. With such order, both terms $f(x,h(a))$ and $f(h(a),x)$ individually are in normal form since $f(b,h(a)) < f(h(a),b)$ ($\sigma = \{ x \rightarrow b\}$) and $f(h(a),c) < f(c,h(a))$ ($\sigma = \{ x \rightarrow c\}$). In other words, there is a way to close each term $t$ such that the resulting term $t\sigma$ is the smallest by $<$ amongst the terms equal modulo $E$, that is $t\sigma = \minOrd[>]{E}(t\sigma)$. However, the set $\{ f(x,h(a)), h(b) \}$ is not in normal form whereas $\{ f(h(a),x), h(b) \}$ is in normal form since we cannot instantiate $x$ by $b$ any more as it already occurs in the sets.
\end{example}


Manipulating the injective substitution $\sigma$ and $\minOrd[>]{E}$ in the definition can be quite cumbersome. Hence, we state below some interesting properties that show how we can manipulate sets of terms that are in normal form, without relying on the renaming nor $\minOrd[>]{E}$. 
The proof of \Cref{lem:norm} can be found in \Cref{sec:app-extended-signature}.


\begin{restatable}{lemma}{lemnormalisation}
\label{lem:norm}
Let $T = (>,\R,\Rn,\En,\Ea)$ be a rewrite theory and $E$ be
a non-trivial equational theory. Assume that $\Ea \subseteq E$, and
that for all terms $t$ and $s$, $t \rwstep[\En]{\Rn} s$ implies $t =_E s$.
Let $\M$ be a set of terms such that $\nf{T}{E}{\M}$. Then the following properties hold:
\begin{enumerate}[label=\textbf{N\arabic*}]
\item for all $t \in \st{\M}$, $t$ is irreducible by $\rwstep[\En]{\Rn}$. \label{lem:norm-irreducible}
\item for all $s,t \in \st{\M}$, if $s =_E t$ then $s =_\En t$. \label{lem:norm-eq}
\item for all $x \in \X$, $\nf{T}{E}{\M \cup \{x\}}$. \label{lem:norm-var}
\item for all terms $t$, there exists a term $s$ such that $t =_E s$ and $\nf{T}{E}{\M \cup \{s\}}$. \label{lem:norm-addition}
\end{enumerate}
\end{restatable}

\Cref{lem:norm-irreducible} is particularly useful when the order $>$
is not (easily) computable as it provides a sound and easy way for checking that $t$ is not in normal form. Additionally, since $>$ is compatible with $\Rn$, $t \rwsteps[\En]{\Rn} t'$ implies $t > t'$. Hence, an irreducible term obtained by rewriting $t$ is a potential candidate for being minimal by $>$ amongst the equivalence class of $t$ modulo $E$.
We can now define when a rewrite theory $T$ mimics an equational theory $E$.


\begin{definition}
\label{def:mimics}
A rewrite theory $T = (>,\R,\Rn,\En,\Ea)$ \emph{mimics an equational theory} $E$ when:
\begin{enumerate}[label=\textbf{M\arabic*}]
\item \label{M:EaE}$\Ea \subseteq E$ and if $t \rwstep[\En]{\Rn} t'$ then $t =_E t'$.
  
\item \label{M:eqcomplete} If $f(t_1, \ldots, t_n) \rightarrow t$ is in $\R$ then $f(t_1, \ldots, t_n) =_E t$.
    
\item \label{M:ind1} If $f(t_1, \ldots, t_n) =_E t$ and $\nf{T}{E}{\{ t_1,\ldots,t_n, t \}}$ then there exist $\sigma$ and $f(s_1,\ldots, s_n) \rightarrow s$ in $\R$ such that, $t =_{\Ea} s\sigma$, and for all $i\in \{ 1, \ldots, n \}$, $t_i =_{\Ea} s_i\sigma$ .

\end{enumerate}
In \Cref{M:ind1}, the variables in the rewrite rule $f(s_1,\ldots, s_n) \rightarrow s$ are assumed to be freshly renamed, to ensure that $\vars{s_1,\ldots,s_n,s} \cap \vars{t_1,\ldots,t_n,t} = \emptyset$.
\end{definition}


\begin{example}
$T_1$ and $T_2$ from \Cref{ex:signature} both mimic $E_{\oplus}$. The rewrite theory $T_2$ trivially mimics $E_{\oplus}$, but is useless since the equational theory $\Ea$ in $T_2$ is also $E_{\oplus}$. The rewrite theory $T_{\mathsf{ed}} = (>_{AC},\R,\Rn,\allowbreak AC, E_\oplus)$ mimics $E_{\mathsf{ed}}$ with $\emptyset \subseteq \Rn \subseteq \R_{\oplus}$ and $\R$ containing the rules:
\[
\begin{array}{c}
x + y \rightarrow x + y \qquad  0 \rightarrow 0 \qquad \mathsf{e}(x,y) \rightarrow \mathsf{e}(x,y) \\
\mathsf{d}(\mathsf{e}(x + k_2,k_1),k_1,k_2) \rightarrow x \qquad \mathsf{d}(x,k_1,k_2) \rightarrow \mathsf{d}(x,k_1,k_2)
\end{array}
\] 
\end{example}


\section{Complete set of \texorpdfstring{$E$}{E}-variants}
\label{sec:variant}

To show the relation between rewrite theories mimicking equational theories and the finite variant property, we first need to define the notion of \emph{to-evaluate symbols} and \emph{to-evaluate terms}. For each function symbol $f \in \F$, we associate a function symbol $\symbeval{f}$ that will correspond to a symbol that needs to be evaluated. We define $\symbeval{\F} = \{ \symbeval{f} \mid f \in \F\}$. The symbols in $\symbeval{\F}$ are called \emph{to-evaluate symbols} (TE symbols for short). 


\begin{definition}
A \emph{to-evaluate term} (TE-term for short) is a term $t \in \T(\symbeval{\F} \cup \F,\X \cup \N)$ such that either $t \in \T(\F,\X \cup \N)$, or $t = \symbeval{f}(t_1,\ldots,t_n)$, and $t_1,\ldots,t_n$ are TE-terms.
\end{definition}

Intuitively, in a to-evaluate term, no to-evaluate symbol in $\overline{\F}$ can occur "below" a standard symbol from $\F$; or in other words, no to-evaluate symbol can occur in a subterm rooted by a symbol from $\F$.
Given a term $t$, we denote by $\symbeval{t}$ the TE-term obtained from $t$ in which all function symbols $f$ have been replaced by $\symbeval{f}$.

\begin{example}
\label{ex:to-evaluate}
Consider $T = (>,\R,\emptyset,\emptyset,\emptyset)$ a rewrite theory  mimicking standard randomized encryption $E = \{ \mathrm{dec}(\allowbreak\mathrm{enc}(x,r,k),k) = x \}$. The set of rewrite rules $\R$ can be composed of the following rules:
\[
\begin{array}{c}
\mathrm{dec}(\mathrm{enc}(x,r,k),k) \rightarrow x \qquad \mathrm{dec}(x,k) \rightarrow \mathrm{dec}(x,k)\\
\mathrm{enc}(x,r,y) \rightarrow \mathrm{enc}(x,r,y)
\end{array}
\]
We have that $\symbeval{\mathrm{dec}}(\mathrm{enc}(m,r,k),k')$ is a TE-term but $\mathrm{enc}(\symbeval{\mathrm{dec}}(x,k'),r,k')$ is not a TE-term. 
\end{example}

We now  define the \emph{close evaluation of TE-terms}, from which we will derive the procedure to check the equality modulo $E$.


\begin{definition}
\label{def:evaluation}
Let $T = (>,\R,\Rn,\En,\Ea)$ be a rewrite theory. We define the non-deterministic \emph{close evaluation of a TE-term} $t$ into a term $s$, denoted $t \Eval{T} s$, as follows:
\[
\begin{array}{l}
t \Eval{T} t \quad \text{if }t \in \T(\F,\X)\\[1.5mm]
\symbeval{f}(t_1, \ldots, t_n) \Eval{T} r\sigma \qquad\hfill \text{if }f(\ell_1, \ldots, \ell_n) \rightarrow r \in \R\text{, and}\\
\quad \hfill\forall i \in \{1, \ldots,n\}, t_i \Eval{T} s_i \text{ and } s_i =_\Ea \ell_i\sigma
\end{array}
\]
\end{definition}

The evaluation on TE-terms intuitively evaluates the function symbols from the bottom up.

\begin{example}
Coming back to the rewrite theory $T$ of \Cref{ex:to-evaluate}, denoting $t = \symbeval{\mathrm{dec}}(\mathrm{enc}(m,r,k'),k')$ we have $t \Eval{T} m$ and $t \Eval{T} \mathrm{dec}(\mathrm{enc}(m,r,k'),k')$. The first evaluation corresponds to the evaluation of $\symbeval{\mathrm{dec}}$ with the rewrite rule $\mathrm{dec}(\mathrm{enc}(x,r,k),k) \rightarrow x$ whereas the second evaluation corresponds to the evaluation of $\symbeval{\mathrm{dec}}$ with the rewrite rule $\mathrm{dec}(x,k) \rightarrow \mathrm{dec}(x,k)$.
\end{example}

The evaluation on TR-terms allows us to have a simple procedure to test the equality modulo $E$, as shown in the following result 
(proof in \Cref{sec:app-proof variant}).


\begin{restatable}[Equality modulo $E$]{theorem}{thequalitymodulo}
\label{th:equality modulo}
Let $T = (>,\R,\Rn,\En,\allowbreak\Ea)$ be a rewrite theory and $E$ a non-trivial equational theory. If $T$ mimics $E$ then for all $t,s \in \T(\F,\X \cup \N)$,
$t =_E s$ if and only if there exist $t',s'$ such that $\symbeval{t} \Eval{T} t'$, $\symbeval{s} \Eval{T} s'$ and $t' =_\Ea s'$.
\end{restatable}

We can now define the open evaluation of a sequence of TE-terms $L$ as a relation $L \OEval{T} (L',\sigma)$ where $\sigma$ are the instantiations of the variables of $L$ obtained during the evaluation of the TE-symbols in $L$ and where $L'$ is the result of such evaluation.


\begin{definition}
\label{def:open evaluation}
Let $T = (>,\R,\Rn,\En,\Ea)$ be a rewrite theory. We define the open evaluation on a sequence of TE-terms $L$, denoted $L \OEval{T} (L',\sigma)$, as follows:
\[
\begin{minipage}{\columnwidth}
\begin{tabbing}
$[] \OEval{T} ([],\emptyset)$\\[1.5mm]
$[t] \OEval{T} ([t],\emptyset)$ \quad if $t \in \T(\F,\X \cup \N)$\\[1.5mm]
$[\symbeval{h}(t_1, \ldots, t_n)] \OEval{T} ([u\sigma_u], \sigma'\sigma_u)$
\\
\quad if $[t_1; \ldots; t_n] \OEval{T} ([s_1; \ldots; s_n],\sigma')$, $h(u_1, \ldots, u_n) \rightarrow u \in \R$, \\
\quad and $\sigma_u \in \mgu{\Ea}{(s_1,\ldots,s_n),(u_1,\ldots,u_n)}$
\\[1.5mm]
$t \cdot L \OEval{T} (s\sigma' \cdot L',  \sigma\sigma')$
\quad if $[t] \OEval{T} ([s],\sigma)$ and $L\sigma \OEval{T} (L',\sigma')$
\end{tabbing}
\end{minipage}
\]
\end{definition}

\begin{example}
Coming back to the rewrite theory $T$ of \Cref{ex:to-evaluate}, we consider the $t = \symbeval{\mathrm{dec}}(x,\mathrm{dec}(y,b))$. Notice here that only the outer symbol $\mathrm{dec}$ is to be evaluated. In particular $[\mathrm{dec}(y,b)] \OEval{T} ([\mathrm{dec}(y,b)],\emptyset)$. Therefore, we have:
$[t] \OEval{T} ([z],\{ x \rightarrow \mathrm{enc}(z,r,\mathrm{dec}(y,b))\})$ and $[t] \OEval{T} ([\mathrm{dec}(x,\mathrm{dec}(y,b))],\emptyset)$.
\end{example}

The open evaluation intuitively computes the variants of terms contained in the sequence of terms $L$ as shown in the next result 
(proof in \Cref{sec:app-proof variant}).


\begin{restatable}[Complete sets of $E$-variants]{theorem}{thvariant}
\label{th:variant}
Let $T = (>,\allowbreak\R,\Rn,\En,\Ea)$ be a rewrite theory and $E$ a non-trivial equational theory. If $T$ mimics $E$ then for all terms $t$, the set $\{ (t',\alpha) \mid [\symbeval{t}] \OEval{T} ([t'],\alpha)\}$ is a complete set of $E$-variants modulo $\Ea$.
\end{restatable}


\begin{corollary}[Finite variant property]
\label{cor:finite variant}
Let $T = (>,\R,\Rn,\allowbreak\En,\Ea)$ be a rewrite theory and $E$ a non-trivial equational theory. If $T$ mimics $E$ then $E$ has the finite variant property modulo $\Ea$.
\end{corollary}

As the open evaluation allows to compute the finite complete set of variants, it also allows to compute the complete set of most general $E$-unifiers, as shown below
 (proof in \Cref{sec:app-proof variant}).


\begin{restatable}[Most general $E$-unifiers]{theorem}{thmostgeneralunifiers}
Let $T = (>, \R,\Rn,\allowbreak\En,\Ea)$ be a rewrite theory and $E$ a non-trivial equational theory. If $T$ mimics $E$ then the set $\{ \alpha\sigma_u \mid [\symbeval{t},\symbeval{s}] \OEval{T} ([t',s'],\alpha) \wedge \sigma_u \in \mgu{\Ea}{t',\allowbreak s'}\}$ is a complete set of most general $E$-unifiers of $t,s$.
\end{restatable}

\Cref{cor:finite variant} showed that a rewrite theory mimicking an equational theory implies that it has the finite variant property. In fact we can show that finite variant property and mimicking of equational theories are equivalent definitions when $\Ea = \En$, provided that the set of function symbols $\F$ contain some free symbols. In essence, we provide a more \emph{practical} characterisation of the finite variant property. 


\begin{restatable}{theorem}{thequivalentdeffinitevariant}
\label{th:FVP-implies-mimics}
Let $\En \subseteq \Ea \subseteq E$ be three non-trivial equational theories. Let $>$ be a $\En$-strong reduction ordering. Assume that there exist a binary function $\bin \in \F$ and a constant $\amin \in \F$ such that $\bin$ does not occur in $E$ and for all ground terms $t$, $\amin \not> t$.

If $E$ has the finite variant property modulo $\Ea$ and 
\begin{itemize}
\item either $\Ea = \En$ 
\item or for all $f/n \in \F$, the complete set $S$ of $E$-variants modulo $\Ea$ of $\bin(f(x_1,\ldots,x_n),\bin(x_1,\bin(x_2,\ldots \bin(x_{n-1},x_n))))
$ satisfies for all $(t,\theta) \in S$, $\vars{t_{|1}} \subseteq
\vars{t_{|2}}$ (where $t_{|i}$ denotes the $i^{\text{th}}$ subterm of $t$)
\end{itemize}
then there exists $T = (>,\R,\Rn,\En,\Ea)$ that mimics
$E$. 
\end{restatable}


\section{Generating rewrite theories}
\label{sec:generation-extended-signature}

We describe in this section a semi-decision procedure to generate rewrite theories that mimic an equational theory. Contrary to previous sections, where we only assumed the existence of an algorithm to compute the complete set of most general $\Ea$-unifiers, we will also assume in this section that we have at our disposal an algorithm for computing the most general $\En$-unifiers. Our approach takes advantage of the following remarks. 

First, by relying on most general $\En$-unifiers only in the generation of rewrite theories, we greatly limit our reliance on it. 
Focusing on the automated verification of protocols use case, this means that we only need to run this semi-decision procedure once for each particular equational theory. In particular, this means that if during the protocol verification the user wants to make several small modifications to the protocol steps, but not to the equational theory, then the procedure has not to be run again. Even different protocols using the same equational theory would require the procedure to be run only once.
This observation applies to several protocol verification tools such as ProVerif~\cite{DBLP:conf/sp/BlanchetCC22},  Tamarin~\cite{DBLP:conf/cav/MeierSCB13} and Maude-NPA~\cite{DBLP:conf/birthday/EscobarMMS15}.

Also, although the generic algorithm solely relies on $\En$-most general unifiers and will typically terminate only if $E$ has the finite variant property modulo $\En$, we will show that our optimisations allow us to reduce the number of variants when focusing on FVP modulo $\Ea$.

Last, although our optimised algorithm can take as input a rewrite system $\Rn$, the algorithm does not really need that particular knowledge. Indeed, a user could start with an empty $\Rn$ and the algorithm will by itself augment $\Rn$ with appropriate rewrite rules. 
This feature is particularly useful in the context of automated verification of protocols, as users do not necessarily have the knowledge to create a decomposition of $E$ by themselves.


\subsection{Overlapping rewrite rules}

The core of the procedure consists in taking two rewrite rules $\ell_1 \rightarrow r_1$ and $\ell_2 \rightarrow r_2$ where $r_1$ and $\ell_2$ overlap and producing a new rewrite rule. 


\begin{definition}
\label{def:overlap}
Let $\En$ be an equational theory. We say that \emph{the rewrite rule $\ell_1 \rightarrow r_1$ is $\En$-overlapping on a position $p$ with the rewrite rule $\ell_2 \rightarrow r_2$} if $p \in \Pos{r_1}$, ${r_1}_{|_p}$ is not a variable and there exists $\sigma \in \mgu{\En}{{r_1}_{|_p},\ell_2}$.

The set of rewrite rules merging the two $\En$-overlapping rules on $p$, denoted $(\overlapset{\ell_1 \rightarrow r_1}{p,\En}{\ell_2 \rightarrow r_2})$, is defined as:
\[
\{\ell_1\sigma \rightarrow r_1\sigma[r_2\sigma]_p \mid \sigma \in \mgu{\En}{{r_1}_{|_p},\ell_2}\}
\]
When $p \not\in \Pos{r_1}$ or ${r_1}_{|_p} \in \X$, $(\overlapset{\ell_1 \rightarrow r_1}{p,\En}{\ell_2 \rightarrow r_2}) = \emptyset$.
\end{definition}

Overlapping rewrite rules are closely related to the classical notion of critical pairs used for example in the Knuth-Bendix completion algorithm. In particular, the critical pairs of $\ell_1 \rightarrow r_1$ and $\ell_2 \rightarrow r_2$ are the pairs $(s,t)$ such that $s \rightarrow t \in (\overlapset{r_1 \rightarrow \ell_1}{p,\En}{\ell_2 \rightarrow r_2})$. Notice that here, the overlap is between $\ell_1$ and $\ell_2$.
We denote by $(\overlapseteq{\ell_1 \rightarrow r_1}{p,\En}{\ell_2 \rightarrow r_2})$ the set of rules from $(\overlapset{\ell_1 \rightarrow r_1}{p,\En}{\ell_2 \rightarrow r_2})$ and their opposite orientation rules.


\subsection{The procedure}


\begin{figure}[ht]
\[
\begin{array}{l}
(\RNormR)\\
\R \cup \{ s \rightarrow t \} \normstep[\Rn,\En] \R \cup \{ s \rightarrow t'\} 
\hfill\qquad \text{if } t =_\En \circ \rwstep{\Rn} t'\\[2mm]
(\RNormL)\\
\R \cup \{ s \rightarrow t \} \normstep[\Rn,\En] \R \cup \{ s[s']_i \rightarrow t, t \rightarrow s[s']_i\} \\
\quad\hfill\text{if }s = f(s_1,\ldots,s_n),  i \in \{1,\ldots,n\}, s_i =_\En \circ \rwstep{\Rn} s'\\[2mm]
(\REq)\\
\R \cup \{ s \rightarrow t \} \normstep[\Rn,\En] \R \hfill \text{if }s =_\En t\\[2mm]
(\RSubsume)\\
\R \cup \{ u \rightarrow t \} \normstep[\Rn,\En] \R\\
\hfill \text{if }(u' \rightarrow t') \in \R, C[u'\sigma] \eqES{\En} u, C[t'\sigma] \eqE{\En} t,\\
\hfill t' \not\rwstepC[\En]{\Rn} \text{ and } \forall p >\varepsilon.\ u'_{|p} \not\rwstepC[\En]{\Rn}
\end{array}
\]
\caption{Normalisation rules}
\label{fig:normalisation rules}
\end{figure}

The main procedure will consist in generating the rewrite systems $\R$ and $\Rn$ that will in the end be part of the rewrite theory $(>,\R,\Rn,\En,\Ea)$. Since $\Rn$ in combination with $\En$ is used to normalise the rules, the procedure naturally relies on a subroutine that will just normalise the rewrite rules in the rewrite system $\R$. We define the subroutine by a set of normalisation rules on $\R$ displayed in \Cref{fig:normalisation rules}. The rule \RNormR takes a rule $s \rightarrow t$ from $\R$ and replaces $t$ by a term $t'$ that $t$ rewrites into,  i.e. $t =_\En \circ \rwstep{\Rn} t'$. There may be different possible $t'$, specially because $\Rn$ is not necessarily $\En$-convergent. However, it suffices to take only one of the terms $t$ rewrites into. 
The rule \RNormL similarly normalises one of the arguments of the left-hand side. However, it also produces the rule where we swap the two sides. 
Intuitively, when transforming a rule $s \rightarrow t$, we need to show that if an instance of the rule is \emph{well-ordered}, that is when $s\sigma > t\sigma$ then one of the instantiated produced rules should also be well-ordered. In the case of \RNormR, we know that $t > t'$ as $>$ is $\En$-compatible with $\Rn$. Thus, $t\sigma>t'\sigma$ which implies $s\sigma > t'\sigma$. Hence, we do not need to consider the rule $t' \rightarrow s$. In the case of \RNormL however, such reasoning does not hold as we might have $t\sigma > s[s']_i\sigma$. Therefore, we output both rules $s[s']_i \rightarrow t$ and $t \rightarrow s[s']_i$.


Note that in~\cite{DBLP:journals/jlp/BlanchetAF08}, when the equational theory can be modelled as a $\emptyset$-convergent rewrite system, their procedure has a similar normalisation rule but does not swap the rules as is done in \RNormL.
Hence, when it is known that $\Rn$ is $\En$-convergent, one could argue that we should add a specific rule where the swap does not occur. However, it is unnecessary thanks to the rule \REq. Assume for simplicity that $t$ is already normalised by \RNormR, once the rule $t \rightarrow s[s']_i$ is added to $\R$, it will be simplified once again by successive applications of the rule \RNormR. As $\Rn$ is $\En$-convergent, it will produce a rule $t \rightarrow t'$ with $t =_\En t'$ which will disappear by \REq.

Finally, the rule \RSubsume removes a rule that is already subsumed by another rule in $\R$. In the definition of the rule, $\eqES{\En}$ denotes the \emph{strict equality modulo $\En$}, that is $t \eqES{\En} s$ iff $t = f(t_1,\ldots,t_n)$, $s = f(s_1,\ldots,s_n)$ and for all $i \in \{1, \ldots,n\}$, $t_i \eqE{\En} s_i$. Intuitively, this equality modulo $\En$ does not affect the root symbol of $t$ and $s$. Additionally, the rule \RSubsume requires that the rule $s' \rightarrow t'$ that subsumes $s \rightarrow t$ should be irreducible by $\rwstepC[\En]{\Rn}$ (strictly for $s'$).

\begin{definition}
A \emph{normalisation step}, denoted $\R \normstep[\Rn,\En] \R'$, is defined when $\R'$ is the result of the application of \RNormR, \RNormL, \REq or \RSubsume on $\R$. The set $\R$ is said to be normalised when $\R \not\normstep[\Rn,\En]$. We define the function \Normalize{$\R,\Rn,\En$} that computes and returns a normalised $\R'$ such that $\R \normstep[\Rn,\En]^* \R'$.
\end{definition}

\begin{lemma}
Let $\Rn$ be a set of rewrite rules. Let $\En$ be an equational theory. Let $>$ be a $\En$-strong reduction ordering compatible with $\Rn$. For all sets of rewrite rules $\R$, the function \Normalize{$\R,\Rn,\En$} terminates.
\end{lemma}

\begin{proof}
As $>$ is $\En$-compatible with $\Rn$, we have that for all terms $t,t'$, if $t =_\En \circ \rwstep{\Rn} t'$ then $t > t'$. Hence, when the rules \RNormL and \RNormR replace a rule $s \rightarrow t$, they do it by replacing one of the term (either $s$ or $t$) by a new term strictly smaller by $>$. As the order $>$ is well-founded, we cannot have an infinite sequence of rules generated by normalisation rules, and so \Normalize{$\R,\Rn,\En$} terminates.
\end{proof}

\begin{algorithm}[ht]
\Fn{\GenExtended{$E',\Rn,\En$}}{
    $\R := \Rn \cup \{ s \rightarrow t, t \rightarrow s \mid (s = t) \in E' \}$\;\label{alg:initial value}
    $\R :=$  \Normalize{$\R,\Rn,\En$}\;\label{alg:first normalisation}

    \Repeat{$\R = \R_0$}{\label{alg:loop}
        $\R' := \emptyset$\;
        $\R_0 := \R$\;
        \ForAll{$(\ell_1 \rightarrow r_1) \in \R \cup \En, (\ell_2 \rightarrow r_2) \in \R$, position $p$}{\label{alg:generation-forall}
            $\R' \mathrel{{:}{=}} \R' \cup (\overlapseteq{r_1 \rightarrow \ell_1}{p,\En}{\ell_2 \rightarrow r_2})$\;\label{alg:generation-left-right}
            $\R' \mathrel{{:}{=}} \R' \cup (\overlapseteq{\ell_1 \rightarrow r_1}{p,\En}{\ell_2 \rightarrow r_2})$\;\label{alg:generation-right-right}
            $\R' \mathrel{{:}{=}} \R' \cup (r_1 \rightarrow \ell_1)$\;\label{alg:generation-inverse}
        }
        $\R :=$  \Normalize{$\R'\cup \R,\Rn,\En$}
    }
    $\R \mathrel{{=}} \R \cup \{ f(x_1,\ldots,x_n) \rightarrow f(x_1,\ldots,x_n) \mid f/n \in \F\}$\;\label{alg:final_modification}
    \KwRet{$\R$}
}
\caption{Generic generation of rewrite theories}
\label{alg:generation}
\end{algorithm}

We can now define the main procedure displayed in \Cref{alg:generation}.
The set of rewrite rules $\R$ gathers the rules generated throughout the algorithm. It is initialised on \Cref{alg:initial value} with all the rules in $E'$, in both orientations, as well as all the rules in $\Rn$. Note that we do not consider both orientations for the rules in $\Rn$. 
Intuitively, in the proof (see \Cref{sec:overview} for an overview), we will only consider applications of rewrite rules that follow the ordering $>$, i.e. if $t \rwstep[\En]{\R} s$ by some rule $\ell \rightarrow r$ with substitution $\sigma$ then $\ell\sigma > r\sigma$ or $\ell\sigma =_\En r\sigma$. As such, since any rule $(\ell \rightarrow r) \in \Rn$ satisfies $\ell > r$, we will never consider an application of the rule $r \rightarrow \ell$. 

After a first normalisation on \Cref{alg:first normalisation}, the algorithm enters the main loop which merges $\En$-overlapping rewrite rules from $\R$, and then normalises these newly generated rewrite rules with the current rules in $\R$. The process repeats until we reach a fixpoint on $\R$. On \Cref{alg:generation-forall}, for sake of readability, we write $\R \cup \En$ for the set $\R \cup \{ \ell \rightarrow r, r \rightarrow \ell \mid (\ell = r) \in \En\}$. 
When looking at the two rewrite rules $\ell_1 \rightarrow r_1$ and $\ell_2 \rightarrow r_2$, we consider the two cases, that are when $\ell_1$ overlaps with $\ell_2$ (\Cref{alg:generation-left-right}) and when $r_1$ overlaps with $\ell_2$ (\Cref{alg:generation-right-right}). 

Notice that \Cref{alg:generation-inverse} also adds all the rules in $\R$ with opposite orientation. Indeed, similarly to the initialisation, the algorithm aims to maintain that rules with both orientations should occur in $\R$. But this may be disrupted by normalising with a non-$\En$-convergent $\Rn$. Hence, \Cref{alg:generation-inverse} ensures the invariant.
Finally, upon exiting the loop, the algorithm adds to $\R$ on \Cref{alg:final_modification} the rule $f(x_1,\ldots,x_n) \rightarrow f(x_1,\ldots,x_n)$ to satisfy \Cref{S:std} of \Cref{def:extended_signature}.

Before stating the theorem indicating the correctness of \Cref{alg:generation}, we need to introduce a final assumption on $\En$: we will require \emph{$\En$ to have finite equivalence classes}, that is for all terms $t \in \T(\F,\X \cup \N)$, the set $\{ t' \mid t =_\En t'\}$ is finite. This is a strong assumption that is satisfied, amongst others, by  AC, C (commutative) and A (associative) equational theories, hence allowing us to cover all the relevant equational theories. For instance, in our implemented prototype (see \Cref{sec:discussion}), we have focused on $\En$ being $AC$. Nevertheless, showing the correctness of the algorithm without this assumption remains open.


\begin{restatable}{theorem}{thgenerationsignature}
\label{th:generation of rewrite theory}
Let $\En, E'$ be two equational theories such that $\En$ has finite equivalence classes. Let $\Rn$ be a set of rewrite rules. Let $>$ be a $\En$-strong reduction ordering compatible with $\Rn$. Let $E = E' \cup \Rn^= \cup \En$.

If $E$ is not trivial and \GenExtended{$E',\Rn,\En$} terminates and returns $\R$ such that for all $(\ell \rightarrow r) \in \R$, $\vars{r} \subseteq \vars{l}$, then $(>,\R,\Rn,\En,\En)$ is a rewrite theory that mimics $E$.
\end{restatable}

Remark that when $\Rn = \emptyset$, the normalisation rules do not affect the rewrite system $\R$ given as input, that is $\Normalize{$\R,\emptyset,\En$} = \R$. Hence, for the algorithm to effectively work, it is preferable to provide a rewrite system $\Rn$ as large as possible (e.g. one that is $\En$-convergent with $E$). However, to avoid for users the requirement to manually provide this rewrite system, we present in \Cref{sec:optimisations} some optimisations that allow the prototype to start with $\Rn = \emptyset$ and to gradually augment it during the execution of the procedure. 


\medskip

\subsubsection*{Comparison with the procedures in \cite{DBLP:journals/jlp/BlanchetAF08}} 
Our algorithm is a direct generalization of the two algorithms presented in \cite{DBLP:journals/jlp/BlanchetAF08}: when $E$ is oriented as a convergent rewrite system $\Rn$, Algorithm~1 of~\cite{DBLP:journals/jlp/BlanchetAF08} actually computes \GenExtended{$\emptyset,\Rn,\emptyset$}.
When $E$ is a linear equational theory, Algorithm~2 of~\cite{DBLP:journals/jlp/BlanchetAF08} actually computes \GenExtended{$E,\emptyset,\emptyset$}. A direct consequence of \Cref{th:generation of rewrite theory} is that Algorithm~2 presented in~\cite{DBLP:journals/jlp/BlanchetAF08} for linear equational theories was in fact sound for any equational theory.


\subsection{Overview of the proof of Theorem~\ref{th:generation of rewrite theory}}
\label{sec:overview}


The complete proof can be found in \Cref{sec:app-proof-main-theorem}.
We place ourselves within the hypotheses of \Cref{th:generation of rewrite theory}, which are: $\Ea = \En$, and $>$ is a $\En$-strong reduction order compatible with $\Rn$, and \GenExtended{$E',\Rn,\En$} terminates and returns $\R$ such that for all $(\ell \rightarrow r) \in \R$, $\vars{r} \subseteq \vars{\ell}$. Let us denote $E = E' \cup \Rn^= \cup \En$ and $T = (>,\R,\Rn,\En,\En)$. 

\subsubsection{\texorpdfstring{$T$}{T} is a rewrite theory}

Showing that $T$ is a rewrite theory is a simple matter as \Cref{S:subset,S:order} are given as assumptions. Moreover, \Cref{S:std} is guaranteed by \Cref{alg:final_modification} of \Cref{alg:generation} and by our assumption that for all $(\ell \rightarrow r) \in \R$, $\vars{r} \subseteq \vars{\ell}$.

\subsubsection{Towards \texorpdfstring{$T$}{T} mimics \texorpdfstring{$E$}{E}} Amongst the three properties required to show that $T$ mimics $E$, only the last one, i.e. \Cref{M:ind1}, is difficult. \Cref{M:EaE} is directly obtained, since $\En = \Ea$ and $\En \cup \Rn^= \cup E' = E$. The proof of \Cref{M:eqcomplete} is mostly given by the following lemma.


\begin{lemma}
\label{lem:overlapset_complete}
Let $\ell_1 \rightarrow r_1$ and $\ell_2 \rightarrow r_2$ be two rewrite rules such that $\ell_1 =_E r_1$ and $\ell_2 =_E r_2$. For all positions $p$, for all $\ell_3 \rightarrow r_3 \in (\overlapset{\ell_1 \rightarrow r_1}{p,\En}{\ell_2 \rightarrow r_2})$, $\ell_3 =_E r_3$.
\end{lemma}

\begin{proof}
If $\ell_3 \rightarrow r_3 \in (\overlapset{\ell_1 \rightarrow r_1}{p,\En}{\ell_2 \rightarrow r_2})$ then $p \in \Pos{r_1}$ and there exists $\sigma \in \mgu{\En}{{r_1}_{|p},\ell_2}$ such that $\ell_3 = \ell_1\sigma$ and $r_3 = r_1\sigma[r_2\sigma]_p$. $\ell_2 =_E r_2$ implies $\ell_2\sigma =_E r_2\sigma$. Similarly, $\ell_1\sigma =_E r_1\sigma$. Since $\sigma$ is a $\En$-unifier of ${r_1}_{|p}$ and $\ell_2$, and since $\En \subseteq E$, we have ${r_1}_{|p}\sigma =_E \ell_2\sigma =_E r_2\sigma$. This implies that $r_1 \sigma = r_1\sigma[{r_1}_{|p}\sigma]_p =_E r_1\sigma[r_2\sigma]_p$.
Hence $\ell_3 = \ell_1\sigma =_E r_1\sigma =_E r_1\sigma[r_2\sigma]_p = r_3$.
\end{proof}

Notice that all rules $\ell \rightarrow r$ in the initial value $\R$ on \Cref{alg:initial value} of \Cref{alg:generation} satisfy $\ell =_E r$. By applying \Cref{lem:overlapset_complete} and noticing that the normalisation rules \RNormL and \RNormR preserve this invariant, since $\Rn^= \cup \En \subseteq E$, we obtain that $\R$ satisfies \Cref{M:eqcomplete}.

\subsubsection{A mountainous landscape of equality modulo $E$} 
\label{sec:mountainous landscape}

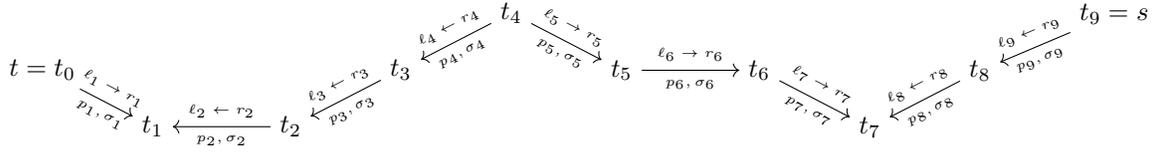
\begin{figure*}
\begin{center}
    \begin{tikzpicture}[
        term/.style={anchor=mid,fill=white}
        ]
        \def\length{1cm}

        \node[term] (T0) {$t = t_0$};
        \node[term,below right=0.5*\length and 1.2*\length of T0.mid] (T1) {$t_1$};
        \node[term,right=1.3cm of T1] (T2) {$t_2$};
        \node[term,above right=0.5*\length and 1.2*\length of T2.mid] (T3) {$t_3$};
        \node[term,above right=0.5*\length and 1.2*\length of T3.mid] (T4) {$t_4$};
        \node[term,below right=0.5*\length and 1.2*\length of T4.mid] (T5) {$t_5$};
        \node[term,right=1.3cm of T5] (T6) {$t_6$};
        \node[term,below right=0.5*\length and 1.2*\length of T6.mid] (T7) {$t_7$};
        \node[term,above right=0.5*\length and 1.2*\length of T7.mid] (T8) {$t_8$};
        \node[term,above right=0.5*\length and 1.2*\length of T8.mid] (T9) {$t_9 = s$};

        \tikzrightrwstep{T0}{p_1}{\sigma_1}{\ell_1}{r_1}{T1}
        \tikzleftrwstep{T1}{p_2}{\sigma_2}{r_2}{\ell_2}{T2}
        \tikzleftrwstep{T2}{p_3}{\sigma_3}{r_3}{\ell_3}{T3}
        \tikzleftrwstep{T3}{p_4}{\sigma_4}{r_4}{\ell_4}{T4}
        \tikzrightrwstep{T4}{p_5}{\sigma_5}{\ell_5}{r_5}{T5}
        \tikzrightrwstep{T5}{p_6}{\sigma_6}{\ell_6}{r_6}{T6}
        \tikzrightrwstep{T6}{p_7}{\sigma_7}{\ell_7}{r_7}{T7}
        \tikzleftrwstep{T7}{p_8}{\sigma_8}{r_8}{\ell_8}{T8}
        \tikzleftrwstep{T8}{p_9}{\sigma_9}{r_9}{\ell_9}{T9}
    \end{tikzpicture}
\end{center}
\caption{Mountainous landscape of $t =_E s$}
\label{fig:mountain}
\end{figure*}

Consider two terms $t$ and $s$ such that $t =_E s$. The definition of $=_E$ is given by being the least congruence such that $u\sigma =_E v\sigma$ for all equations $u = v \in E$. Another way of viewing this definition is that there exists a finite sequence $t = t_0 \rwstep{\R_E} t_1 \rwstep{\R_E} \ldots\rwstep{\R_E} t_{n-1} \rwstep{\R_E} t_n = s$ where $\R_E$ are the rules $\ell \rightarrow r$ such that $(\ell = r)$ or $(r = \ell)$ is in $E$. 

Let us assume for the moment that all $t_0, \ldots, t_n$ are ground. 
If $p_i$, $\sigma_i$ and $\ell_i \rightarrow r_i$ are respectively the position, the substitution and the rewrite rule used in the rewrite step $t_{i-1} \rwstep{\R_E} t_i$ then, as $>$ is a $\En$-total, we know that either $\ell_i\sigma_i > r_i\sigma_i$ or $\ell_i\sigma_i =_\En r_i\sigma_i$ or $\ell_i\sigma_i < r_i\sigma_i$. The main idea behind the proof is to only consider rewrite steps that \emph{follow the order $>$}. Formally, we write $t \rightrwstep{p}{\sigma}{\ell}{r} s$ when $t_{|p} = \ell\sigma$ and $s = t[r\sigma]_p$ and ($\ell\sigma > r\sigma$ or $\ell\sigma =_\En r\sigma$). Similarly, we write $s \leftrwstep{p}{\sigma}{\ell}{r} t$ when $t \rightrwstep{p}{\sigma}{\ell}{r} s$. 

Of course, when a rule $\ell \rightarrow r$ is not already ordered by $>$, i.e. $\ell > r$, there may be some substitutions $\sigma$ for which $\ell\sigma > r\sigma$ and some substitutions $\sigma'$ for which $\ell\sigma' < r\sigma'$. This explains why in the initial set $\R$ defined in \Cref{alg:initial value} of \Cref{alg:generation}, when $\ell = r \in E'$, both $\ell \rightarrow r$ and $r \rightarrow \ell$ are in $\R$, as well as why only $\Rn \subseteq \R$ and not the rules in $\Rn$ with opposite orientation.

Using this ordered rewrite step, we can graphically represent the equality modulo $E$, $t =_E s$ as a mountainous landscape with peaks, plateaus and valleys (see \Cref{fig:mountain}) where each increase or decrease of altitude is due to an ordered rewrite step from a rule in $\R$.
The first part of the proof intuitively consists in transforming the mountain into a single valley. To do so, we will apply successive transformations that will replace local peaks into local valleys, until no peak remains.

\subsubsection{Transforming peaks into valleys}
\label{sec:peaks-into-valleys}

In this section, we present a subset of the transformations needed to reshape the mountainous landscape. In particular, we first start by looking at peaks of the form $u \leftrwstep{p}{\sigma}{\ell}{r} t \rightrwstep{p'}{\sigma'}{\ell'}{r'} v$. 




\paragraph{Peak with parallel positions.} Assume that $p \para p'$, which means that $p$ and $p'$ are not prefix of each other. We know that $\ell\sigma = t_{|p}$ and $\ell'\sigma' = t_{|p'}$. Since $p \para p'$, $t$ is in fact of the form $t =C[\ell\sigma,\ell'\sigma']$ with $C[\_,\_]$ a term context. Hence, taking $t' = C[r\sigma,r'\sigma']$, we have $u \rightrwstep{p'}{\sigma'}{\ell'}{r'} t'$ and $t' \leftrwstep{p}{\sigma}{\ell}{r} v$. Graphically, we thus applied the following transformation:

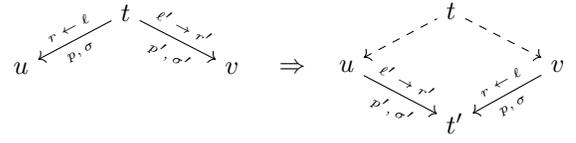
\begin{center}
    \begin{tikzpicture}[term/.style={anchor=mid,fill=white}]
        \def\length{1cm}
        \node[term] (T0) {$u$};
        \node[term,above right=0.5*\length and 1.2*\length of T0.mid] (T1) {$t$};
        \node[term,below right=0.5*\length and 1.2*\length of T1.mid] (T2) {$v$};

        \tikzleftrwstep{T0}{p}{\sigma}{\ell}{r}{T1}
        \tikzrightrwstep{T1}{p'}{\sigma'}{\ell'}{r'}{T2}

        \node[right=0.5cm of T2.mid] (Arrow) {$\Rightarrow$};

        \node[term,right=0.5cm of Arrow.mid] (T0') {$u$};
        \node[term,above right=0.5*\length and 1.2*\length of T0'.mid] (T1') {$t$};
        \node[term,below right=0.5*\length and 1.2*\length of T1'.mid] (T2') {$v$};
        \node[term,below right=0.5*\length and 1.2*\length of T0'.mid] (T1'') {$t'$};

        \draw[<-,dashed] (T0') edge (T1');
        \draw[->,dashed] (T1') edge (T2');

        \tikzrightrwstep{T0'}{p'}{\sigma'}{\ell'}{r'}{T1''}
        \tikzleftrwstep{T1''}{p}{\sigma}{\ell}{r}{T2'}
    \end{tikzpicture}
\end{center}

Since $>$ is an $\En$-compatible reduction order, $\ell\sigma > r\sigma$ implies $t > u$ and $v > t'$. Similarly, $\ell'\sigma'$ implies $t > v$ and $u > t'$. Therefore $t > t'$. Notice that the peak, whose highest altitude was represented by $t$, becomes a local valley whose highest altitude is either $u$ or $v$.


\paragraph{Peak with overlapping positions.} 
\label{sec:peak-with-overlapping-positions}

When the positions $p$ and $p'$ are not parallel, we cannot apply the previous transformation. However, we can rely on the rules obtained by merging $\En$-overlapping rules. Assume that $p' = p \cdot q$, i.e. $p$ is a prefix of $p'$, and $q \in \Pos{\ell}$ and $\ell_{|q} \not\in \X$. In such a case, $\ell_{|q}\sigma = \ell'\sigma'$. W.l.o.g., we assume that distinct rules in the mountainous landscape have distinct variables. We can therefore define $\gamma = \sigma \cup \sigma'$ yielding $\ell_{|q}\gamma = \ell'\gamma$ and so $\ell_{|q}$ and $\ell'$ being unifiable. 
Hence, not only there exist $\alpha \in \mgu{\En}{\ell_{|q},\ell'}$ and $\theta$ such that $\gamma =_\En (\alpha\theta)_{|\dom{\gamma}}$; but  the rule $r \rightarrow \ell$ is also $\En$-overlapping on $q$ with $\ell' \rightarrow r'$. Hence, we can find $(s \rightarrow w) \in (\overlapset{r \rightarrow \ell}{q}{\ell' \rightarrow r'})$ such that $s\theta =_\En r\sigma$ and $w\theta =_\En \ell\sigma[r'\sigma']_q$.

Note that we cannot deduce how $s\theta$ and $w\theta$ are ordered, i.e. whether $s\theta > w\theta$ or $s\theta =_\En w\theta$ or $w\theta > s\theta$. We thus consider two cases when transforming the peak. When $\En = \emptyset$, the two transformations are graphically represented below.

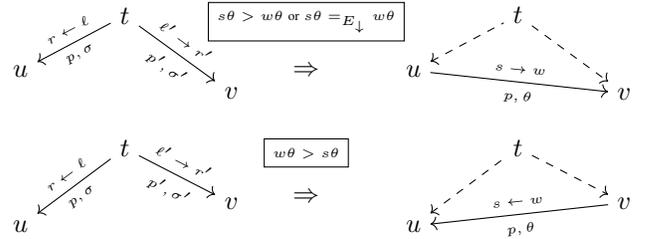
\begin{center}
    \begin{tikzpicture}[term/.style={anchor=mid,fill=white}]
        \def\length{1cm}
        \node[term] (T0) {$u$};
        \node[term,above right=0.5*\length and 1.2*\length of T0.mid] (T1) {$t$};
        \node[term,below right=0.8*\length and 1.2*\length of T1.mid] (T2) {$v$};

        \tikzleftrwstep{T0}{p}{\sigma}{\ell}{r}{T1}
        \tikzrightrwstep{T1}{p'}{\sigma'}{\ell'}{r'}{T2}
        
        \node[term,right=3.5cm of T0.mid] (Arrow) {$\Rightarrow$};
        \node[draw,rectangle,above= 0.2cm of Arrow] (Label) {\tiny $s\theta > w\theta$ or $s\theta =_\En w\theta$};

        \node[term,right=5cm of T0.mid] (T0') {$u$};
        \node[term,above right=0.5*\length and 1.2*\length of T0'.mid] (T1') {$t$};
        \node[term,below right=0.8\length and 1.2*\length of T1'.mid] (T2') {$v$};

        \draw[<-,dashed] (T0') edge (T1');
        \draw[->,dashed] (T1') edge (T2');
        \tikzrightrwstep{T0'}{p}{\theta}{s}{w}{T2'}
    \end{tikzpicture}
\end{center}

\begin{center}
    \begin{tikzpicture}[term/.style={anchor=mid,fill=white}]
        \def\length{1cm}
        \node[term] (T0) {$u$};
        \node[term,above right=0.8*\length and 1.2*\length of T0.mid] (T1) {$t$};
        \node[term,below right=0.5*\length and 1.2*\length of T1.mid] (T2) {$v$};

        \tikzleftrwstep{T0}{p}{\sigma}{\ell}{r}{T1}
        \tikzrightrwstep{T1}{p'}{\sigma'}{\ell'}{r'}{T2}
        
        \node[term,above right=0.2*\length and 3.5cm of T0.mid] (Arrow) {$\Rightarrow$};
        \node[draw,rectangle,above= 0.2cm of Arrow] (Label) {\tiny $w\theta > s\theta$};

        \node[term,right=5cm of T0.mid] (T0') {$u$};
        \node[term,above right=0.8*\length and 1.2*\length of T0'.mid] (T1') {$t$};
        \node[term,below right=0.5*\length and 1.2*\length of T1'.mid] (T2') {$v$};

        \draw[<-,dashed] (T0') edge (T1');
        \draw[->,dashed] (T1') edge (T2');
        \tikzleftrwstep{T0'}{p}{\theta}{w}{s}{T2'}
    \end{tikzpicture}
\end{center}
These two cases explain why in \Cref{alg:generation-forall} of \Cref{alg:generation}, we augment $\R'$ with $\overlapseteq{r \rightarrow \ell}{q}{\ell' \rightarrow r'}$.

When $\En \neq \emptyset$, we cannot apply exactly the same transformation since, in the mountainous landscape, all rewrite steps are syntactic and not modulo $\En$. However, $s\theta =_\En r\sigma$ implies that there exists a sequence of rewrite steps from $s\theta$ to $r\sigma$ using only rewrite rules in $\En$. To discuss about such sequences more easily, 
we write $u \Leftrightrwstep{\tr} v$ when $\tr$ is a \emph{rewrite trace}, \emph{i.e.} a sequence of rewrite step arguments $\leftrightrwlabel[\sim_1]{p_1}{\sigma_1}{\ell_1}{r_1} \ldots \leftrightrwlabel[\sim_n]{p_n}{\sigma_n}{\ell_n}{r_n}$ with ${\sim_1},\ldots,{\sim_n} \in \{\leftarrow,\rightarrow\}$ and when there exist terms $t_0,\ldots, t_n$ such that $t_0 = u$, $t_n = v$ and for all $i \in \{1, \ldots, n\}$,
\[
t_{i-1} \rightrwstep{p_i}{\sigma_i}{\ell_i}{r_i} t_i \text{ when }{\sim_i} = {\rightarrow} \text{ and } t_{i-1} \leftrwstep{p_i}{\sigma_i}{r_i}{\ell_i} t_i \text{ otherwise.}
\]
We call each $\leftrightrwlabel[\sim_i]{p_i}{\sigma_i}{\ell_i}{r_i}$ a \emph{rewrite label}.
When $\tr$ contains only right (resp. left) oriented rewrite rules, we say that it is a \emph{right (resp. left) rewrite trace}, and we denote $u \Rightrwstep{\tr} v$ (resp. $u \Leftrwstep{\tr} v$).

Coming back to the peak of our landscape, $s\theta =_\En r\sigma$ and $w\theta =_\En \ell\sigma[r'\sigma']_q$ imply $u =_\En u[s\theta]_p$ and $v =_\En v[w\theta]_p$, which in turn imply that there exist two rewrite traces $\tr_L$ and $\tr_R$ with rules only in $\En$ such that $u \Rightrwstep{\tr_L} u[s\theta]_p$ and $v[w\theta]_p \Leftrwstep{\tr_R} v$. We can therefore amend our transformations to add these rewrite traces. For example, when $w\theta > s\theta$, the transformation can be graphically represented as follows:
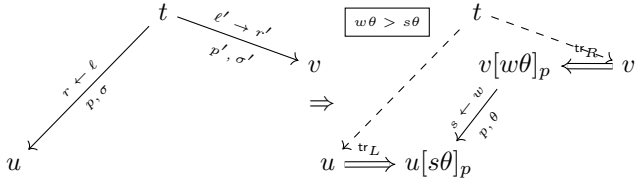
\begin{center}
    \begin{tikzpicture}[term/.style={anchor=mid,fill=white}]
        \def\length{1cm}
        \node[term] (T0) {$u$};
        \node[term,above right=1.8*\length and 1.8*\length of T0.mid] (T1) {$t$};
        \node[term,below right=0.5*\length and 1.8*\length of T1.mid] (T2) {$v$};

        \tikzleftrwstep{T0}{p}{\sigma}{\ell}{r}{T1}
        \tikzrightrwstep{T1}{p'}{\sigma'}{\ell'}{r'}{T2}
        
        \node[term,above right=0.6cm and 3.8cm of T0.mid] (Arrow) {$\Rightarrow$};

        \node[term,right=3.95cm of T0.mid] (T0') {$u$};
        \node[term,above right=1.8*\length and 1.8*\length of T0'.mid] (T1') {$t$};
        \node[term,below right=0.5*\length and 1.8*\length of T1'.mid] (T2') {$v$};

        \node[draw,rectangle,above right= 1.5cm and 0cm of T0'] (Label) {\tiny $w\theta > s\theta$};

        \node[term,right=0.9*\length of T0'.mid] (T3') {$u[s\theta]_p$};
        \node[term,left=0.9*\length of T2'.mid] (T4') {$v[w\theta]_p$};

        \draw[<-,dashed] (T0') edge (T1');
        \draw[->,dashed] (T1') edge (T2');
        \tikzleftrwstep{T3'}{p}{\theta}{w}{s}{T4'}
        \tikzRightrwstep{T0'}{\tr_L}{T3'}
        \tikzLeftrwstep{T4'}{\tr_R}{T2'}
    \end{tikzpicture}
\end{center}
Although the \emph{altitude} of the mountainous landscape decreases, the presence of $\tr_L$ and $\tr_R$ may increase its \emph{length}, i.e. the number of rewrite steps. This will be taken into account when showing that repeated applications of transformations necessarily terminate.

For brevity, we omit here the other transformations used to remove peaks, for instance when $p' = p \cdot q$ and $q \not\in \Pos{\ell}$ or $\ell_{|q} \in \X$. 
They can however be found in \Cref{sec:app-proof-main-theorem}.


\subsubsection{Ordering slopes by decreasing position}
\label{sec:ordering slopes}

In addition to removing peaks of the mountainous landscape, we also order the rewrite steps on a slope by decreasing position. Intuitively, our aim is to ensure that if the rewrite steps $u \rightrwstep{p}{\sigma}{\ell}{r} t \Rightrwstep{\tr} w \rightrwstep{p'}{\sigma'}{\ell'}{r'} v$ are part of our mountainous landscape then either $p \para p'$ or $p' < p$. In other words, if $u \Rightrwstep{\tr} v$ then any rule that affects the root symbol of $u$ should be the last rule  applied in $\tr$. 
Once again, we can achieve this by transforming the landscape. 

Consider for example the slope $u \rightrwstep{p}{\sigma}{\ell}{r} t \rightrwstep{p'}{\sigma'}{\ell'}{r'} v$ where $p' = p \cdot q$, i.e. $p$ is a prefix of $p'$, and $q \in \Pos{\ell}$ and $\ell_{|q} \not\in \X$ with $\ell'\sigma' > r'\sigma'$. These two rewrite steps are not ordered by decreasing position. 
Similarly to our transformation that removes peaks with overlapping positions, there exist $(s \rightarrow w) \in (\overlapset{\ell \rightarrow r}{q}{\ell'\rightarrow r'})$ and a substitution $\theta$ such that $s\theta =_\En \ell\sigma$ and $w\theta =_\En r\sigma[r'\sigma']_q$. 
In this case, one can show that we necessarily have $s\theta > w\theta$ as $>$ is a $\En$-strong reduction ordering and $\ell'\sigma' > r'\sigma'$. Once again, $s\theta =_\En \ell\sigma$ and $w\theta =_\En r\sigma[r'\sigma']_q$ imply that there exist two rewrite traces $\tr_L$ and $\tr_R$ with rules only in $\En$ such that $u \Rightrwstep{\tr_L} u[s\theta]_p$ and $v[w\theta]_p \Rightrwstep{\tr_R} v$. The transformation is graphically represented below.

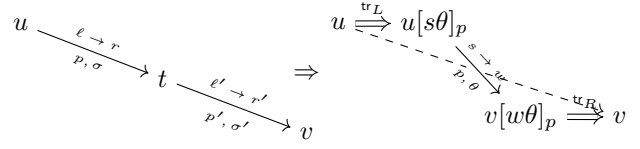
\begin{center}
    \begin{tikzpicture}[term/.style={anchor=mid,fill=white}]
        \def\length{1cm}
        \node[term] (T0) {$u$};
        \node[term,below right=0.5*\length and 1.7*\length of T0.mid] (T1) {$t$};
        \node[term,below right=0.5*\length and 1.7*\length of T1.mid] (T2) {$v$};

        \tikzrightrwstep{T0}{p}{\sigma}{\ell}{r}{T1}
        \tikzrightrwstep{T1}{p'}{\sigma'}{\ell'}{r'}{T2}
        
        \node[term,above=0.6cm of T2.mid] (Arrow) {$\Rightarrow$};

        \node[term,right=4cm of T0.mid] (T0') {$u$};
        \node[term,below right=0.5*\length and 1.7*\length of T0'.mid] (T1') {};
        \node[term,below right=0.5*\length and 1.7*\length of T1'.mid] (T2') {$v$};

        \node[term,right=0.7*\length of T0'.mid] (T3') {$u[s\theta]_p$};
        \node[term,left=0.7*\length of T2'.mid] (T4') {$v[w\theta]_p$};

        \draw[-,dashed] (T0') edge (T1');
        \draw[->,dashed] (T1') edge (T2');
        \tikzrightrwstep{T3'}{p}{\theta}{s}{w}{T4'}
        \tikzRightrwstep{T0'}{\tr_L}{T3'}
        \tikzRightrwstep{T4'}{\tr_R}{T2'}
    \end{tikzpicture}
\end{center}

When $\En = \emptyset$, as $u = u[s\theta]_p$ and $v[w\theta]_p = v$, the sequences $\tr_L$ and $\tr_R$ would be empty hence the number of wrongly ordered rewrite steps would decrease. However, when $\En \neq \emptyset$, $\tr_L$ and $\tr_R$ may introduce wrongly ordered rewrite steps. Therefore, we relax the notion of rewrite steps ordered by decreasing position by only looking at cases where 
$\ell'\sigma' > r'\sigma'$, thus excluding rules in $\En$. 


\begin{definition}
A right rewrite trace $\tr$ is \emph{ordered by decreasing position} when for all sub-traces of $\tr$ of the form $\rightrwlabel{p}{\sigma}{\ell}{r} \tr' \rightrwlabel{p'}{\sigma'}{\ell'}{r'}$, if $\ell'\sigma' > r'\sigma'$ then $p \para p'$ or $p' < p$. Similarly, a left rewrite trace $\tr$ is ordered by decreasing position when for all sub-traces of $\tr$ of the form $\leftrwlabel{p}{\sigma}{\ell}{r} \tr' \leftrwlabel{p'}{\sigma'}{\ell'}{r'}$, if $\ell\sigma > r\sigma$ then $p \para p'$ or $p < p'$.
\end{definition}

With this new definition, as $\tr_R$ only contains rules from $\En$, the rewrite trace $\rightrwlabel{p}{\theta}{s}{w} \tr_R$ is naturally ordered by decreasing position. To handle $\tr_L$, one can notice that in the equality $s\theta =_\En \ell\sigma$, we in fact have $s = \ell\alpha$ with $(\alpha\theta)_{|\dom{\sigma}} =_\En \sigma $. Thus, after showing that $\ell$ cannot be a variable as $\En$ has finite equivalence classes and $>$ has the subterm property on ground terms, we thus obtain $\ell\sigma \Rightrwstep{\tr'_L} s\theta$ for some $\tr'_L$ where the positions in $\tr'_L$ are all different from the root position $\varepsilon$, i.e. they do not affect the root of $\ell$. Therefore, we deduce that there exists $u \Rightrwstep{\tr_L} u[s\theta]_p$ where $p$ is a strict prefix of all positions of $\tr_L$.

The proof of \Cref{th:generation of rewrite theory} and in particular the proof of \Cref{M:ind1} from \Cref{def:mimics} is completed by showing that when no more landscape transformation rule is applicable, the rewrite trace is a valley whose slopes are ordered by decreasing position. In such a case, as $f(t_1,\ldots,t_n) =_E t$, $\M = \{t_1,\ldots,t_n,t\}$ and $\nf{T}{E}{\M}$ holds, we can show that the ordered slopes is in fact only composed of a single rule, that is the one used to prove \Cref{M:ind1}.


\section{Optimisations}
\label{sec:optimisations}

The procedure presented above has two main flaws: it may not terminate
and it may produce rewrite rules where variables of the right hand
side are not all included in the left hand side, the latter thus possibly violating the requirement from \Cref{S:std} of
\Cref{def:extended_signature}. Although we will never be able to
guarantee termination (otherwise all equational theories would satisfy
the finite variant property, which is not the case), we suggest in
this section some optimisations that will help the procedure
terminate more often. Moreover, we will also propose an additional
optimisation that will ensure that all rewrite rules generated have
variables in their right hand side occurring in their left-hand side.


\subsection{Dealing with right-hand side variables}

Recall that $>$ is a $\En$-strong reduction order and $\En$ has finite equivalence classes. Without restriction, the smallest ground term by $>$ could be either a name or a constant. It cannot be any other term as $>$ satisfies the subterm property on ground term, i.e. any strict subterm is strictly smaller than the term itself. To apply our optimisation, we will assume that the smallest ground term by $>$, denoted $\amin$, is a constant that does not occur in $\En$. In such a way, $\amin$ will be able to appear within our rewrite rules and not be affected by $\En$. In practice, we can always augment $\F$ with a new constant disjoint from $\F$, and define a new reduction from $>$ where this constant is minimal. 

The main idea of the rule comes from the following simple observation. Assume that $s =_E t$ and $x \in \vars{t} \setminus \vars{s}$. Thus, for all terms $u_1,u_2$, denoting $\sigma_1 = \{ x \mapsto u_1 \}$ and $\sigma_2 = \{ x \mapsto u_2 \}$, we have $t\sigma_1 =_E s =_E t\sigma_2$. In other words, the value of $x$ does not really matter. Thus, when $s \rightarrow t$ occurs in $\R$, we can replace $x$ by the minimal term $\amin$. In order to achieve correctness, we also introduce the rewrite rule $t \rightarrow t\{x\mapsto \amin\}$. Therefore, we augment the set of normalisation rules with the following rule \RVar.

\[
\begin{array}{l}
    (\RVar)\\
    \R \cup \{ s \rightarrow t \}  \normstep[\Rn,\En] \R \cup \{ s \rightarrow t\rho, t \rightarrow t\rho\} \\
    \quad\hfill\text{if }\rho : V \rightarrow \{\amin\} \text{ and }V = (\vars{t}\setminus \vars{s}) \neq \emptyset
\end{array}
\]

For example, this optimisation allows our prototype to generate a rewrite theory for the equational theory $\{g(x,y) = f(x,z);  h(x) = a; f(a,z) = p(x,z)\}$.


\subsection{Checking the order of rules}
\label{sec:order of rules}

So far, we assumed the existence of an $\En$-strong reduction order
but it is only used in the proof and not in the algorithm
itself. However, if we have an effective way of testing whether $t >
s$ then we can remove any rule $s \rightarrow t$ from $\R$
such that $t > s$ since they will never be used in a rewrite trace between two terms. Indeed, 
$u \rightrwstep{p}{\sigma}{s}{t} v$ and $u \leftrwstep{p}{\sigma}{s}{t} v$ requires that either $s\sigma =_\En t\sigma$ or $s\sigma > t\sigma$; and $t > s$ implies $t\sigma > s\sigma$ for all $\sigma$. In a similar fashion, as $>$ satisfies the subterm property on ground terms, $s$ cannot be a strict subterm of $t$ as it would imply $t\sigma > s\sigma$. We can therefore augment the set of normalisation rules with the following rule \ROrd.

\[
\begin{array}{l}
    (\ROrd)\\
    \R \cup \{ s \rightarrow t \}  \normstep[\Rn,\En] \R \quad\text{if $t > s$ or $s \in \sst{t}$}
\end{array}
\]

On the other hand, when a rule $s \rightarrow t$ already satisfies $s > t$, we also have the opportunity to add it to $\Rn$ as the only requirement on $\Rn$ is to be $\En$-terminating. When it occurs, we restart the algorithm by re-initialising $\R$ to its initial value with the augmented $\Rn$. This optimisation allows the prototype to terminate on the equational theory $\{\mathrm{exp}(\mathrm{g},\mathrm{one}) = g; \mathrm{exp}(\mathrm{exp}(\mathrm{g},x),y) = \mathrm{exp}(\mathrm{exp}(\mathrm{g},y),x)\}$.


\subsection{Subsumed rules modulo \texorpdfstring{$\Ea$}{Ea}}

\Cref{th:generation of rewrite theory} generates a rewrite theory mimicking $E$, of the form $(>,\R,\Rn,\En,\En)$, where the
equational theory $\En$ used for normalisation is also used for the
computation of most general unifiers. As mentioned at the beginning of
the paper, we aim for rewrite theories of the form
$(>,\R,\Rn,\En,\Ea)$. Of course, when $\En \subseteq \Ea$, it is easy to see that if $(>,\R,\Rn,\En,\En)$ mimics $E$ then so does $(>,\allowbreak\R,\Rn,\En,\Ea)$. However, many of the rules in $\R$ would become superfluous: typically, when two rules are equal modulo $\Ea$, only one of them is needed. We therefore consider a function that will remove these superfluous rules (\Cref{alg:superfluous}).

\begin{algorithm}[ht]
\Fn{\Cleanup{$\R,\Ea$}}{
    \lWhile{$\R = \R' \cup \{ s \rightarrow r, s' \rightarrow r'\}$ and $\exists \alpha$ s.t. $s\alpha \eqES{\Ea} s'$ and $r\alpha\eqE{\Ea} r'$}{\label{line:superfluous-while}
        $\R := \R' \cup \{ s \rightarrow r\}$
    }
    \KwRet{$\R$}
}
\caption{Removing superfluous rewrite rules}
\label{alg:superfluous}
\end{algorithm}

Correctness of this function is given below
(proof in \Cref{sec:app-optimisation}).


\begin{restatable}{lemma}{lemcleanup}
\label{lem:cleanup}
If $(>,\R,\Rn,\En,\Ea)$ is a rewrite theory mimicking an equational theory $E$ and $\R' = \Cleanup{$\R,\Ea$}$ then $(>,\R',\Rn,\En,\allowbreak\Ea)$ is a rewrite theory mimicking $E$.
\end{restatable}

\section{Experimentation Results and Discussion}
\label{sec:discussion}

\newcommand{\pow}{\,\text{\^{}}\,}

\begin{table*}[h!]
    \scriptsize
\begin{center}
\begin{tabular}{|@{\,}c@{\,}|@{\,}c@{\,}|c|c|c|c|c|}
\hline
Acronym & Equations & time & \# rules & $|\R|$ & Conv & $\Rn$ \\
\hline
\hline

ACI & AC($+$) $\cup$ $\{ x + x = x \}$ &  5s  & 54k & 6 & yes & no \\
\hline

ACN & AC$(+)$ $\cup$ $\{ x + x = 0 \}$ &  35s & 399k & 14 & yes & no \\
\hline

$\XOR$ & AC($+$) $\cup$ $\{ x + 0 = x, x + x = 0\} $ & 1s & 38k & 8 & yes & no \\
\hline

$\XOR$ed & $\XOR \cup \{ \mathrm{d}(\mathrm{e}(x + k_2,k_1),k_1,k_2) = x \}$ & 3s & 56k & 17 & yes & no \\
\hline

$\mathrm{H}(+_{\text{in}},+_{\text{out}})$  & $\{h(x +_{\text{in}} y) = h(x) +_{\text{out}} h(y)  \}$ & 1s & 10 & 5 & yes & no \\
\hline
$\XOR h_{\neq}$ & $\XOR \cup \mathrm{H}(+,+_{\text{out}})$ & 2s & 40k & 17 & yes & no \\
    \hline

$\AG$ & AC$(*)$ $\cup$ $\{ x * 1 = x, x * x^{-1} = 1\}$ & 7min 26s & 530k & 52 & yes & yes \\
\hline
DH& $\AG$ $\cup$ $\{ x \pow 1 = x, (x \pow y) \pow z = x \pow (y*z)\}$  & 17min & 828k & 99 & yes & yes \\
\hline
\multirow{2}{*}{Bilinear Pairing} & DH $\cup\ \{ (x \cdot y) \cdot z = x \cdot (y * z), x \cdot 1 = x,$ & \multirow{2}{*}{1h 14min} & \multirow{2}{*}{1402k} &
\multirow{2}{*}{194} & \multirow{2}{*}{no} & \multirow{2}{*}{yes} \\ 
&  $\mathrm{em}(x \cdot y,z) = \mathrm{em}(x,z) * y, \mathrm{em}(z,x \cdot y) = \mathrm{em}(z,x) * y\}$ & & & & &\\
\hline
OldDHWeak & $\{ \mathrm{g}_w \pow x = x, 1 \pow x = 1, (\mathrm{g} \pow y) \pow z = (\mathrm{g} \pow z) \pow y  \}$ & 1s & 32 & 10 & no & no\\

\hline    

$\AG h_{\neq}$& $\AG \cup H(*,*_{\text{out}})$ & 16min 42s & 815k & 101 & yes & yes \\
\hline

\multirow{2}{*}{EG-Mixnet~\cite{usenix-mixnet}} & AC$(*) \cup \{ (x \pow y) \pow z = x \pow (y*z),
\mathrm{dec}(\mathrm{enc}(m, x, x \pow y, r), x, y) = m, $ & \multirow{2}{*}{1s} & \multirow{2}{*}{276} &
\multirow{2}{*}{15} & \multirow{2}{*}{yes} & \multirow{2}{*}{no} \\ 
&  $\mathrm{check}(\mathrm{sign}(m, x, s), x, x \pow s) = m,
\mathrm{get}(\mathrm{sign}(m, x, s)) = m\}$ & & & & &\\

\hline
EG-Renc~\cite{usenix-mixnet} & EG-Mixnet $\cup$ AC$(+) \ \cup\{ renc(enc(m,x,x \pow y,r),r',x,x \pow y) = enc(m,x,x \pow y,r + r') \}$  & 1s & 428 & 19 & yes & no \\ 
\hline
Enc$h_{\neq}$ & AC$(+) \cup \{ \mathrm{enc}(x,z) +_{\text{out}} \mathrm{enc}(y,z) = \mathrm{enc}(x+y,z), 
\mathrm{dec}(\mathrm{enc}(x,y),y) = x \}$ & 1s & 24 & 8 & yes & no\\
\hline
AEnc$h_{\neq}$ & AC$(+) \cup \{ \mathrm{enc}(x,\mathrm{pk}(z)) +_{\text{out}} \mathrm{enc}(y,\mathrm{pk}(z)) = \mathrm{enc}(x+y,\mathrm{pk}(z)), 
\mathrm{dec}(\mathrm{enc}(x,\mathrm{pk}(y)),y) = x \}$ & 1s & 24 & 7 & yes & no\\

\hline
\multirow{2}{*}{Blind Signature} & $\{ \mathrm{unblind}(\mathrm{blind}(x,y),y) = x, 
\mathrm{get}(\mathrm{sign}(x,y)) = x,$&\multirow{2}{*}{1s}&\multirow{2}{*}{72}&\multirow{2}{*}{11}& \multirow{2}{*}{yes}&\multirow{2}{*}{no}\\
&  $\mathrm{unblind}(\mathrm{sign}(\mathrm{blind}(x,y),z),y) = \mathrm{sign}(x,z),
\mathrm{verify}(\mathrm{sign}(x,y),\mathrm{vk}(y)) = \mathrm{ok}\}$ & & & & &\\
\hline
Toy example 1 & $\{ g(x,y) = f(x,z), h(x) = 0, f(0,z) = p(x,z)\}$ & 1s & 218 & 14 & no & no \\
\hline
Toy example 2 & $\{g(x,y) = f(x,x,z),h(x,x) = 0,f(0,x,z) = p(x,z)\}$ & 1s & 159 & 12 & no & no \\
\hline
Toy example 3 & $\{\mathrm{dh}(x,\mathrm{pk}(y)) = \mathrm{dh}(y,\mathrm{pk}(x)), \mathrm{dh}(x,\mathrm{invalid}) = \mathrm{invalid}\}$ & 1s & 21 & 5 & no & no \\
\hline
\end{tabular}
\end{center}

\caption{Extract of equational theories handled by our prototype.}
\label{tab:theories}
\end{table*}

We developed a prototype implementing our semi-decision procedure in
OCaml and used Maude version 3.4~\cite{maude} as backend. Our source code and examples are available in \cite{prototypeFVP}. The prototype currently natively only support matching and unification modulo AC, meaning that given an input equational theory $E$, and an optional rewrite system $\Rn$, the prototype will  execute \GenExtended{$E,\Rn, AC$}, returning a rewrite system $\R$. 
In \Cref{tab:theories}, we provide an extract of our experimental results on different equational theories, indicating the equations we consider, the computation time on a MacBook Pro M2 8-core with 24GB of memory, the number of rewrite rules generated by the algorithm (\# rules), the size of $\R$ ($|\mathcal{R}|$) and whether our prototype was able to determine that $\R$ is convergent (Conv). Finally, we also indicate whether  a rewrite system $\Rn$ was initially provided to the prototype ($\Rn$). We explain below how our prototype check convergence.

In \Cref{tab:theories}, we denote by AC$(+)$ the equations for associativity and commutativity of the symbol $+$. In~\cite{DBLP:conf/rta/Comon-LundhD05}, it was shown that adding to $\XOR$ an  homomorphic symbol $h$ with the equation $h(x+y) = h(x) + h(y)$ yields an equational theory that does not have the FVP modulo AC. However, we retrieve the FVP modulo AC when considering different operators $+_{\text{in}}$ and $+_{\text{out}}$ with the equation $\{ h(x +_{\text{in}} y) = h(x) +_{\text{out}} h(y)\}$. In the case of $\XOR$ for example, the equations $\XOR \cup H(+,+_{\text{out}})$ ensure that $+_{\text{out}}$ is also AC, nilpotent with a unit element but only over the outputs of the $h$ function (not over all terms). The same technique applies to Abelian groups with an homomorphic symbol and (a)symmetric homomorphic encryption/decryption scheme. 

We also consider more complex equations used in the literature. For instance, we consider a model of ElGamal asymmetric encryption used to model Exponentiation Mix-Nets in~\cite{usenix-mixnet}. We also consider the Bilinear Pairing equations used in Tamarin, which is basically Diffie Hellman (DH) augmented with a bilinear map $\mathrm{em}$ and a scalar multiplier $\cdot$. We relied on the generic definition where the function $\mathrm{em}$ is not commutative, corresponding to the case where the groups used on the left and right arguments of $\mathrm{em}$ may be different. Interestingly, 
when adding the commutative property to $\mathrm{em}$, our algorithm seems \emph{stuck} in the computation of an AC-unification in Maude. However, when considering $\mathrm{em}$ to be both associative and commutative, our algorithm concludes and generates the same number of rewrite rules as the non-commutative case. 
This can be explained by the fact that our prototype currently only implements natively AC unification and not commutative-only unification. Therefore the algorithm needs to take care of the commutative property which seems to lead to some very costly unification. In a future version of our prototype with commutative unification implemented, we expect the bilinear pairing with a commutative $\mathrm{em}$ to conclude without problem.

Additional examples, including some toys examples, can be found in \cite{prototypeFVP}.

We now discuss some additional observations on our procedure and its implementation.


\subsubsection{Detecting convergent equational theories}

Upon closer examination, our algorithm shares several similarities with the Knuth-Bendix algorithm used to show that an equational theory is convergent. The generation of $(\overlapseteq{r_1 \rightarrow \ell_1}{p,\En}{\ell_2 \rightarrow r_2})$ and some of our normalisation rules intuitively correspond to the rules in the Knuth-Bendix algorithm. As such, it is hardly surprising that our algorithm also allows us to determine whether an equational theory can be represented by a rewrite system $\R$ convergent modulo $\En$. In particular, if our algorithms with optimisations returns the set $\R$ then it suffices to check whether (almost) all rules in $\R$ are ordered by $>$. 
To be more specific, removing from $\R$ the rules of the form $\ell \rightarrow \ell$ added on \Cref{alg:final_modification} yields a rewrite system $\R'$. It then suffices to check that $>$ is compatible with $\R'$ to show that $\R'$ is $\En$-convergent for $E$. 

Termination is directly guaranteed. Confluence is given by  \Cref{th:equality modulo}. Indeed, denoting $T = (>,\allowbreak\R,\Rn,\En,\En)$, with a simple induction one can show that $\symbeval{t} \Eval{T} t'$ implies that $t \rwsteps[\En]{\R'} t'$. Hence, the $\En$-confluence of $\R'$ is given by  \Cref{th:equality modulo}, since $t =_E s$ implies that there exist $t',s'$ such that $\symbeval{t} \Eval{T} t'$, $\symbeval{s} \Eval{T} s'$, and $t' =_\En s'$. This leads to an interesting observation:


\begin{lemma}
If $E$ has the finite variant property modulo $\En$, then there exists a finite rewrite system $\R$ that is $\En$-confluent for $E$.
\end{lemma}


\subsubsection{Existence of an AC-strong reduction order compatible with a rewrite system}

Our procedure relies at minimum on the existence of an $\En$-strong reduction order which, as previously mentioned, can be quite challenging to prove.
In our prototype, we rely on the AC-compatible and AC-total reduction order of \cite{DBLP:conf/rta/Rubio99}. 
It is fully syntactic and RPO based, meaning that it is efficient and allows us to apply the optimisation described in \Cref{sec:order of rules}. Notice that for this order to be AC-strong, it additionally requires to be stable by renaming. Being RPO based, this can easily be achieved by encoding names into ground terms $\mathsf{N} = \{ \mathsf{a}, \mathsf{sc}(\mathsf{a}),\allowbreak \mathsf{sc}(\mathsf{sc}(\mathsf{a})), \ldots\}$.


Although our experiments showed that using $>_{AC}$ works well in practice, it can be enhanced by providing a rewrite system to the algorithm instead of letting the algorithm try to create one on its own. This is for example the case with the Abelian Group ($\mathcal{AG}$) equational theory. In \cite{DBLP:conf/rta/Comon-LundhD05}, $\AG$ was shown to satisfy the FVP modulo AC using an AC-convergent rewrite system first proposed by Lankford~\cite{hullot1980catalogue}.
Denoting this rewrite system $\R_{\mathcal{AG}}$, we call our prototype with $\R_{\mathcal{AG}}$ given as input. However, this is only correct if we can show the existence of an AC-strong reduction order compatible with $\R_{\mathcal{AG}}$. It is well known that from an $E$-terminating rewrite system, one can build an $E$-compatible reduction order, but obtaining totality is not always possible. 
For example, consider the rewrite system $\R = \{ f(x,x) \rightarrow k ,  k \rightarrow f(k_1,k_2)\}$ with $k,k_1,k_2$  constants. $\R$ is convergent but there exists no reduction order $\emptyset$-total and $\emptyset$-compatible with $\R$. Indeed, totality implies $f(k_1,k_2) > f(a,a)$ for some minimal ground term $a$; and compatibility with $\R$ entails that $f(a,a) > k > f(k_1,k_2)$, leading to a contradiction.

Nevertheless, we show the existence of an AC-strong reduction order compatible with $\R_{\mathcal{AG}}$ by composing $>_{AC}$ with the order that was used to show termination of $\R_{\mathcal{AG}}$~\cite{hullot1980catalogue}. 
The composition of orders is given by the following result 
(proof in \Cref{sec:app-order}).

\begin{restatable}{lemma}{lembuildingEstrongorder}
\label{lem:building E-strong reduction order}
Let $E$ be an equational theory. Let $\equiv$ be an equivalence relation on terms closed by application of contexts, substitutions and renaming, and such that $u =_E t$ implies $u \equiv t$. Let $>_1$ be an $E$-strong reduction order. Let $\R$ be a set of rewrite rules and $>$ be a reduction order stable by renaming such that:
\begin{itemize}
\item if $s \equiv u > v \equiv t$ then $s > t$ (\emph{$\equiv$-compatible});
\item for all ground terms $u,v$, either $u > v$ or $v > u$ or $u \equiv v$ (\emph{$\equiv$-total});
\item for all $(\ell \rightarrow r) \in \R$, $\ell > r$;
\item for all $a,b \in \N$, $a > b$ implies $a >_1 b$.
\end{itemize}
Then there exists an $E$-strong reduction order $>_2$ compatible with $\R$.
\end{restatable}

\begin{example}
Consider the reduction order $>_\#$ from \Cref{ex:reduction order count} and the order $>_{AC}$ from \Cref{ex:order AC}. Define $\R = \{ s \rightarrow t \mid s >_\# t \}$ and $\equiv$ the smallest equivalence relation such that $s \equiv t$ implies either $s =_{AC} t$ or there exist $a,b$ ground terms and a term context $C[\_]$ such that $s = C[a]$ and $C[b] = t$ and $\#(a) = \#(b)$. Applying \Cref{lem:building E-strong reduction order} with $>$ being $>_\#$ and $>_1$ being $>_{AC}$, we deduce the existence of an $AC$-strong reduction order $>_2$ compatible with $\R$. More precisely, the proof of \Cref{lem:building E-strong reduction order} defines $>_2$ as the transitive closure of $>'_2$ where $s >'_2 t$ iff:
\begin{itemize}
\item either $s >_\# t$
\item or $s \equiv t$ and there exists $u,v$ ground terms and a term context $C[\_]$ such that $s =_{AC} C[u]$, and $t =_{AC} C[v]$, and $u >_{AC} v$
\end{itemize}
In our prototype, we have also implemented the reduction order $>_2$. Note that increasing the number of reduction orders handled by our prototype also increases its chance to generate a rewrite theory mimicking the target equational theory. Indeed, changing the reduction order impacts the termination of the algorithm. It is however difficult to estimate beforehand which order is more suited to a given equational theory. 
\end{example}


\subsubsection{Limitations and future work}

Although the optimised algorithm allows to transform more equational theories than \Cref{alg:generation}, both are still confined to the realm of theories having the FVP modulo $\En$. The function \Cleanup{$\R,\Ea$} takes $\Ea$ into account by removing superfluous rules but does not change the fact that $\R$ was built with $\En$, i.e. the equational theory has the FVP modulo $\En$. This is a limitation of our algorithm. Indeed, even if an equational theory $E$ does not have the FVP modulo $E$ (e.g. $\XOR$ with homomorphic symbol), $E$ trivially has the FVP modulo $E$ (itself), and so we could expect other interesting equational theories that include $E$ to have the FVP modulo $E$.
We conjecture that modifying the normalisation rule \REq and \RSubsume to consider $=_\Ea$ instead of $=_\En$ should be a step in the right direction, possibly with the hypothesis that there exists a rewrite system $\R_{\mathcal{A}}$ that is $\En$-convergent for $\Ea$. However, the proof we present in this paper cannot be as simply adapted. Indeed, by going from $\En$ to $\Ea$, we are losing very critical properties ($>$ is not $\Ea$-compatible and $\Ea$ does not have finite equivalence classes) for the termination of our transformations in the proof. 
As future work, we plan to address these challenges.

Our prototype showed that for some examples relying on AC-unification, the computation time could go from several minutes up to more than one hour. Although our prototype was developed as a proof of concept, we used the tool Maude~\cite{maude} to efficiently compute the most general unifiers modulo AC. However, most generated rules are duplicates of other generated rules. For example, applying our algorithm on ACN generates around 400k rules despite finally keeping only 14. We plan to identify new invariants in the proof of \Cref{th:generation of rewrite theory} that would reduce the number of generated rules, hence increasing our algorithm's efficiency. 

Finally, the genericity of our prototype makes it an ideal candidate to be integrated to cryptographic verifiers such as Tamarin and ProVerif.  We plan to first tackle its integration to ProVerif as the rewrite theory introduced in this paper is a generalisation of the framework used in ProVerif. Nevertheless, for both tools, it may require a major overhaul of the theory behind them. But this is, in our opinion, anyway the next logical step for automatic verifiers to handle more cryptographic primitives.

\bibliographystyle{IEEEtran}
\bibliography{references}

\newpage 

\onecolumn

\appendix


\subsection{Proofs of Section~\ref{sec:prelim}}
\label{sec:app-section}


\begin{restatable}{lemma}{lemminandrenaming}
\label{lem:min and renaming}
Let $E \subseteq E'$ be two equational theories. Let $>$ be a $E$-strong reduction order. Let $k \in \N$ be the smallest name by $>$. For all $t\in \T(\F,\N)$, for all renamings $\rho$ preserving $>$, if
\begin{itemize}
\item $\names{t} \subseteq \dom{\rho}$
\item $k \in \img{\rho} \Leftrightarrow k \in \dom{\rho}$
\end{itemize}
then $\minOrd[>]{E'}(t\rho) =_E \minOrd[>]{E'}(t)\rho$.
\end{restatable}

\begin{proof}
We first prove a small property on $\minOrd[>]{E'}(\cdot)$: for all terms $u$, for all $a \in \names{\minOrd[>]{E'}(u)} \setminus \names{u}$, $a = k$. Indeed, by definition, $u =_{E'} \minOrd[>]{E'}(u)$. Thus, if $a \in \names{\minOrd[>]{E'}(u)} \setminus \names{u}$, then we have that for all $b \in \N$, $\minOrd[>]{E'}(u)\{ a \mapsto b\} =_{E'} u$. By contradiction, assume $k \neq a$. As $k$ is the smallest name by $>$, we deduce that $a > k$. Since $>$ is a reduction order, we deduce that $\minOrd[>]{E'}(u) > \minOrd[>]{E'}(u)\{ a \mapsto k\}$. This is in contradiction with the definition of $\minOrd[>]{E'}(u)$ considering that $\minOrd[>]{E'}(u) =_{E'} \minOrd[>]{E'}(u)\{ a \mapsto k\}$.

We now proceed with the proof of the main result. Assume by contradiction that $\minOrd[>]{E'}(t\rho) \neq_E \minOrd[>]{E'}(t)\rho$. Therefore, $\minOrd[>]{E'}(t)\rho > \minOrd[>]{E'}(t\rho)$. By hypothesis, $\names{t} \subseteq \dom{\rho}$. If $k \not\in \dom{\rho}$ then we define $\rho' = \rho \{ k \mapsto k\}$ else $\rho' = \rho$. Hence, we deduce that $\rho'$ is a renaming preserving $>$ such that $k\rho' = k$ and $k \in \dom{\rho'}$. As such, $\rho^{-1}$ is also a renaming preserving $>$. 
From our small property, we know that $\names{\minOrd[>]{E'}(t)} \subseteq \dom{\rho'}$ and $\minOrd[>]{E'}(t)\rho = \minOrd[>]{E'}(t)\rho'$. Similarly, $t\rho = t\rho'$ and $\names{\minOrd[>]{E'}(t\rho)} \subseteq \img{\rho'}$. 

Therefore, we have $\names{\minOrd[>]{E'}(t\rho),\minOrd[>]{E'}(t)\rho} \subseteq \dom{\rho^{'-1}}$. As $>$ is stable by renaming, we have $\minOrd[>]{E'}(t) > \minOrd[>]{E'}(t\rho)\rho^{'-1}$. Recall that $t =_{E'} \minOrd[>]{E'}(t)$, meaning that $t\rho =_{E'} \minOrd[>]{E'}(t)\rho$ and so  $\minOrd[>]{E'}(t\rho) =_{E'} \minOrd[>]{E'}(t)\rho$ by definition of $\minOrd[>]{E'}(t\rho)$. This yields $\minOrd[>]{E'}(t\rho)\rho^{'-1} =_{E'} \minOrd[>]{E'}(t)$, which is a contradiction with the definition of $\minOrd[>]{E'}(t)$ and with $\minOrd[>]{E'}(t) > \minOrd[>]{E'}(t\rho)\rho^{'-1}$.
\end{proof}

\subsection{Proof of Section~\ref{sec:extended-signature}}
\label{sec:app-extended-signature}

\lemnormalisation*

\begin{proof}
Let $k$ be the smallest name by $>$. By definition, there exists an injective substitution $\sigma: \vars{\M} \rightarrow \N$ such that:
\begin{itemize}
\item $\forall a \in \img{\sigma}$, $a \neq k$ and $\forall b \in \names{\M}$, $a > b$
\item $\forall t \in \M$, $t\sigma =_\En \minOrd[>]{E}(t\sigma)$.
\end{itemize}
Before tackling the four properties, we start by proving the following property: for all $t \in \st{\M}$, $t\sigma =_\En \minOrd[E]{>}(t\sigma)$. Assume by contradiction that there exist $t \in \M$ and $p \in \Pos{t}$ such that $t_{|p}\sigma \neq_\En \minOrd[E]{>}(t_{|p}\sigma)$. By definition of $\minOrd[E]{>}(t_{|p}\sigma)$, $t_{|p}\sigma =_E \minOrd[E]{>}(t_{|p}\sigma)$. Since $>$ is $\En$-total, we deduce that $t_{|p}\sigma > \minOrd[E]{>}(t_{|p}\sigma)$. Moreover, $>$ is also an $\En$-compatible reduction order, thus we deduce that $t\sigma > t\sigma[\minOrd[E]{>}(t_{|p}\sigma)]_p$ and $t\sigma =_E t\sigma[\minOrd[E]{>}(t_{|p}\sigma)]_p$. It would contradict the fact that $t\sigma =_\En \minOrd[E]{>}(t\sigma)$.

Let us start with the proof of \Cref{lem:norm-irreducible}.
Assume that $t$ is reducible by $\rwstep[\En]{\Rn}$. Hence there exist $(\ell \rightarrow r) \in \Rn$, a substitution $\alpha$ and $p \in \Pos{t}$ such that $t_{|p} =_\En \ell\alpha$ and $t \rwstep[\En]{\Rn} t[r\alpha]_p$. As $\En$ is regular and as $\vars{r} \subseteq \vars{\ell}$ thanks to \Cref{S:std}, we deduce that $\vars{r\alpha} \subseteq \vars{t}$.

Moreover, $t_{|p} =_\En \ell\alpha$ implies $t_{|p}\sigma =_\En \ell\alpha\sigma$, meaning that $t\sigma \rwstep[\En]{\Rn} t\sigma[r\alpha\sigma]_p$. By hypothesis, we deduce that $t\sigma =_E t\sigma[r\alpha\sigma]_p$. Since $t\sigma$ is ground and $\vars{r\alpha} \subseteq \vars{t}$, $r\alpha\sigma$ and $t\sigma[r\alpha\sigma]_p$ are also ground. However, the ordering $>$ is $\En$-compatible with $\Rn$. Hence $t\sigma > t\sigma[r\alpha\sigma]_p$, which is a contradiction with $t\sigma =_\En \minOrd[>]{E}(t\sigma)$.

For the proof of \Cref{lem:norm-eq}, our initial property gives us $s\sigma =_\En \minOrd[>]{E}(s\sigma)$ and $t\sigma =_\En \minOrd[>]{E}(t\sigma)$. Moreover, $s =_E t$ implies $s\sigma =_E t\sigma$. Therefore, by definition of $\minOrd[>]{E}(\cdot)$,we have  $\minOrd[>]{E}(s\sigma) = \minOrd[>]{E}(t\sigma)$ and then $s\sigma =_\En t\sigma$. It entails $s =_\En t$ as $\sigma$ is injective and $\names{s,t} \cap \img{\sigma} = \emptyset$.

For the proof of \Cref{lem:norm-var}, we start by showing that for all names $a$, $\minOrd[>]{E}(a) =_\En a$. Let us denote $t = \minOrd[>]{E}(a)$ and assume that $t \neq_\En a$, and then $t < a$. By definition, $t =_E a$. Assume first that $a \in \st{t}$. Since $>$ is a reduction order, $a$ cannot be a strict subterm of $t$ (otherwise, if $t=C[a]$, we would have the infinite sequence of distinct terms $a > C[a] > C[C[a]] > C[C[C[a]]] > \ldots$, which is a contradiction with $>$ being well-founded). Thus $a \in \st{t}$ implies $t = a$, which is a contradiction with $t \neq_\En a$. Assume now that $a \not\in \st{t}$. As equality modulo $E$ is stable by replacement of names by arbitrary terms, we deduce that for all terms $t_1,t_2$, we obtain $t_1 =_E t =_E t_2$ (by replacing $a$ once by $t_1$ and once by $t_2$). This contradicts the fact that $E$ is a non-trivial equational theory.

We can now complete the proof of \Cref{lem:norm-var} as follows. If $x \not\in \dom{\sigma}$, we can define $\sigma' = \sigma\{ x \mapsto a\}$ for some $a \in \N$ such that $a$ is different from $k$ and greater than all names in $\img{\sigma}$ and in $\names{\M}$. Since $\N$ is infinite, it is always possible to do so. If $x \in \dom{\sigma}$, we define $\sigma' = \sigma$. Hence, $\M\sigma = \M\sigma'$ and $x\sigma'\in \N$. As we proved that for all names $a$, $\minOrd[>]{E}(a) =_\En a$, we obtain $\minOrd[>]{E}(x\sigma') =_\En x\sigma'$,  which allows us to conclude.

Finally, let us focus on the proof of \Cref{lem:norm-addition}. We first want to avoid the case where $\img{\sigma} \cap \names{t} \neq \emptyset$. We thus take a fresh renaming $\rho$ preserving $>$, that renames the names in $\img{\sigma}$. As $\N$ is infinite, we can even take the new names to be all greater than the one in $\sigma$, i.e. $\dom{\rho} = \img{\sigma}$ and for all $a \in \img{\rho}$, $a \neq k$ and for all $b \in \img{\sigma} \cup \names{t}$, $a > b$. Let us define $\rho' = \rho \{ a \mapsto a \mid a \in \names{\M,t}\}$. Hence, as $\rho$ is preserving $>$ and all names in $\img{\rho}$ are greater than the ones in $t$ and $\M$, we deduce that $\rho'$ is also preserving $>$, $\names{t\sigma,\M\sigma} \subseteq \dom{\rho'}$ and if $k \in \img{\rho'} \cup \dom{\rho'}$ then $k\rho' = k$. 
Hence, by \Cref{lem:min and renaming}, we deduce that for all $v \in \M$, $v\sigma\rho' =_\En \minOrd[>]{E}(v\sigma)\rho' =_\En \minOrd[>]{E}(v\sigma\rho')$. Finally, let us take the substitution $\alpha$ such that for all $x \in \dom{\sigma}$, $x\alpha = x\sigma\rho'$ and for all $x \in \vars{t} \setminus \dom{\sigma}$, $x\alpha$ is a fresh name (we take it greater than the ones in $\rho'$). We know that $t\alpha =_E \minOrd[>]{E}(t\alpha)$. Hence, by taking $s$ the term $\minOrd[>]{E}(t\alpha)$ where we replace all occurrences of names $a$ by $a\alpha^{-1}$, we obtain that $s =_E t$, and $s\alpha = \minOrd[>]{E}(t\alpha) = \minOrd[>]{E}(s\alpha)$ (since $\minOrd[>]{E}(\minOrd[>]{E}(u)) = \minOrd[>]{E}(u)$ for all ground terms $u$). As we already showed that for all $v \in \M$, $v\alpha =_\En \minOrd[>]{E}(v\alpha)$, we conclude.
\end{proof}

\subsection{Proofs of Section~\ref{sec:extended-signature}}
\label{sec:app-proof variant}

In this section, all lemmas consider an equational theory $E$ and a rewrite theory $T = (>,\R,\Rn,\En,\Ea)$ mimicking $E$.

\subsubsection{Equality modulo \texorpdfstring{$E$}{E}}

\begin{lemma}
\label{lem:eval}
Let $\sigma$ be a substitution and $t,t'$ two terms. If $t' =_E t\sigma$ and $\nf{T}{E}{\{ t' \} \cup \{ x\sigma \mid x \in \vars{t}\}}$ then $\symbeval{t}\sigma \Eval{T} t''$ and $t'' =_\Ea t'$ for some $t''$.
\end{lemma}

\begin{proof}
The proof is by induction on $t$.

\paragraph{Case $t = x$:} By definition $\symbeval{x} = x$. Moreover, we have $x\sigma =_E t'$. Since $\nf{T}{E}{\{x\sigma,t'\}}$, we deduce that $x\sigma =_\En t'$. By \Cref{S:subset} of \Cref{def:extended_signature}, we know that $\En \subseteq \Ea$ hence the result holds with $t'' = x\sigma$.

\paragraph{Case $t = f(t_1,\ldots,t_n)$:} We have $t' =_E f(t_1\sigma,\ldots,t_n\sigma)$. Let us denote $\M = \{ t' \} \cup \{ x\sigma \in \vars{t}\}$. By \Cref{lem:norm-addition} of \Cref{lem:norm}, there exists $t'_1,\ldots,t'_n$ such that $\nf{T}{E}{\M \cup \{ t'_1,\ldots, t'_n\}}$. By \Cref{M:ind1} of \Cref{def:mimics}, there exist a substitution $\sigma'$ and $f(\ell_1,\ldots,\ell_n) \rightarrow r$ in $\R$ such that $t' =_\Ea r\sigma'$ and for all $i \in \{1,\ldots,n\}$, $\ell_i\sigma' =_\Ea t'_i$. By applying our inductive hypothesis on $t_1,\ldots,t_n$, we deduce that there exist $t''_1,\ldots,t''_n$ such that for all $i \in \{1,\ldots,n\}$, $\symbeval{t_i}\sigma \Eval{T} t''_i$ and $t''_i =_\Ea t'_i =_\Ea \ell_i\sigma'$. By \Cref{def:evaluation}, we deduce that $\symbeval{f}(\symbeval{t_1},\ldots,\symbeval{t_n}) \Eval{T} r\sigma'$. Since we already showed that $r\sigma' =_\Ea t'$, we conclude.
\end{proof}


\begin{lemma}
\label{lem:eval_complete}
Let $t$ be a TE-term. If $t \Eval{T} s$ then $\removeeval{t} =_E s$.
\end{lemma}

\begin{proof}
The proof is by induction on $t$.

\paragraph{Case $t \in \T(\F,\X)$:} In such a case $\removeeval{t} = t \Eval{T} t$ and so the result holds.

\paragraph{Case $t = \symbeval{f}(t_1,\ldots,t_n)$:} By definition, $t \Eval{T} s$ implies that there exists $f(\ell_1,\ldots,\ell_n) \rightarrow r \in \R$ and a substitution $\sigma$ such that $s = r\sigma$ and for all $i \in \{1, \ldots, n\}$, $t_i \Eval{T} s_i$ and $s_i =_\Ea \ell_i\sigma$. By \Cref{M:EaE}, we know that $\Ea \subseteq E$. Hence $s_i =_E \ell_i\sigma$ for all $i \in \{1, \ldots,n\}$. Moreover, by \Cref{M:eqcomplete} of \Cref{def:mimics}, we know that $f(\ell_1,\ldots,\ell_n) =_E r$. Thus $f(s_1,\ldots,s_n) =_E f(\ell_1\sigma,\ldots,\ell_n\sigma) =_E r\sigma$. Applying our inductive hypothesis on $t_1,\ldots,t_n$ gives us $\removeeval{t_i} =_E s_i$ for all $i \in \{1,\ldots,n\}$. This allows us to deduce that $\removeeval{t} = f(\removeeval{t_1},\ldots,\removeeval{t_n}) =_E r\sigma = s$.
\end{proof}


\thequalitymodulo*

\begin{proof}
Assume that $t =_E s$. By \Cref{lem:norm-addition} of \Cref{lem:norm}, we know that there exists $u$ such that $u =_E t$ and $\nf{T}{E}{\{ u\}}$ (by taking $\M = \emptyset$). Hence, by \Cref{lem:norm-var} of \Cref{lem:norm}, we deduce that $\nf{T}{E}{\{ u\} \cup \vars{t,s}}$. Applying \Cref{lem:eval}, we obtain that there exist $t',s'$ such that $\symbeval{t} \Eval{T} t'$ and $\symbeval{s} \Eval{T} s'$ and $t' =_\Ea u =_\Ea s'$. 

Assume now that there exist $t',s'$ such that $\symbeval{t} \Eval{T} t'$, $\symbeval{s} \Eval{T} s'$ and $t' =_\Ea s'$. We know from \Cref{M:EaE} that $\Ea \subseteq E$ thus $t' =_E s'$. Moreover, by \Cref{lem:eval_complete}, $t =_E t'$ and $s =_E s'$.
\end{proof}


\subsubsection{Properties on open evaluation}

For the next proofs, we will define the size of an evaluation term $t$, denoted $\sizeEval{t}$, as the number of function symbols from $\symbeval{\F}$ in $t$. We also naturally extend the definition of close evaluation to sequence of ET-term, as follows: $[t_1,\ldots,t_n] \Eval{T} [s_1,\ldots,s_n]$ if $t_i \Eval{T} s_i$ for all $i \in \{1, \ldots,n\}$. Given two sequences of terms $L,L'$, we also write $L =_E L'$ when $L = [t_1,\ldots,t_n]$, $L' = [s_1,\ldots,s_n]$ and for all $i \in \{1, \ldots,n\}$, $t_i =_E s_i$.


\begin{lemma}
\label{lem:open evaluation variables}
Let $L$ be a sequence of TE-terms. If $L \OEval{T} (L_t,\sigma)$ then $\dom{\sigma} \cap \vars{L_t} = \emptyset$.
\end{lemma}

\begin{proof}
We prove the result by induction on $\multiset{\sizeEval{t} \mid t \in L}$.

\paragraph{Case $L = []$:} In such a case, $L_t = []$ and $\sigma = \emptyset$ hence the result trivially holds.

\paragraph{Case $L = [t]$ with $t \in T(\F,\X \cup \N)$:} In such a case, $L_t = L$ and $\sigma = \emptyset$. Hence the result holds.

\paragraph{Case $L = [\symbeval{f}(t_1,\ldots,t_n)]$:} Let us denote $L' = [t_1,\ldots,t_n]$. We know that there exist  $f(\ell_1,\ldots,\ell_n) \rightarrow r \in \R$ and $L'_t$, $\sigma',\sigma_u$ such that, denoting $L_\ell = [\ell_1,\ldots,\ell_n]$, $L' \OEval{T} (L'_t,\sigma')$, and $\sigma_u \in \mgu{\Ea}{L'_t,L_\ell}$, and $L_t = [r\sigma_u]$ and $\sigma = \sigma'\sigma_u$. 

Let $x \in \dom{\sigma} \cap \vars{L_t}$. Hence either $x \in \dom{\sigma_u}$ or $x \in \dom{\sigma'}$. In the former case, since $\sigma_u$ is acyclic then $x \in \dom{\sigma_u}$ implies that $x \not\in \vars{r\sigma_u}$ which is in contradiction with $x\in \vars{L_t}$. Hence $x \in \dom{\sigma'}$. By our inductive hypothesis, it implies that $x \not\in \vars{L'_t}$. Since the variables of $f(\ell_1,\ldots,\ell_n)\rightarrow r \in \R$ can always be renaming and since the variables of a most general unifier $\sigma_u \in \mgu{\Ea}{L'_t,L_\ell}$ can also be renamed, we have that $x \not\in \vars{\sigma_u}$ and $x\not\in L_\ell$. Hence $x\not\in \vars{L'_t\sigma_u}$ and $x \not\in L_\ell\sigma_u$. But $\vars{r\sigma_u} \subseteq \vars{L_\ell\sigma_u}$. Hence $x \not\in \vars{L_t}$ which is a contradiction with $x \in \dom{\sigma} \cap \vars{L_t}$. This conclude the proof that $\dom{\sigma} \cap \vars{L_t} = \emptyset$.
\end{proof}


\begin{lemma}
\label{lem:open evaluation remove}
For all sequences of TE-terms $L$, if $L \OEval{T} (L_t,\sigma)$ then $\removeeval{L}\sigma =_E L_t$.
\end{lemma}

\begin{proof}
We prove the result by induction on $\multiset{\sizeEval{t} \mid t \in L}$.

\paragraph{Case $L = []$:} In such a case, $L_t = []$ hence the result trivially holds.

\paragraph{Case $L = [t]$ with $t \in T(\F,\X \cup \N)$:} In such a case, $\removeeval{L_t} = L_t = L$ and $\sigma = \emptyset$. Hence the result holds.

\paragraph{Case $L = [\symbeval{f}(t_1,\ldots,t_n)]$:} Let us denote $L' = [t_1,\ldots,t_n]$. We know that there exist  $f(\ell_1,\ldots,\ell_n) \rightarrow r \in \R$ and $L'_t$, $\sigma',\sigma_u$ such that, denoting $L_\ell = [\ell_1,\ldots,\ell_n]$, $L' \OEval{T} (L'_t,\sigma')$, and $\sigma_u \in \mgu{\Ea}{L'_t,L_\ell}$, and $L_t = [r\sigma_u]$ and $\sigma = \sigma'\sigma_u$. By inductive hypothesis on $L'$, we know that $\removeeval{L'}\sigma' =_E L'_t$. Hence, $\removeeval{L}\sigma'\sigma_u =_E L'_t\sigma_u =_E L_\ell\sigma_u$ as $\Ea \subseteq E$ by \Cref{M:EaE}. Note that the variables of $\ell_1,\ldots,\ell_n,r$ are fresh, i.e. $\vars{\ell_1,\ldots,\ell_n,r} \cap \vars{\sigma} = \emptyset$. Hence $L_\ell\sigma\sigma_u = L_\ell\sigma_u$ and $r\sigma\sigma_u = r\sigma_u$. This allow us to deduce that $f(t_1,\ldots,t_n)\sigma =_E f(\ell_1,\ldots,\ell_n)\sigma$. By \Cref{M:eqcomplete}, $f(\ell_1,\ldots,\ell_n) =_E r$ and so $f(\ell_1,\ldots,\ell_n)\sigma = r\sigma$ which allows us to conclude.

\paragraph{Case $L = t \cdot L'$:} In such a case, we have $[t] \OEval{T} ([s],\sigma')$ and $L'\sigma' \OEval{T} (L'_t,\sigma'')$ with $L_t = s\sigma'' \cdot L'_t$ and $\sigma = \sigma'\sigma''$. Note that $\removeeval{L'\sigma'} = \removeeval{L'}\sigma'$. By inductive hypothesis on both $[t]$ and $L'\sigma'$, we deduce that $\removeeval{t}\sigma' =_E s$ and $\removeeval{L'}\sigma'\sigma'' =_E L'_t$. Hence $(\removeeval{t} \cdot \removeeval{L'}) =_E s\sigma'' \cdot L'_t$ and so $\removeeval{L}\sigma =_E L_t$.
\end{proof}


\begin{lemma}
\label{lem:evaluation Ea}
Let $t$ be a TE-term and two substitutions $\sigma_1,\sigma_2$ such that $t\sigma_1 =_\Ea t\sigma_2$. For all terms $s_1$, if $t\sigma_1 \Eval{T} s_1$ then there exists a term $s_2$ such that $t\sigma_2 \Eval{T} s_2$ and $s_1 =_\Ea s_2$.
\end{lemma}

\begin{proof}
The proof is by induction on $t$.

\paragraph{Case $t \in \T(\F,\X \cup \N)$:} In that case, $t\sigma_1 \Eval{T} t\sigma_1$ and $t\sigma_2 \Eval{T} t\sigma_2$. Thus the result holds.

\paragraph{Case $t = \symbeval{f}(t_1,\ldots,t_n)$:} In such a case, by \Cref{def:evaluation}, there exists $f(\ell_1,\ldots,\ell_n) \rightarrow r \in \R$ and $s_1,\ldots,s_n$ terms and a substitutions $\sigma$ such that for all $i \in \{1,\ldots,n\}$, $t_i\sigma_1 \Eval{T} s_i$ and $s_i =_\Ea \ell_i\sigma$. By inductive hypothesis on $t_1,\ldots,t_n$, there exist terms $s'_1,\ldots,s'_n$ such that for all $i \in \{1,\ldots,n\}$, $t_i\sigma_2 \Eval{T} s'_i$ and $s_i = \Ea s'_i$ and so $s'_i = \Ea \ell_i\sigma$. This leads to $t\sigma_2 \Eval{T} r\sigma$.
\end{proof}


\begin{lemma}
\label{lem:open-evaluation}
Let $L$ be a sequence of TE-terms. For all substitutions $\sigma$, for all sets of variables $\V$, if $L\sigma \Eval{T} L_s$ and $\dom{\sigma} \cup \vars{L,L_s} \subseteq \V$ then there exist a sequence of term $L_t$, two substitutions $\alpha,\beta$ such that $L \OEval{T} (L_t,\alpha)$, and $L_t\beta =_\Ea L_s$ and $\sigma =_\Ea (\alpha\beta)_{|\dom{\sigma}}$ and $\dom{\alpha\beta} \cap \V = \dom{\sigma}$.
\end{lemma}

\begin{proof}
We prove this result by induction on $\multiset{\sizeEval{t} \mid t \in L}$.

\paragraph{Case $L = []$:} In such case, $[] \Eval{T} []$ and $[] \OEval{T} ([],\emptyset)$. So the result holds directly with $L_t = L$, $\alpha = \emptyset$ and $\beta = \sigma$.

\paragraph{Case $L = [t]$ with $t \in \T(\F,\X \cup \N)$:} In such a case, $[t\sigma] \Eval{T} [t\sigma]$ and $[t] \OEval{T} ([t],\emptyset)$. Once again the result directly holds with $L_t = L$, $\alpha = \emptyset$ and $\beta = \sigma$.

\paragraph{Case $L = [\symbeval{f}(t_1,\ldots,t_n)]$:} Let us denote $t = \symbeval{f}(t_1,\ldots,t_n)$. By definition, $L\sigma \Eval{T} L_s$ implies that there exist a term $s$ such that $L_s = [s]$ and $t\sigma \Eval{T} s$. Hence there exist $f(\ell_1,\ldots,\ell_n) \rightarrow r \in \R$ and a substitution $\theta$ and terms $s_1,\ldots,s_n$ such that $r\theta = s$ and for all $i \in \{1,\ldots,n\}$, $\ell_i\theta =_\Ea s_i$ and $t_i\sigma \Eval{T} s_i$. 
Let $L' = [t_1,\ldots,t_n]$ and $L'_s = [s_1,\ldots,s_n]$. By inductive hypothesis on $L'$ with $\sigma$ and $\V' = \V \cup \vars{s_1,\ldots,s_n}$, there exists a sequence of terms $L'_t$ and two substitutions $\alpha',\beta'$ such that $L' \OEval{T} (L'_t,\alpha')$, and $L'_t\beta' =_\Ea L'_s$, $\sigma =_\Ea (\alpha'\beta')_{|\dom{\sigma}}$ and $\dom{\alpha'\beta'} \cap \V' = \dom{\sigma}$. 
Since the variables of $\ell_1,\ldots,\ell_m,r$ can always be renamed, we can w.l.o.g. assume that $\dom{\theta} \cap \vars{\alpha'\beta'} = \emptyset$ and $\dom{\theta} \cap \V' = \emptyset$. 

Since $L'_t\beta' =_\Ea L'_s =_\Ea [\ell_1\theta,\ldots,\ell_n\theta]$, we can define $L_\ell = [\ell_1,\ldots,\ell_n]$  and $\gamma = \beta'\theta$ to have $L'_t\beta' = L'_t\gamma$ and $L_\ell\theta = L_\ell\gamma$ and $r\gamma = r\theta$. Hence $\gamma$ is a $\Ea$-unifier of $L'_t$ and $L_\ell$. Hence, there exist $\sigma_u \in \mgu{\Ea}{L'_t,L_\ell}$ and $\theta_u$ such that for all $x \in \vars{L'_t,L_\ell}$, $x\gamma =_\Ea x\sigma_u\theta_u$. Note that using renaming of variables in $\sigma_u$ introduced by the computation of most general unifier, we can always have $\gamma = (\sigma_u\theta_u)_{\dom{\gamma}}$ and the variables of $\dom{\sigma_u\theta_u} \setminus \dom{\gamma}$ are all fresh.
Hence by \Cref{def:open evaluation}, $L \OEval{T} ([r\sigma_u],\alpha'\sigma_u)$. Let us define $\beta = \theta_u$ and $L_t = [r\sigma_u]$. By \Cref{S:std}, $\vars{r} \subseteq \vars{L_\ell}$, hence $r\sigma_u\theta_u = r\gamma$. It implies that $L_t\beta = [r\sigma_u\theta_u] =_\Ea [r\gamma] = [r\theta] = [s]
 = L_s$. Moreover, for all $x\in \dom{\sigma}$, $x\sigma =_\Ea x\alpha'\beta'$. 
As $\dom{\theta} \cap \vars{\alpha'\beta'} = \emptyset$, $x\alpha'\beta' = x\alpha'\beta'\theta = x\alpha'\gamma =_\Ea x\alpha'\sigma_u\theta_u$.

Hence, defining $\alpha = \alpha'\sigma_u$, we obtain $x\sigma =_\Ea x\alpha\beta$. Take $x \in \dom{\alpha\beta} \cap \V$. Note that $\dom{\alpha\beta} = \dom{\alpha'\sigma_u\theta_u}$. But $\dom{\sigma_u\theta_u}$ is $\dom{\gamma}$ with some fresh variables introduced in the computation of the most general unifier. However, $\dom{\gamma} = \dom{\beta'} \cup \dom{\theta}$ and we already know that $\dom{\theta} \cap \V' = \emptyset$. Since $\V \subseteq \V'$, we deduce that $\dom{\theta} \cap \V = \emptyset$. Therefore, $\dom{\alpha\beta} \cap \V = (\dom{\alpha'} \cup \dom{\beta'}) \cap \V = \dom{\sigma}$.

 \paragraph{Case $L = t \cdot L'$:} By definition $L\sigma \Eval{T} L_s$ implies that $t\sigma \Eval{T} s$ and $L'\sigma \Eval{T} L'_s$ and $L_s = s \cdot L'_s$ for some term $s$ and sequence of terms $L'_s$. By inductive hypothesis, there exists a term $t'$ and two substitutions $\alpha',\beta'$ such that $[t] \OEval{T} ([t'],\alpha')$ and $t'\beta' =_\Ea s$ and $\sigma =_\Ea (\alpha'\beta')_{|\dom{\sigma}}$ and $\dom{\alpha'\beta'} \cap \V = \dom{\sigma}$.

 Since $\vars{L} \subseteq \V$, we deduce that $L'\sigma =_\Ea L'\alpha'\beta'$. By \Cref{lem:evaluation Ea}, we deduce that there exists a sequence of terms $L''_s$ such that $L'\alpha'\beta' \Eval{T} L''_s$ and $L''_s =_\Ea L'_s$. By applying our inductive hypothesis on $(L'\alpha')$ and $\beta'$ and $\V' = \V \cup \vars{L''_s} \cup \dom{\alpha'\beta'} \cup \vars{t'}$, we obtain that there exists a sequence of terms $L''_t$, two substitutions $\alpha'',\beta''$ such that $L'\alpha' \OEval{T} (L''_t,\alpha'')$ and $L''_t\beta'' =_\Ea L''_s$ and $\beta' =_\Ea (\alpha''\beta'')_{|\dom{\beta'}}$ and $\dom{\alpha''\beta''} \cap \V' = \dom{\beta'}$.

 By \Cref{def:open evaluation}, $L \OEval{T} (t'\alpha'' \cdot L''_t,\alpha'\alpha'')$. Let us denote $\alpha = \alpha'\alpha''$ and $\beta = \beta''$ and $L_t = t'\alpha'' \cdot L''_t$. Note that $\vars{t'} \subseteq \V'$. Hence $t'\alpha''\beta'' =_\Ea t'\beta' =_\Ea s$. Moreover, $L''_t\beta'' =_\Ea L''_s =_\Ea L'_s$. Hence, $L_t\beta =_\Ea L_s$. Let us show $\sigma =_\Ea (\alpha\beta)_{\dom{\sigma}}$. Let $x\in \dom{\sigma}$. We know that $x\sigma =_\Ea x\alpha'\beta'$. As $\beta' =_\Ea (\alpha''\beta'')_{|\dom{\beta'}}$, we obtain $x\alpha'\beta' =_\Ea x\alpha'\alpha''\beta'' = x\alpha\beta$. It also implies that $\dom{\sigma} \subseteq \dom{\alpha\beta}$. As $\dom{\sigma} \subseteq \V$, we have $\dom{\sigma} \subseteq \dom{\alpha\beta} \cap \V$. Take $x \in \dom{\alpha\beta} \cap \V$. Hence $x \in \dom{\alpha'\alpha''\beta''} = \dom{\alpha'} \cup \dom{\alpha''\beta''}$. If $x \in \dom{\alpha''\beta''}$ then, since $\dom{\alpha''\beta''} \cap \V' = \dom{\beta'}$ and $\V \subseteq \V'$, we deduce that $x \in \dom{\beta'}$. Therefore $x \in \dom{\alpha'\alpha''\beta''}$ implies $x \in \dom{\alpha'\beta'}$. As we proved that $\dom{\alpha'\beta'} \cap \V = \dom{\sigma}$, we conclude that $x\in \dom{\sigma}$ and so $\dom{\alpha\beta} \cap \V = \dom{\sigma}$.
\end{proof}


\begin{lemma}
\label{lem:open evaluation variant}
For all sequences of terms $L,L_s$, for all substitutions $\sigma$, if $L\sigma =_E L_s$ and $\nf{T}{E}{L_s}$ then there exist $\symbeval{L} \OEval{T} (L_t,\alpha)$ and a substitution $\beta$ such that $L_t\beta =_\Ea L_s$ and $\sigma =_E (\alpha\beta)_{\dom{\sigma}}$ and $\vars{L} \cap \dom{\alpha\beta} \subseteq \dom{\sigma}$.
\end{lemma}

\begin{proof}
By \Cref{lem:norm-addition} of \Cref{lem:norm}, there exists a substitution $\sigma'$ such that $\nf{T}{E}{L_s \cup \img{\sigma'}}$ and $\sigma =_E \sigma'$. Moreover, by \Cref{lem:norm-var} of \Cref{lem:norm}, we deduce that $\nf{T}{E}{L_s \cup \img{\sigma'} \cup \{ x\sigma' \mid x \in \vars{L}\}}$. By \Cref{lem:eval}, $\symbeval{L}\sigma' \Eval{T} L'_s$ and $L'_s =_\Ea L_s$ for some $L'_s$. 
Take $\V = \dom{\sigma'} \cup \vars{L,L'_s}$. By \Cref{lem:open-evaluation}, there exists a sequence of terms $L_t$, two substitutions $\alpha,\beta$ such that $\symbeval{L} \OEval{T} (L_t,\alpha)$ and $L'_s =_\Ea L_t\beta$ and $\sigma' =_\Ea (\alpha\beta)_{|\dom{\sigma'}}$ and $\dom{\alpha\beta} \cap \V = \dom{\sigma'}$. Since $\Ea \subseteq E$ by \Cref{M:EaE}, we conclude that $L_t\beta =_\Ea L_s$ and $\sigma =_E \sigma' =_\Ea (\alpha\beta)_{|\dom{\sigma'}} = (\alpha\beta)_{|\dom{\sigma}}$. Since $\vars{L} \subseteq \V$, we deduce from $\dom{\alpha\beta} \cap \V = \dom{\sigma'}$ that $\dom{\alpha\beta} \cap \vars{L} \subseteq \dom{\sigma'}$.
\end{proof}


\subsubsection{Finite variant property and most general \texorpdfstring{$E$}{E}-unifiers}


\thvariant*

\begin{proof}
Let $t$ be a term. Let us define $\mathcal{S} = \{ t' \mid [\symbeval{t}] \OEval{T} (t',\alpha)\}$. We can show that $\mathcal{S}$ is a complete set of $E$-variants modulo $\Ea$ of $t$. Let $t' \in \mathcal{S}$. Hence $[\symbeval{t}] \OEval{T} (t',\alpha)$. By \Cref{lem:open evaluation remove}, $t\alpha =_E t'$ meaning that $t'$ is a $E$-variant of $t$. Hence $\mathcal{S}$ is a set of $E$-variant of $t$. 

Let us show the second part of the theorem. Let $\sigma$ be a substitution closing for $t$. Let $s = \minOrd[>]{E}(t\sigma)$. Since we have $\minOrd[>]{E}(t\sigma) = \minOrd[>]{E}(\minOrd[>]{E}(t\sigma)) $, we have by definition that $\nf{T}{E}{\{s\}}$ and $s =_E t\sigma$ (using the empty substitution since $t\sigma$ is ground).
By \Cref{lem:open evaluation variant}, we deduce that there exists $[\symbeval{t}] \OEval{T} ([t'],\alpha)$ and a substitution $\beta$ such that $t'\beta =_\Ea s$. Hence, $t' \in \mathcal{S}$ and $t'\beta =_\Ea s = \minOrd[>]{E}(t\sigma)$.
\end{proof}


\thmostgeneralunifiers*

\begin{proof}
Let $\sigma$ be a $E$-unifier of $t$ and $s$, i.e. $t\sigma =_E s\sigma$. By \Cref{lem:norm-addition} of \Cref{lem:norm}, there exists a term $u$ such that $\nf{T}{E}{\{u\}}$ and $t\sigma =_E u =_E s\sigma$. Hence defining $L_s = [u,u]$, we have $\nf{T}{E}{L_s}$. By \Cref{lem:open evaluation variant}, there exists $[\symbeval{t},\symbeval{s}] \OEval{T} ([t',s'],\alpha)$ and a substitution $\beta$ such that $[t',s']\beta =_\Ea L_s$ and $\sigma =_E (\alpha\beta)_{\dom{\sigma}}$ and $\vars{t,s} \cap \dom{\alpha\beta} \subseteq \dom{\sigma}$. Thus, $t'\beta =_\Ea u =_\Ea s'\beta$. Since $\beta$ is a $\Ea$-unifier of $t'$ and $s'$, there exists $\sigma_u \in \mgu{\Ea}{t',s'}$ and $\theta$ such that for all $\beta =_\Ea (\sigma_u\theta)_{\dom{\beta}}$ and the variables of $\dom{\sigma_u\theta}\setminus \dom{\beta}$ are fresh variables created in the computation of the most general unifier. Thus, $\sigma =_E (\alpha\sigma_u\theta)_{\dom{\sigma}}$. 
By denoting $\sigma' = \alpha\sigma_u$ and $\theta' = \theta$, we thus have $\sigma' \in \mathcal{S}$ such that $\sigma =_E (\sigma'\theta')_{\dom{\sigma}}$.

It remains to show that all substitutions in $\mathcal{S}$ are $E$-unifier of $t,s$.
By \Cref{lem:open evaluation remove}, $t\alpha =_E t'$ and $s\alpha =_E s'$. Since $\sigma_u \in \mgu{E}{t',s'}$, we have $t'\sigma_u =_E s'\sigma_u$ and so $t\alpha\sigma_u =_E s\alpha\sigma_u$ which means that $\alpha\sigma_u$ is a $E$-unifier of $s,t$.
\end{proof}


\thequivalentdeffinitevariant*

\begin{proof}
Let us define $\Rn = \emptyset$. Hence \Cref{S:subset,S:order} of \Cref{def:extended_signature} and \Cref{M:EaE} of \Cref{def:mimics} directly hold.
Let us build $\R$ as follows:
For all $f \in \F$ of arity $n$, we build the term $u_f = \bin(f(x_1,\ldots,x_n),\bin(x_1,\bin(x_2,\ldots \bin(x_{n-1},x_n))))$. Let us denote by $p$ the position of $f(x_1,\ldots,x_n)$ in $u_f$, i.e. $1$. and $p_1,\ldots,p_n$ the positions of $x_1,\ldots,x_n$ in $u_f$ respectively, i.e. $2\cdot 1$, $2 \cdot 2 \cdot 1$, \ldots

Since $(E,\Ea)$ has the finite variant property then there exists a finite set of $E$-variant $S_f$ such that for all $\sigma$ closing for $u$, there exist a substitution $\gamma$ and $u' \in S_f$ such that $u'\gamma =_\Ea \minOrd[>]{E}(u\sigma)$. 

For each $u' \in S_f$, since $u'$ is a $E$-variant of $u$, we know that there exists $\theta$ such that $u' =_E u\theta$. Recall that $\bin$ is a free binary function. Hence, the positions $p, p_1,\ldots,p_n$ are also positions of $u'$ and $u'_{|p} =_E u_{|p}\theta$ and for all $i \in \{1, \ldots, n\}$, $u'_{|p_i} =_E u_{|p_i}\theta$. Moreover, let us define by $\sigma_{u'}$ the substitution from $\vars{u'_{|p}} \setminus \vars{u'_{|p_1}, \ldots, u'_{|p_n}}$ to $\{ \amin\}$.  
Thus, for each $u' \in S_f$, we define the following rewrite rule that is added to $\R$:
\[
f(u'_{|p_1}, \ldots, u'_{|p_n}) \rightarrow u'_{|p}\sigma_{u'}
\]
Finally, for all $f \in \F$ of arity $n$, we also add to $\R$ the rewrite rule $f(x_1,\ldots,x_n) \rightarrow f(x_1,\ldots,x_n)$. By construction, the property \Cref{S:std} of \Cref{def:extended_signature} is satisfied, implying that $(>,\R,\emptyset,\En,\Ea)$ is a rewrite theory. Notice that when $\En \neq \Ea$, our hypothesis ensures that $\sigma_{u'}$ is empty. 

To prove \Cref{M:eqcomplete} of \Cref{def:mimics}, recall that by definition of $u$, we have $u_{|p} = f(u_{|p_1},\ldots,u_{|p_n})$ and so $u'_{|p} =_E u_{|p}\theta = f(u_{|p_1},\ldots,u_{|p_n})\theta$. As for all $i \in \{1, \ldots, n\}$, $u'_{|p_i} =_E u_{|p_i}\theta$, we conclude that $u'_{|p} =_E f(u'_{|p_1},\ldots,u'_{|p_n})$. Recall that $\dom{\sigma_{u'}} = \vars{u'_{|p}} \setminus \vars{u'_{|p_1}, \ldots, u'_{|p_n}}$. Therefore, $f(u'_{|p_1},\ldots,u'_{|p_n}) = f(u'_{|p_1},\ldots,u'_{|p_n})\sigma_{u'} =_E u'_{|p}\sigma_{u'}$.

It remains to prove \Cref{M:ind1}. Take $f(t_1,\ldots,t_n) =_E t$ with $\nf{T}{E}{\M}$ and $\M = \{ t_1,\ldots,t_n,t\}$. Hence, by \Cref{def:normal form}, there exists $\sigma$ an injective substitution with $\dom{\sigma} = \vars{\M}$, $\img{\sigma} \subseteq \N$, $\names{\M} \cap \img{\sigma} = \emptyset$ and for all $s \in \M$, $s\sigma =_\En \minOrd[>]{E}(s\sigma)$. Take $\alpha = \{x_1 \mapsto t_1, \ldots, x_n \mapsto t_n\}$. Since $\alpha\sigma$ is closing for $u$, there exists a substitution $\gamma$ and $u' \in S_f$ such that $u'\gamma =_\Ea \minOrd[>]{E}(u\alpha\sigma)$. 

If we denote $w = \minOrd[>]{E}(u\alpha\sigma)$ then once again as $\bin$ is a free binary symbol, we know that $p,p_1,\ldots,p_n$ are positions of $w$. Moreover, we can show that $w_{|p} =_\En \minOrd[>]{E}(u_{|p}\alpha\sigma)$ and for all $i \in \{1, \ldots, n\}$, $w_{|p_i} =_\En \minOrd[>]{E}(u_{|p_i}\alpha\sigma)$. Indeed, assume by contradiction that this is not the case, w.l.o.g. $w_{|p} \neq_{\En} \minOrd[>]{E}(u_{|p}\alpha\sigma)$. Hence, $\minOrd[>]{E}(u_{|p}\alpha\sigma) < w_{|p}$ (since $w_{|p} =_E \minOrd[>]{E}(u_{|p}\alpha\sigma)$). However, $>$ is $\En$-compatible which implies that $w[\minOrd[>]{E}(u_{|p}\alpha\sigma)]_{p} < w$. This leads to a contradiction with $w = \minOrd[>]{E}(u\alpha\sigma)$.

To summarize, we have proved that $u'_{|p}\gamma =_\Ea \minOrd[>]{E}(u_{|p}\alpha\sigma)$ and for all $i \in \{1, \ldots, n\}$, $u'_{|p_i}\gamma =_\Ea \minOrd[>]{E}(u_{|p_i}\alpha\sigma)$. Recall that $\nf{T}{E}{\M}$ implies that $t\sigma =_\En \minOrd[>]{E}(t\sigma) =_\En \minOrd[>]{E}(f(t_1,\ldots,t_n)\sigma) =_\En \minOrd[>]{E}(u_{|p}\alpha\sigma)$ and for all $i \in \{1,\ldots,n\}$, $t_i\sigma =_\En \minOrd[>]{E}(t_i\sigma) = \minOrd[>]{E}(u_{|p_i}\alpha\sigma)$. Therefore, we obtain $u'_{|p}\gamma =_\Ea t\sigma$ and for all $i \in \{1, \ldots, n\}$, $u'_{|p_i}\gamma =_\Ea t_i\sigma$.

When $\Ea \neq \En$, we know that $\sigma_{u'}$ is the identity and so the rule $f(u'_{|p_1}, \ldots, u'_{|p_n}) \rightarrow u'_{|p}$ is in the set $\R$. As $\sigma$ is injective with $\names{\M} \cap \img{\sigma} = \emptyset$, we conclude by taking the substitution $\sigma' = \gamma\sigma^{-1}$.

We now look at the case $\En = \Ea$. Recall that we already proved that $u'_{|p} =_E u'_{|p}\sigma_{u'}$. By definition of $\sigma_{u'}$, we know that either $u'_{|p}\gamma =_\En u'_{|p}\sigma_{u'}\gamma$ or $u'_{|p}\gamma > u'_{|p}\sigma_{u'}\gamma$. As $>$ is $\En$-compatible, the latter case would imply that $\minOrd[>]{E}(u_{|p}\alpha\sigma) > u'_{|p}\sigma_{u'}\gamma$ which is a contradiction with the definition of $\minOrd[>]{E}(u_{|p}\alpha\sigma)$. Therefore, $u'_{|p}\gamma =_\En u'_{|p}\sigma_{u'}\gamma$. Once again, as $\sigma$ is injective with $\names{\M} \cap \img{\sigma} = \emptyset$, we conclude by taking the substitution $\sigma' = \gamma\sigma^{-1}$ and the rule $f(u'_{|p_1}, \ldots, u'_{|p_n}) \rightarrow u'_{|p}\sigma_{u'}$ in $\R$.
\end{proof}

\subsection{Proofs of Theorem \ref{th:generation of rewrite theory}}
\label{sec:app-proof-main-theorem}

We already introduced many notions in \Cref{sec:overview}. We need to introduce a few more in order to write the proof of \Cref{th:generation of rewrite theory}.
In this section, we will consider an equational theory $\En$ and an $\En$-strong reduction ordering $>$. Given two terms $u,v$, we will denote by $u \geq v$ when either $u > v$ or $u =_\En v$. In further sections, we will also refer to $\En$ as the rewrite set $\{ \ell \rightarrow r, r \rightarrow \ell \mid (\ell = r) \in \En\}$ when discussing from where a rewrite trace is from.


\paragraph{Rewrite steps}

In the proof, we will consider two types of rewrite steps: syntactic and modulo $\En$. We already introduced the notation $t \rightrwstep{p}{\sigma}{\ell}{r} s$ and $t \leftrwstep{p}{\sigma}{\ell}{r} s$. We augment this notation with an equation theory: We write $t \rightrwstep[E]{p}{\sigma}{\ell}{r} s$ when $p \in \Pos{t}$, $t_{|p} =_E \ell\sigma$ and $s = t[r\sigma]_p$ and $s' =_{E_2} r\sigma$ and $\ell\sigma \geq r\sigma$. Similarly, we write $s \leftrwstep[E]{p}{\sigma}{\ell}{r} t$ when $t \rightrwstep[E]{p}{\sigma}{\ell}{r} s$. In practice, $E$ will either be $\En$ or the emptyset $\emptyset$. When $E = \En$, it will correspond to the rewrite of $t$ by $s$ modulo $\En$. When $E = \emptyset$, the equality becomes syntactic and we omit it, hence yielding $t \rightrwstep{p}{\sigma}{\ell}{r} s$. 

Recall that compare to standard rewriting, we impose an order between the instantiated terms of the rule, \emph{i.e.}, $\ell\sigma \geq r\sigma$. Notice that when $\ell\sigma =_\En r\sigma$, we have $t \rightrwstep[\En]{p}{\sigma}{\ell}{r} s$ iff $t \leftrwstep[\En]{p}{\sigma}{r}{\ell} s$.


\paragraph{Rewrite labels and rewrite traces} In Paragraph~\ref{sec:peak-with-overlapping-positions}, we defined the notion on rewrite labels and rewrite traces. As we extended rewrite steps with an equational theory $E$, we also need to update the definition of rewrite labels by just adding the equational theory $E$ in the tuple: 
A rewrite label is a tuple $\omega = \leftrightrwlabel{p}{\sigma}{\ell}{r}[E]$ with ${\sim} \in \{ \rightarrow, \leftarrow \}$ corresponding to the argument of a rewrite step. The

One of the advantages of introducing rewrite traces is that we can define simple operations that will ease the reading of proofs: Given a position $p$ and a substitution $\sigma$, we define recursively $p \cdot \tr$, $\tr\sigma$ and $\reverse{\tr}$ as follows.
\begin{itemize}
\item $p \cdot \varepsilon = \varepsilon$ and $p \cdot \leftrightrwlabel{q}{\alpha}{\ell}{r}[E]\tr = \leftrightrwlabel{p \cdot q}{\sigma}{\ell}{r}[E] (p\cdot \tr)$
\item $\varepsilon\sigma = \varepsilon$ and $(\leftrightrwlabel{q}{\alpha}{\ell}{r}[E]\tr)\sigma = \leftrightrwlabel{q}{\alpha\sigma}{\ell}{r}[E] \tr\sigma$
\item $\reverse{\varepsilon} = \varepsilon$ and $\reverse{(\rightrwlabel{q}{\alpha}{\ell}{r}[E] \tr)} = \reverse{tr}\leftrwlabel{q}{\alpha}{\ell}{r}[E]$ and $\reverse{(\leftrwlabel{q}{\alpha}{\ell}{r}[E] \tr)} = \reverse{tr}\rightrwlabel{q}{\alpha}{\ell}{r}[E]$
\end{itemize}

We can provide some basic and useful properties:


\begin{lemma}
\label{lem:operation-properties}
Let $u \Leftrightrwstep{\tr} v$. For all substitutions $\sigma$, for all term contexts $C[\_]$ with $C[\_]_{|p} = \_$, we have
$u\sigma \Leftrightrwstep{\tr\sigma} v\sigma$, and $C[u] \Leftrightrwstep{p\cdot \tr} C[v]$, and $v \Leftrightrwstep{\reverse{\tr}} u$. Moreover, if $u \Rightrwstep{\tr} v$ (resp. $u \Leftrwstep{\tr} v$) then $v \Leftrwstep{\reverse{\tr}} u$ (resp. $v \Rightrwstep{\tr} u$).
\end{lemma}

\begin{proof}
Directly follows from the definitions.
\end{proof}


\paragraph{Ranking function} 
The transformations presented in \Cref{sec:peaks-into-valleys} that transform peaks into valleys intuitively reduce the altitude of the mountainous landscape, while sometimes increasing the length of the landscape. However, the transformations presented in \Cref{sec:ordering slopes} not only do not decrease the altitude but also increase the length of the landscape, possibly adding terms at the highest altitude. To prove termination of applications of landscape transformations, we rely on an insight found the previous section: we showed that $p$ was a strict prefix of all positions of rule applications in $\tr_L$. For the rewrite trace $\tr_R$, we only had the guarantee that the positions of rule applications were prefixes of $p$ (not necessarily strict), but the terms rewritten by the rule were intuitively at a strictly lower altitude. 
Therefore, to show termination of our landscape transformations, we will compare both the positions at which a rule is applied, as well as the terms rewritten. 


\begin{definition}
We define the \emph{ranking function on rewrite labels}, denoted $\measureF$, by $\measure{\omega} = (p,r\sigma,\ell\sigma)$ when $\omega$ is a right rewrite label $\rightrwlabel{p}{\sigma}{\ell}{r}$ or a left rewrite label $\leftrwlabel{p}{\sigma}{\ell}{r}$.

We extend the ranking function to rewrite traces by defining $\measure{\tr} = \multiset{ \measure{\omega} \mid \omega \in \tr}$, that is the multiset of the ranking values of all rewrite labels in $\tr$.
\end{definition}


To compare two instances of the ranking function, we will use the lexicographic order where terms are ordered by $<$ the $\En$-strong reduction ordering, and where positions are compared through the opposite of the natural prefix order, that is $p_1$ will be ``strictly smaller'' than $p_2$ when $p_2$ is a strict prefix of $p_1$. Intuitively, the closer to $\varepsilon$ a position is, the more rewriting is possible on a term. However, this order on positions is not well founded as there is an infinite sequence of ``strictly smaller'' positions than $p_1$. We therefore restrict ourselves to positions contained in a finite set of positions $\Pset$, that is for $p_1$ and $p_2$ to be ordered, they both need to be part of $\Pset$. As $\Pset$ is finite, the order, that we denote $\lessRwLbl{\Pset}$, becomes well-founded.

The sets of position $\Pset$ will be generated from all the terms in the initial rewrite trace that transformed an equality modulo $E$ into a mountainous landscape, as shown in \Cref{sec:mountainous landscape}. In particular, if we gather all terms in this initial landscape into a set $\mathcal{S}$, we can show that all terms in any transformed landscape are necessarily smaller (w.r.t. $<$ the $\En$-strong order) than a term in $\mathcal{S}$. As $<$ is well-founded and $\En$ has finite equivalence classes, the set of terms smaller than terms in $\mathcal{S}$ is therefore finite. Hence, we generate $\Pset$ as the set of all positions in all terms smaller than terms in $\mathcal{S}$. 

By abuse of notation, given two rewrite traces $\tr,\tr'$, we also denote $\measure{\tr'} \lessRwLbl{\Pset} \measure{\tr}$ when using the multiset order induced by the order $\lessRwLbl{\Pset}$ on instances of the ranking function on rewrite labels. We show that for any landscape transformation converting a rewrite trace $\tr$ into a new rewrite trace $\tr'$, we have $\measure{\tr'} \lessRwLbl{\Pset} \measure{\tr}$, ensuring that the application of landscape transformations necessarily terminates.

\begin{remark*}
The termination of applications of landscape transformations does not imply that \GenExtended{} terminates.
\end{remark*}

The rewrite labels with equation theory does not impact their value w.r.t. the ranking function, that is $\measure{\leftrightrwlabel{p}{\sigma}{\ell}{r}[E]} = \measure{\leftrightrwlabel{p}{\sigma}{\ell}{r}}$. Only the inverse operation on rewrite trace preserves the ranking function:


\begin{lemma}
\label{lem:reverse-measure}
For all rewrite trace $\tr$, $\measure{\reverse{\tr}} = \measure{\tr}$.
\end{lemma}

\begin{proof}
By induction on $\tr$ and by noticing that $\measure{\rightrwlabel{q}{\alpha}{\ell}{r}[E]} = \measure{\leftrwlabel{q}{\alpha}{\ell}{r}[E]}$.
\end{proof}

Let $\Pset$ be a set of positions. To compare the measurements, we denote by $\lessRwLbl{\Pset}$ the strict ordering on rewrite label measurement w.r.t. $\Pset$, defined as follows: $(p_1,u_1,v_1) \lessRwLbl{\Pset} (p_2,u_2,v_2)$ when $p_1,p_2 \in \Pset$, and either
\begin{itemize}
\item $p_2$ is strict prefix of $p_1$, or
\item $p_1 = p_2$ and either $u_1 < u_2$ or ($u_1 =_\En u_2$ and $v_1 < v_2$)
\end{itemize}
We also define equality between measurement as $(p_1,u_1,v_1) \eqRwLbl{\Pset} (p_2,u_2,v_2)$ when $p_1 = p_2 \in \Pset$, $u_1 =_\En u_2$ and $v_1 =_\En v_2$. We naturally define $\leqRwLbl{\Pset}$ as ${\lessRwLbl{\Pset}} \cup {\eqRwLbl{\Pset}}$.

Notice that in the definition of $\lessRwLbl{\Pset}$, the position $p_1$ is ``smaller'' than $p_2$ when $p_2$ is a strict prefix of $p_1$, which is opposite to the natural prefix ordering. Intuitively, the closer to $\varepsilon$ a position is, the more rewriting is possible on a term. However, as such ordering is not well-founded, we additionally require the compared positions to be part of the set $\Pset$. When $\Pset$ is finite, we retrieve the well-foundedness of the ordering, as shown below.

\begin{lemma}
\label{lem:well-founded}
For all finite sets of positions $\Pset$, the ordering $\lessRwLbl{\Pset}$ is well-founded
\end{lemma}

\begin{proof}
The ordering $\lessRwLbl{\Pset}$ is in fact a lexicographic ordering on $(p,u,v)$. To compare the terms $u$ and $v$, we use the ordering $<$ that is a $\En$-strong reduction order, meaning that it is well founded. To compare the position $p$, we use the negation of the prefix ordering, i.e. \emph{$p_1$ is strictly smaller than $p_2$} when $p_2$ is a strict prefix of $p_1$. This is not well-founded in itself but as we require that both $p_1$ and $p_2$ are positions in $\Pset$, which is finite and a parameter of the ordering, we retrieve the well-foundedness of the ordering. 
\end{proof}

By abuse of notation, given two rewrite traces $\tr,\tr'$, we also denote $\measure{\tr} \lessRwLbl{\Pset} \measure{\tr'}$ when using the multiset ordering induced by $\lessRwLbl{\Pset}$ on ordering labels. 

In additional to the main ranking function $\measure{\cdot}$ on rewrite traces, we consider a second ranking function defined on a rewrite traces and the terms it rewrites, that is $u \Leftrightrwstep{\tr} v$, denoted $\measureTerm{u,\tr,v}$, as the multisets of the terms $\multiset{t_0,\ldots,t_n}$ where 
\[
u = t_0 \Leftrightrwstep{\omega_1} t_1 \Leftrightrwstep{\omega_2} \ldots \Leftrightrwstep{\omega_n} t_n = v     
\]
with $\tr = \omega_1\ldots\omega_n$ and the $\omega_i$s are rewrite labels. Two instances of this ranking function are compared using the multiset order induced by the $\En$-strong reduction order $>$.


\paragraph{Orderly rewriting of terms} 

The sets of position $\Pset$ used in the ordering above will be generated from all the terms in the initial rewrite trace, by relying on our $\En$-strong reduction order $>$ (and it's associated non-strict $\geq$ relation). Given a set of terms $\M$, we say that \emph{$\M$ is greater than a term $t$} (or \emph{$t$ smaller than $\M$}) when there exists $s \in \M$ such that $s \geq t$. Similarly, we say that $\M$ is greater than $u \Leftrightrwstep{\tr} v$ when for all $\tr'$ of $\tr$, if $u \rwLeftrightStep{\tr'} w$ then $\M$ is greater than $w$.

Denoting $\Pset(\M)$ the set of positions defined as $\{ p \in \Pos{t} \mid s \in \M, s \geq t \}$, we obtain the following property.


\begin{lemma}
\label{lem:finite-equivalence-classes}
If $\En$ has finite equivalence classes then for all finite set of ground terms $\M$, the set $\Pset(\M)$ is finite.
\end{lemma}

\begin{proof}
Since $\M$ is ground and finite and since $>$ is well founded then the number of $\En$-equivalence classes of ground terms smaller than $\M$ is finite. As $\En$ having finite equivalence classes implies, we directly conclude. 
\end{proof}

\paragraph{Saturated sets of rewrite rules}

Our generic procedure and its optimisations basically keep merging $\En$-overlapping rewrite rules until a fixpoint is reached, that is until the set of rewrite rules is saturated. We define below the properties satisfied by such saturated set.


\begin{definition}
\label{def:saturated}
Let $\R,\Rn$ be two sets of rewrites rule. We say that $\R$ is \emph{saturated w.r.t. $\Rn$} when:
\begin{itemize}
\item $\R \cup \Rn \normstep[\Rn,\En]^* \R$
\item for all $(\ell_1 \rightarrow r_1), (\ell_2 \rightarrow r_2) \in \R$, for all positions $p$, $(\overlapseteq{r_1 \rightarrow \ell_1}{p}{\ell_2 \rightarrow r_2}) \cup \R \cup \Rn \normstep[\Rn,\En]^* \R$
\item for all $(\ell_1 \rightarrow r_1) \in \R \cup \En$, for all $(\ell_2 \rightarrow r_2) \in \R$, for all positions $p$, $(\overlapset{\ell_1 \rightarrow r_1}{p}{\ell_2 \rightarrow r_2}) \cup \R \cup \Rn \normstep[\Rn,\En]^* \R$
\item for all $(\ell \rightarrow r) \in \R$, $\{ r \rightarrow \ell\} \cup \R \cup \Rn \normstep[\Rn,\En]^* \R$
\end{itemize}
\end{definition}


\subsubsection{Normalisation rules}

Recall that $>$ is a $\En$-strong reduction order compatible with $\Rn$. Before we state stating the correctness of our normalisation rules, we show that if $u > v$ then $\vars{v} \subseteq \vars{u}$. Indeed, assume that $u$ contains a variable $x$ that is not a variable of $v$. Thus $v > u = C[x,\ldots,x]$. As $>$ is stable by application of substitutions, we deduce that $v > C[v,\ldots,v] > C[C[v,\ldots,v],\ldots,C[v,\ldots,v]] > \ldots$ leading to a contradiction with $>$ being well founded.

\begin{lemma}[Rule \RNormR]
\label{lem:normalisation-normR}
Let $u \rightrwstep{p}{\sigma}{\ell}{r} v$ be a ground rewrite step and $\tr_0 = \rightrwlabel{p}{\sigma}{\ell}{r}$. Let $\M$ a set of terms greater than $u \Leftrightrwstep{\tr_0} v$.
If $r =_\En t \rightrwstep{p'}{\sigma'}{\ell'}{r'} s$ and $\ell' > r'$ then there exists a syntactic rewrite trace $\tr_R$ from $\En$ such that:
\[
    u \rightrwstep{p}{\sigma}{\ell}{s} u[s\sigma]_p \leftrwstep{p\cdot p'}{\sigma'\sigma}{\ell'}{r'} u[t\sigma]_p \Leftrwstep{p\cdot \tr_R\sigma}  v
\]
and denoting $\tr_1 = \rightrwlabel{p}{\sigma}{\ell}{s} \leftrwlabel{p\cdot p'}{\sigma'\sigma}{\ell'}{r'} (p\cdot \tr_R\sigma)$, we deduce that $\tr_1$ is a ground syntactic rewrite trace and:
\begin{enumerate}[label=\textbf{P\arabic*}]
    \item $u \rwLeftrightStep{\tr_1} v$, and \label{item:normalisation-normR-Leftright}
    \item $\M$ is greater than $u \rwLeftrightStep{\tr_1} v$, and \label{item:normalisation-normR-derive}
    \item for all $\omega \in \tr_1$, $\measure{\omega} \lessRwLbl{\Pset(\M)} (p,r\sigma,\ell\sigma)$ or ($\measure{\omega} \eqRwLbl{\Pset(\M)} (p,r\sigma,\ell\sigma)$ and $\omega\in p\cdot\tr_R\sigma$).\label{item:normalisation-normR-measure}
    \end{enumerate}
\end{lemma}

\begin{proof}
By definition, $u \rightrwstep{p}{\sigma}{\ell}{r} v$ implies $v = u[r\sigma]_p$, and $u_{|p} = \ell\sigma$ and $\ell\sigma \geq r\sigma$. 
By definition $t \rightrwstep{p'}{\sigma'}{\ell'}{r'} s$, we deduce that $t_{|p'} = \ell'\sigma'$ and $s = t[r'\sigma']_{p'}$. Since $r =_\En t$, there exists $\tr_R$ a syntactic rewrite trace from $\En$ such that $t \Leftrwstep{\tr_R} r$. Thus, $t\sigma \Leftrwstep{\tr_R\sigma} r\sigma$, implying $u[t\sigma]_p \Leftrwstep{p\cdot\tr_R\sigma} v$. Since $s = t[r'\sigma']_{p'}$ and $t_{|p'} = \ell\sigma'$, we have $u[s\sigma]_p \leftrwstep{p\cdot p'}{\sigma'\sigma}{\ell'}{r'} u[t\sigma]_p$. Recall that $(\ell' \rightarrow r') \in \Rn$ meaning that $\vars{r'} \subseteq \vars{\ell'}$.

Recall that $\ell' > r'$. Hence, $t = t[\ell'\sigma']_{p'} > t[r'\sigma']_{p'} = s$. As $>$ is $\En$-compatible and $r =_\En t$, we deduce $r > s$. As $\ell\sigma \geq r\sigma$, we deduce that $\ell\sigma > s\sigma$. Thus, $u \rightrwstep{p}{\sigma}{\ell}{s} u[s\sigma]_p$ and so \Cref{item:normalisation-normR-Leftright} holds. By construction of $\tr_1$, as $\M$ is greater then $u \Leftrightrwstep{\tr_0} v$, we directly obtain that $\M$ is greater than $u \Leftrightrwstep{\tr_1} v$, i.e. \Cref{item:normalisation-normR-derive}.

Finally, we already showed that $r\sigma > s\sigma$ thus $(p,s\sigma,\ell\sigma) \lessRwLbl{\Pset(\M)} (p,r\sigma,\ell\sigma)$. Additionally, if $p' \neq \varepsilon$ then we trivially have that $(p\cdot p',r'\sigma'\sigma,\ell'\sigma'\sigma) \lessRwLbl{\Pset(\M)} (p,r\sigma,\ell\sigma)$. Otherwise, $r =_\En \ell'\sigma'$ and $s = r'\sigma'$. Hence $r\sigma =_\En \ell'\sigma'\sigma > r'\sigma'\sigma$ and so $(p,r'\sigma'\sigma,\ell'\sigma'\sigma) \lessRwLbl{\Pset(\M)} (p,r\sigma,\ell\sigma)$. It remains to consider $\omega \in p\cdot \tr_R\sigma$, that is $\omega = \leftrwlabel{p''}{\alpha\sigma}{a}{b}$ with $p \leq p''$: Once again, if $p'' \neq p$ then we directly have that $\measure{\omega} \lessRwLbl{\Pset(\M)} (p,r\sigma,\ell\sigma)$. Otherwise, it implies that $a\alpha\sigma =_\En b\alpha\sigma =_\En r\sigma$. Therefore since $\ell\sigma \geq r\sigma$, we deduce that $\measure{\omega} \leqRwLbl{\Pset(\M)} (p,r\sigma,\ell\sigma)$ which concludes the proof of \Cref{item:normalisation-normR-measure}.

We can represent graphically the transformation as follows:
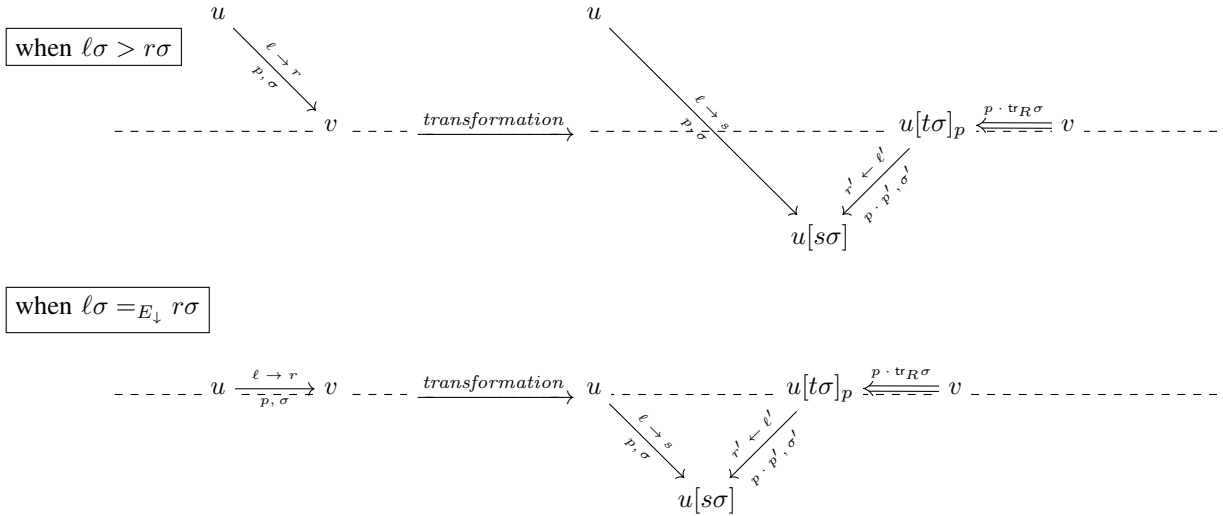
\begin{center}
    \begin{tikzpicture}[
        term/.style={}
        ]
        \def\length{1.5cm}
        \node[anchor=mid] (L) {};
        \node[right=15cm of L.mid,anchor=mid] (R) {};
        \draw[dashed] (L) -- (R);

        \node[draw,rectangle,above left= 1cm and 1.2cm of L,anchor=west] (Label) {when $\ell\sigma > r\sigma$};

        \node[fill=white,right=4cm of L.mid] (Label) {$\xrightarrow{\mathit{transformation}}$};

        \node[fill=white,term,above right=\length and \length of L.mid,anchor=mid] (T0) {$u$};
        \node[fill=white,term,below right=\length and \length of T0.mid,anchor=mid] (T1) {$v$};

        \tikzrightrwstep{T0}{p}{\sigma}{\ell}{r}{T1}

        \node[term,fill=white,right=5cm of T0.mid,anchor=mid] (T0') {$u$};
        \node[term,fill=white,below right= 2*\length and 2*\length of T0'.mid,anchor=mid] (T1') {$u[s\sigma]$};
        \node[term,fill=white,above right=\length and \length of T1'.mid,anchor=mid] (T2') {$u[t\sigma]_p$};
        \node[term,fill=white,right= 1.2*\length of T2'.mid,anchor=mid] (T3') {$v$};

        \tikzrightrwstep{T0'}{p}{\sigma}{\ell}{s}{T1'}
        \tikzleftrwstep{T1'}{p\cdot p'}{\sigma'}{\ell'}{r'}{T2'}
        \tikzLeftrwstep{T2'}{p \cdot \tr_R\sigma}{T3'}
    \end{tikzpicture}
\end{center}

\begin{center}
    \begin{tikzpicture}[
        term/.style={}
        ]
        \def\length{1.5cm}
        \node[anchor=mid] (L) {};
        \node[right=15cm of L.mid,anchor=mid] (R) {};
        \draw[dashed] (L) -- (R);

        \node[draw,rectangle,above left= 1cm and 1.2cm of L,anchor=west] (Label) {when $\ell\sigma =_\En r\sigma$};

        \node[fill=white,right=4cm of L.mid] (Label) {$\xrightarrow{\mathit{transformation}}$};

        \node[fill=white,term,right=\length of L.mid,anchor=mid] (T0) {$u$};
        \node[fill=white,term,right=\length of T0.mid,anchor=mid] (T1) {$v$};

        \tikzrightrwstep{T0}{p}{\sigma}{\ell}{r}{T1}

        \node[term,fill=white,right=5cm of T0.mid,anchor=mid] (T0') {$u$};
        \node[term,fill=white,below right= \length and \length of T0'.mid,anchor=mid] (T1') {$u[s\sigma]$};
        \node[term,fill=white,above right=\length and \length of T1'.mid,anchor=mid] (T2') {$u[t\sigma]_p$};
        \node[term,fill=white,right= 1.2*\length of T2'.mid,anchor=mid] (T3') {$v$};

        \tikzrightrwstep{T0'}{p}{\sigma}{\ell}{s}{T1'}
        \tikzleftrwstep{T1'}{p\cdot p'}{\sigma'}{\ell'}{r'}{T2'}
        \tikzLeftrwstep{T2'}{p \cdot \tr_R\sigma}{T3'}
    \end{tikzpicture}
\end{center}
\end{proof}

\Cref{lem:normalisation-normR} focuses only a right rewrite step $u \rightrwstep{p}{\sigma}{\ell}{r} v$ but we can easily have similar result on a left rewrite step by using the properties of the inverse operator in \Cref{lem:reverse-measure,lem:operation-properties}.


\begin{lemma}[Rule \RNormL]
\label{lem:normalisation-normL}
Let $u \rightrwstep{p}{\sigma}{\ell}{r} v$ be a ground rewrite step and $\tr_0 = \rightrwlabel{p}{\sigma}{\ell}{r}$. Let $\M$ a set of terms greater than $u \Leftrightrwstep{\tr_0} v$.
If $\ell = f(\ell_1,\ldots,\ell_n)$ and $\ell_i =_\En t \rightrwstep{p'}{\sigma'}{\ell'}{r'} s$ and $\ell' > r'$ and $i \in \{1, \ldots, n\}$ then there exists a ground syntactic rewrite trace $\tr_1$ from $\En \cup \{r \rightarrow \ell[s]_i, \ell[s]_i \rightarrow r, \ell' \rightarrow r'\}$ such that:
\begin{enumerate}[label=\textbf{P\arabic*}]
\item $u \Leftrightrwstep{\tr_1} v$, and \label{item:normalisation-normL-Leftright}
\item $\M$ greater than $u \Leftrightrwstep{\tr_1} v$, and \label{item:normalisation-normL-derive}
\item for all $\omega \in \tr_1$, $\measure{\omega} \lessRwLbl{\Pset(\M)} (p,r\sigma,\ell\sigma)$.\label{item:normalisation-normL-measure}
\end{enumerate}
\end{lemma}

\begin{proof}
Let us denote $w = \ell[s]_i$. Since $\ell_{|i} =_\En t$, there exists $\tr_L$ a syntactic rewrite trace from $\En$ such that $\ell_{|i} \Rightrwstep{\tr_t} t$. By definition of $t \rightrwstep{p'}{\sigma'}{\ell'}{r'} s$, we have $t_{|p'} = \ell'\sigma'$ and $s = t[r'\sigma']_{p'}$.
Moreover, $u \rightrwstep{p}{\sigma}{\ell}{r} v$ implies $u_{|p} = \ell\sigma$ and $v = u[r\sigma]_p$. Therefore, $u \Rightrwstep{p \cdot i \cdot \tr_L\sigma} u[\ell\sigma[t\sigma]_i]_p  \rightrwstep{p\cdot i\cdot p'}{\sigma'\sigma}{\ell'}{r'} u[\ell\sigma[t\sigma[r'\sigma'\sigma]_{p'}]_i]_p = u[\ell\sigma[s\sigma]_i]_p = u[w\sigma]_p$.

Let us define $\omega_M = \rightrwlabel{p}{\sigma}{w}{r}$ when $w\sigma > r\sigma$ and $\omega_M = \leftrwlabel{p}{\sigma}{r}{w}$ otherwise. In both cases, we deduce that $u[w\sigma]_p \Leftrightrwstep{\omega_M} u[r\sigma]_p = v$. Defining $\tr_1 = (p\cdot i \cdot \tr_L\sigma) \rightrwlabel{p\cdot i \cdot p'}{\sigma'\sigma}{\ell'}{r'} \omega_M$, we deduce that $u \rwLeftrightStep{\tr_1} v$ proving \Cref{item:normalisation-normL-Leftright}. 
Since $\M$ is greater than $u \Leftrightrwstep{\tr_0} v$, we obtain by construction of $\tr_1$ that $\M$ is greater than $u \Leftrightrwstep{\tr_1} v$ thus proving \Cref{item:normalisation-normL-derive}.

We now focus on \Cref{item:normalisation-normL-measure}: Since $\ell' > r'$, we deduce that $\ell'\sigma' > r'\sigma'$ and so $\ell[t]_i > \ell[s]_i$. As $>$ is $\En$-compatible and $t =_\En \ell_i$, we deduce that $\ell > w$, implying $\ell\sigma > w\sigma$.
If $w\sigma > r\sigma$ then we obtain $\measure{\omega_M} = (p,r\sigma,w\sigma) \lessRwLbl{\Pset(\M)} (p,r\sigma,\ell\sigma)$ else if $r\sigma > w\sigma$ then $\measure{\omega_M} = (p,w\sigma,r\sigma) \lessRwLbl{\Pset(\M)} (p,r\sigma,\ell\sigma)$, else $r\sigma =_\En w\sigma$ but $\ell\sigma > w\sigma$ hence $\measure{\omega_M} = (p,w\sigma,r\sigma) \lessRwLbl{\Pset(\M)} (p,r\sigma,\ell\sigma)$. 

Consider now $\omega = \rightrwlabel{p\cdot i \cdot p'}{\sigma'\sigma}{\ell'}{r'}$, we have $\measure{\omega} = (p\cdot i \cdot p',r'\sigma'\sigma,\ell'\sigma'\sigma)$. As $p$ is a strict prefix of $p\cdot i \cdot p'$, we directly have $\measure{\omega} \lessRwLbl{\Pset(\M)} (p,r\sigma,\ell\sigma)$. For the same reason, we deduce that for all $\omega \in p\cdot i \cdot \tr_L\sigma$, $\measure{\omega} < (p,r\sigma,\ell\sigma)$.

We can represent graphically the transformation as follows:
\begin{center}
    \begin{tikzpicture}[
        term/.style={}
        ]
        \def\length{1.5cm}
        \node[anchor=mid] (L) {};
        \node[right=15cm of L.mid,anchor=mid] (R) {};
        \draw[dashed] (L) -- (R);

        \node[draw,rectangle,above left= 1cm and 1.2cm of L,anchor=west,text width=2.1cm] (Label) {when $\ell\sigma > r\sigma$ and $w\sigma < r\sigma$};

        \node[fill=white,right=4.45cm of L.mid] (Label) {$\xrightarrow{\mathit{transformation}}$};

        \node[fill=white,term,above right=\length and 1.5cm of L.mid,anchor=mid] (T0) {$u$};
        \node[fill=white,term,below right=\length and \length of T0.mid,anchor=mid] (T1) {$v$};

        \tikzrightrwstep{T0}{p}{\sigma}{\ell}{r}{T1}

        \node[term,fill=white,right=5.5cm of T0.mid,anchor=mid] (T0') {$u$};
        \node[term,fill=white,right= 2*\length of T0'.mid,anchor=mid] (T1') {$u[\ell\sigma[t\sigma]_i]_p$};
        \node[term,fill=white,below right=2*\length and 2*\length of T1'.mid,anchor=mid] (T2') {$u[w\sigma]_p$};
        \node[term,fill=white,above right= \length and \length of T2'.mid,anchor=mid] (T3') {$v$};

        \tikzRightrwstep{T0'}{p \cdot i \cdot \tr_L\sigma}{T1'}
        \draw[->] (T1') edge node[auto,sloped] {\tiny$\ell' \rightarrow r'$} node[auto,sloped,below] {\tiny$p\cdot i\cdot p', \sigma'\sigma$} (T2');
        \draw[<-] (T2') edge node[auto,sloped] {\tiny$w \leftarrow r$} node[auto,sloped,below] {\tiny$p, \sigma$} (T3');
    \end{tikzpicture}
\end{center}
\begin{center}
    \begin{tikzpicture}[
        term/.style={}
        ]
        \def\length{1.5cm}
        \node[anchor=mid] (L) {};
        \node[right=15cm of L.mid,anchor=mid] (R) {};
        \draw[dashed] (L) -- (R);

        \node[draw,rectangle,above left= 1cm and 1.2cm of L,anchor=west,text width=2.1cm] (Label) {when $\ell\sigma > r\sigma$ and $w\sigma =_\En r\sigma$};

        \node[fill=white,right=4.5cm of L.mid] (Label) {$\xrightarrow{\mathit{transformation}}$};

        \node[fill=white,term,above right=\length and 1.5cm of L.mid,anchor=mid] (T0) {$u$};
        \node[fill=white,term,below right=\length and \length of T0.mid,anchor=mid] (T1) {$v$};

        \draw[->] (T0) edge node[auto,sloped] {\tiny$\ell \rightarrow r$} node[auto,sloped,below] {\tiny$p, \sigma$} (T1);

        \node[term,fill=white,right=5.5cm of T0.mid,anchor=mid] (T0') {$u$};
        \node[term,fill=white,right= 2*\length of T0'.mid,anchor=mid] (T1') {$u[\ell\sigma[t\sigma]_i]_p$};
        \node[term,fill=white,below right=\length and \length of T1'.mid,anchor=mid] (T2') {$u[w\sigma]_p$};
        \node[term,fill=white,right= 1.5*\length of T2'.mid,anchor=mid] (T3') {$v$};

        \draw (T0') edge[double equal sign distance,-Implies] node[auto,sloped] {\tiny$p \cdot i \cdot \tr_L\sigma$} (T1');
        \draw[->] (T1') edge node[auto,sloped] {\tiny$\ell' \rightarrow r'$} node[auto,sloped,below] {\tiny$p\cdot i \cdot p', \sigma'\sigma$} (T2');
        \draw[<-] (T2') edge node[auto,sloped] {\tiny$w \leftarrow r$} node[auto,sloped,below] {\tiny$p, \sigma$} (T3');
    \end{tikzpicture}
\end{center}
\begin{center}
    \begin{tikzpicture}[
        term/.style={}
        ]
        \def\length{1.5cm}
        \node[anchor=mid] (L) {};
        \node[right=15cm of L.mid,anchor=mid] (R) {};
        \draw[dashed] (L) -- (R);

        \node[draw,rectangle,above left= 1cm and 1.2cm of L,anchor=west,text width=2.1cm] (Label) {when $\ell\sigma > r\sigma$ and $w\sigma > r\sigma$};

        \node[fill=white,right=4.5cm of L.mid] (Label) {$\xrightarrow{\mathit{transformation}}$};

        \node[fill=white,term,above right=\length and 1.5cm of L.mid,anchor=mid] (T0) {$u$};
        \node[fill=white,term,below right=2*\length and 2*\length of T0.mid,anchor=mid] (T1) {$v$};

        \draw[->] (T0) edge node[auto,sloped] {\tiny$\ell \rightarrow r$} node[auto,sloped,below] {\tiny$p, \sigma$} (T1);

        \node[term,fill=white,right=5.5cm of T0.mid,anchor=mid] (T0') {$u$};
        \node[term,fill=white,right= 2*\length of T0'.mid,anchor=mid] (T1') {$u[\ell\sigma[t\sigma]_i]_p$};
        \node[term,fill=white,below right=\length and \length of T1'.mid,anchor=mid] (T2') {$u[w\sigma]_p$};
        \node[term,fill=white,below right= \length and \length of T2'.mid,anchor=mid] (T3') {$v$};

        \draw (T0') edge[double equal sign distance,-Implies] node[auto,sloped] {\tiny$p \cdot i \cdot \tr_L\sigma$} (T1');
        \draw[->] (T1') edge node[auto,sloped] {\tiny$\ell' \rightarrow r'$} node[auto,sloped,below] {\tiny$p\cdot i\cdot p', \sigma'\sigma$} (T2');
        \draw[->] (T2') edge node[auto,sloped] {\tiny$w \rightarrow r$} node[auto,sloped,below] {\tiny$p, \sigma$} (T3');
    \end{tikzpicture}
\end{center}
\begin{center}
    \begin{tikzpicture}[
        term/.style={}
        ]
        \def\length{1.5cm}
        \node[anchor=mid] (L) {};
        \node[right=15cm of L.mid,anchor=mid] (R) {};
        \draw[dashed] (L) -- (R);

        \node[draw,rectangle,above left= 1cm and 1.2cm of L,anchor=west,text width=2.3cm] (Label) {when $\ell\sigma =_\En r\sigma$};

        \node[fill=white,right=4.5cm of L.mid] (Label) {$\xrightarrow{\mathit{transformation}}$};

        \node[fill=white,term,right=1.5cm of L.mid,anchor=mid] (T0) {$u$};
        \node[fill=white,term,right=\length of T0.mid,anchor=mid] (T1) {$v$};

        \draw[->] (T0) edge node[auto,sloped] {\tiny$\ell \rightarrow r$} node[auto,sloped,below] {\tiny$p, \sigma$} (T1);

        \node[term,fill=white,right=5.5cm of T0.mid,anchor=mid] (T0') {$u$};
        \node[term,fill=white,right= 2*\length of T0'.mid,anchor=mid] (T1') {$u[\ell\sigma[t\sigma]_i]_p$};
        \node[term,fill=white,below right=\length and \length of T1'.mid,anchor=mid] (T2') {$u[w\sigma]_p$};
        \node[term,fill=white,above right= \length and \length of T2'.mid,anchor=mid] (T3') {$v$};

        \draw (T0') edge[double equal sign distance,-Implies] node[auto,sloped] {\tiny$p \cdot i \cdot \tr_L\sigma$} (T1');
        \draw[->] (T1') edge node[auto,sloped] {\tiny$\ell' \rightarrow r'$} node[auto,sloped,below] {\tiny$p\cdot i\cdot p', \sigma'\sigma$} (T2');
        \draw[<-] (T2') edge node[auto,sloped] {\tiny$w \leftarrow r$} node[auto,sloped,below] {\tiny$p, \sigma$} (T3');
    \end{tikzpicture}
\end{center}
\end{proof}


\begin{lemma}[Rule \RSubsume]
\label{lem:normalisation-subsume}
Let $u \rightrwstep{p}{\sigma}{\ell}{r} v$ be a ground rewrite step and $\tr_0 = \rightrwlabel{p}{\sigma}{\ell}{r}$. Let $\M$ a set of terms greater than $u \Leftrightrwstep{\tr_0} v$.
If there exist $\ell' \rightarrow r'$, a term context $C[\_]$ and a substitution $\sigma'$ such that $C[\ell'\sigma'] \eqES{\En} \ell$ and $C[r'\sigma'] \eqE{\En} r$ then there exists $\tr_1$ a ground syntactic rewrite trace from $\En \cup \{ \ell' \rightarrow r' \}$ such that:
\begin{enumerate}[label=\textbf{P\arabic*}]
\item $u \Leftrightrwstep{\tr_1} v$, and \label{item:normalisation-subsume-Leftright}
\item $\M$ is greater than $u \Leftrightrwstep{\tr_1} v$, and \label{item:normalisation-subsume-derive}
\item for all $\omega \in \tr_1$, $\measure{\omega} \leqRwLbl{\Pset(\M)} (p,r\sigma,\ell\sigma)$.\label{item:normalisation-subsume-measure}
\end{enumerate}
\end{lemma}

\begin{proof}
By definition of $C[\ell'\sigma'] \eqES{\En} \ell$, there exists $\tr_L$ a syntactic rewrite trace from $\En$ such that $\ell \rwRightStep{\tr_L} C[\ell'\sigma']$ and for all $\omega \in \tr_L$, $\pos(\omega) \neq \varepsilon$. As $u \rightrwstep{p}{\sigma}{\ell}{r} v$ implies $u_{|p} = \ell\sigma$ and $v = u[r\sigma]_p$, we deduce that $u \Rightrwstep{p \cdot \tr_L\sigma} u[C\sigma[\ell'\sigma'\sigma]]_p$.
Similarly, $C[r'\sigma'] \eqE{\En} r$ implies that there exists $\tr_R$ a syntactic rewrite trace from $\En$ such that $C[r'\sigma'] \Rightrwstep{\tr_R} r$. Hence $u[C\sigma[r'\sigma'\sigma]]_p \rwRightStep{p \cdot \tr_R\sigma} v$.

Recall that $u \rightrwstep{p}{\sigma}{\ell}{r} v$ implies $\ell\sigma \geq r\sigma$. By contradiction, assume that $r'\sigma'\sigma > \ell'\sigma'\sigma$. Since $>$ is a $\En$-compatible reduction ordering, we deduce that $r\sigma =_\En C\sigma[r'\sigma'\sigma] > C\sigma[\ell'\sigma'\sigma] =_\En \ell\sigma$ and so $r\sigma > \ell\sigma$, which is a contradiction with $\ell\sigma \geq r\sigma$. Therefore $r'\sigma'\sigma \leq \ell'\sigma'\sigma$ and so we obtain $C\sigma[\ell'\sigma'\sigma] \rightrwstep{p'}{\sigma'\sigma}{\ell'}{r'} C\sigma[r'\sigma'\sigma]$ with $C[\_]_{|p'} = \_$ which entails $u[C\sigma[\ell'\sigma'\sigma]]_p \rightrwstep{p\cdot p'}{\sigma'\sigma}{\ell'}{r'} u[C\sigma[r'\sigma'\sigma]]_p$. Defining $\tr_1 = (p\cdot \tr_L\sigma) \rightrwlabel{p\cdot p'}{\sigma'\sigma}{\ell'}{r'} (p\cdot \tr_R\sigma)$, we deduce that $u \Leftrightrwstep{\tr_1} v$, yielding \Cref{item:normalisation-subsume-Leftright}. 
As $\M$ is greater than $u \Leftrightrwstep{\tr_0} v$, we also directly obtain by construction of $\tr_1$ that $\M$ is greater than $u \Leftrightrwstep{\tr_1} v$, hence proving \Cref{item:normalisation-subsume-derive}.

We now focus on \Cref{item:normalisation-subsume-measure}. We know that for all $\omega \in \tr_L$, $\pos(\omega) \neq \varepsilon$. Hence for all $\omega \in p\cdot \tr_L\sigma$, $p , \pos(\omega)$ meaning that $\measure{\omega} \lessRwLbl{\Pset(\M)} (p,r\sigma,\ell\sigma)$. Let us denote $\omega' = \rightrwlabel{p\cdot p'}{\sigma'\sigma}{\ell'}{r'}$. Recall that $C[\_]_{|p'} = \_$. Thus, if $p' \neq \varepsilon$ then we also directly have $\measure{\omega'} \lessRwLbl{\Pset(\M)} (p,r\sigma,\ell\sigma)$; and otherwise $\ell'\sigma'\sigma =_\En \ell\sigma$ and $r'\sigma'\sigma =_\En r\sigma$ which implies $\measure{\omega'} = (p,r'\sigma'\sigma,\ell'\sigma'\sigma) \eqRwLbl{\Pset(\M)} (p,r\sigma,\ell\sigma)$. Therefore, in both cases, we have $\measure{\omega'} \leqRwLbl{\Pset(\M)} (p,r\sigma,\ell\sigma)$.
Finally, for all $\rightrwlabel{p''}{\alpha}{s}{t} \in p\cdot \tr_R\sigma$, if $p'' = p$ then $s\alpha \eqE{\En} t\alpha \eqE{\En} r\sigma$. As $\ell\sigma \geq r\sigma$, we deduce that either $(p'',t\alpha,s\alpha) \leqRwLbl{\Pset(\M)} (p,r\sigma,\ell\sigma)$. We conclude the proof of \Cref{item:normalisation-subsume-measure}.

We can represent graphically the transformation as follows:
\begin{center}
    \begin{tikzpicture}[
        term/.style={}
        ]
        \def\length{1.5cm}
        \node[anchor=mid] (L) {};
        \node[right=15cm of L.mid,anchor=mid] (R) {};
        \draw[dashed] (L) -- (R);


        \node[fill=white,right=3cm of L.mid] (Label) {$\xrightarrow{\mathit{transformation}}$};

        \node[fill=white,term,right=1cm of L.mid,anchor=mid] (T0) {$u$};
        \node[fill=white,term,right=\length of T0.mid,anchor=mid] (T1) {$v$};

        \draw[->] (T0) edge node[auto,sloped] {\tiny$\ell \rightarrow r$} node[auto,sloped,below] {\tiny$p, \sigma$} (T1);

        \node[term,fill=white,right=4.5cm of T0.mid,anchor=mid] (T0') {$u$};
        \node[term,fill=white,right= 1.5*\length of T0'.mid,anchor=mid] (T1') {$u[C\sigma[\ell'\sigma'\sigma]]_p$};
        \node[term,fill=white,right=2.5*\length of T1'.mid,anchor=mid] (T2') {$u[C\sigma[r'\sigma'\sigma]]_p$};
        \node[term,fill=white,right= 1.5*\length of T2'.mid,anchor=mid] (T3') {$v$};

        \draw (T0') edge[double equal sign distance,-Implies] node[auto,sloped] {\tiny$p \cdot \tr_L\sigma$} (T1');
        \draw[->] (T1') edge node[auto,sloped] {\tiny$\ell' \rightarrow r'$} node[auto,sloped,below] {\tiny$p\cdot p', \sigma'\sigma$} (T2');
        \draw (T2') edge[double equal sign distance,-Implies] node[auto,sloped] {\tiny$p \cdot \tr_R\sigma$} (T3');
    \end{tikzpicture}
\end{center}
\end{proof}


\begin{lemma}[Rule \REq]
\label{lem:normalisation-eq}
Let $u \rightrwstep{p}{\sigma}{\ell}{r} v$ be a ground rewrite step and $\tr_0 = \rightrwlabel{p}{\sigma}{\ell}{r}$. Let $\M$ a set of terms greater than $u \Leftrightrwstep{\tr_0} v$.
If $\ell \eqE{\En} r$ then there exists $\tr_1$ a ground syntactic rewrite trace from $\En$ such that:
\begin{enumerate}[label=\textbf{P\arabic*}]
\item $u \Leftrightrwstep{\tr_1} v$, and \label{item:normalisation-eq-Leftright}
\item $\M$ is greater than $u \Leftrightrwstep{\tr_1} v$, and \label{item:normalisation-eq-derive}
\item for all $\omega \in \tr_1$, $\measure{\omega} \leqRwLbl{\Pset(\M)} (p,r\sigma,\ell\sigma)$.\label{item:normalisation-eq-measure}
\end{enumerate}
\end{lemma}

\begin{proof}
By definition of $u \rightrwstep{p}{\sigma}{\ell}{r} v$, $u_{|p} = \ell\sigma$ and $v_{|p} = r\sigma$. As $\ell \eqE{\En} r$, there exists $\tr_M$ a syntactic rewrite trace from $\En$ such that $\ell \Rightrwstep{\tr_M} r$. Thus, $u_{|p} \Rightrwstep{\tr_M\sigma} v_{|p}$ and so $u \rwRightStep{p\cdot \tr_M\sigma} v$. Defining $\tr_1 = p \cdot \tr_M\sigma$, we directly conclude \Cref{item:normalisation-eq-Leftright}. As $\M$ is greater than $u \Leftrightrwstep{\tr_0} v$, we also deduce by construction of $\tr_1$ that $\M$ is greater than $u \Leftrightrwstep{\tr_1} v$, hence proving \Cref{item:normalisation-subsume-derive}. Moreover, as $\ell\sigma =_\En r\sigma = u_{|p}$, we deduce that for all $\omega = \rightrwlabel{p'}{\alpha}{s}{t} \in \tr_1$, $p \leq p'$ and if $p = p'$ then $s\alpha =_\En t\alpha =_\En \ell\sigma =_\En r\sigma$. Hence we deduce that $\measure{\omega} \leqRwLbl{\Pset(\M)} (p,r\sigma,\ell\sigma)$ which allows us to conclude the proof of \Cref{item:normalisation-eq-measure}.

We can represent graphically the transformation as follows:
\begin{center}
    \begin{tikzpicture}[
        term/.style={}
        ]
        \def\length{1.5cm}
        \node[anchor=mid] (L) {};
        \node[right=8cm of L.mid,anchor=mid] (R) {};
        \draw[dashed] (L) -- (R);


        \node[fill=white,right=3cm of L.mid] (Label) {$\xrightarrow{\mathit{transformation}}$};

        \node[fill=white,term,right=1cm of L.mid,anchor=mid] (T0) {$u$};
        \node[fill=white,term,right=\length of T0.mid,anchor=mid] (T1) {$v$};

        \draw[->] (T0) edge node[auto,sloped] {\tiny$\ell \rightarrow r$} node[auto,sloped,below] {\tiny$p, \sigma$} (T1);

        \node[term,fill=white,right=4.5cm of T0.mid,anchor=mid] (T0') {$u$};
        \node[term,fill=white,right= 1*\length of T0'.mid,anchor=mid] (T1') {$v$};

        \draw (T0') edge[double equal sign distance,-Implies] node[auto,sloped] {\tiny$p \cdot \tr_M\sigma$} (T1');
    \end{tikzpicture}
\end{center}
\end{proof}


\begin{lemma}
\label{lem:normalisation-all}
Let $\Rn$ and $\R$ two sets of rewrite rules such that $\Rn \subseteq \R$. Assume that $\R \normstep[\Rn,\En]^* \R'$ and let $u \rightrwstep{p}{\sigma}{\ell}{r} v$ be a ground rewrite step. Let $\M$ be a set of terms greater than $u$ and $v$. Let $q$ be a position and $a,b$ two terms. There exists a ground syntactic rewrite trace $\tr_1$ from $\En \cup \R'$ such that:
\begin{enumerate}[label=\textbf{P\arabic*}]
\item $u \rwLeftrightStep{\tr_1} v$ \label{item:normalisation-all-Leftright}
\item $\M$ is greater than $u \rwLeftrightStep{\tr_1} v$ \label{item:normalisation-all-derive}
\item if $(p,r\sigma,\ell\sigma) \lessRwLbl{\Pset(\M)} (q,a,b)$ then for all $\omega \in \tr$, $\measure{\omega} \lessRwLbl{\Pset(\M)} (q,a,b)$\label{item:normalisation-all-measure}
\end{enumerate}
\end{lemma}

\begin{proof}
Since $\R \normstep[\Rn,\En]^* \R'$, we have in fact $\R = \R_0 \normstep[\Rn,\En] \ldots \normstep[\Rn,\En] \R_n = \R'$ for some $\R_0,\ldots,\R_n$. We show the following result:

For all $i \in \{0,\ldots,n\}$, for all $u \rightrwstep{p}{\sigma}{\ell}{r} v$ (resp. $u \leftrwstep{p}{\sigma}{\ell}{r} v$), if $(\ell \rightarrow r) \in \R_i$ and $\M$ is greater than $u,v$ and $(p,r\sigma,\ell\sigma) \lessRwLbl{\Pset(\M)} (q,a,b)$ then there exists a syntactic rewrite trace $\tr$ from $\En \cup \R'$ such that:
\begin{itemize}
\item $u \rwLeftrightStep{\tr_1} v$
\item $\M$ is greater than $u \rwLeftrightStep{\tr_1} v$
\item for all $\omega \in \tr$, $\measure{\omega} \lessRwLbl{\Pset(\M)} (q,a,b)$
\end{itemize}

We prove this result by induction on $((p,r\sigma,\ell\sigma), n-i)$ using the lexicographic ordering on pair, using $\lessRwLbl{\Pset(\M)}$ to compare the first element of the pair and standard natural number ordering on the second element of the pair. 

Once we proved that the result holds for $u \rightrwstep{p}{\sigma}{\ell}{r} v$, we can derive the result for $u \leftrwstep{p}{\sigma}{\ell}{r} v$ by applying the result on $v \rightrwstep{p}{\sigma}{\ell}{r} u$, leading to the existence of a syntactic trace $\tr_0$ from $\En \cup \R'$ such that $v \rwLeftrightStep{\tr_0} u$, $\M$ rewrites into $v \rwLeftrightStep{\tr_0} u$ and for all $\omega \in \tr_0$, $\measure{\omega} \lessRwLbl{\Pset(\M)} (q,a,b)$. We conclude by taking $\tr_1 = \reverse{\tr_0}$ and by applying \Cref{lem:operation-properties,lem:reverse-measure}.

We now focus on proving the result for $u \rightrwstep{p}{\sigma}{\ell}{r} v$.

\paragraph{Base case $((p,r\sigma,\ell\sigma), 0)$:} In such a case, we have $i = n$ and so $\R_i = \R'$. Hence the result directly holds with $\tr = \rightrwlabel{p}{\sigma}{\ell}{r}$.

\paragraph{Inductive step $((p,r\sigma,\ell\sigma), n-i)$:} Let us look at the normalisation rule $\R_i \normstep[\Rn,\En] \R_{i+1}$. If $(\ell \rightarrow r) \in \R_{i+1}$ then we can apply our inductive hypothesis on $((p,r\sigma,\ell\sigma), n-i-1)$ which allows us to conclude. Therefore we assume that $\ell \rightarrow r \not\in \R_{i+1}$. We do a case analysis on the normalisation rule $\R_i \normstep[\Rn,\En] \R_{i+1}$.
\begin{itemize}
\item \emph{Rule \RNormR:} In such a case, $\R_i = \R'' \cup \{ \ell \rightarrow r\}$, $r =_\En t \rwstep{\Rn} s$ and $\R_{i+1} = \R'' \cup \{ \ell \rightarrow s\}$. As $r =_\En t \rwstep{\Rn} s$, we have $r =_\En t \rightrwstep{p'}{\sigma'}{\ell'}{r'} s$ for some $p',\sigma'$ and $\ell' \rightarrow r' \in \Rn$. Hence we can apply \Cref{lem:normalisation-normR} to obtain that there exist a syntactic rewrite trace $\tr_R$ from $\En$ and $\tr_1 = \rightrwlabel{p}{\sigma}{\ell}{s} \leftrwlabel{p\cdot p'}{\sigma'\sigma}{\ell'}{r'} (p\cdot \tr_R\sigma)$ such that
\[
    u \rightrwstep{p}{\sigma}{\ell}{s} u[s\sigma]_p \leftrwstep{p\cdot p'}{\sigma'\sigma}{\ell'}{r'} u[t\sigma]_p \Leftrwstep{p\cdot \tr_R\sigma}  v
\]
and the following properties hold:
\begin{itemize}
\item $u \rwLeftrightStep{\tr_1} v$, and
\item $\M$ is greater than $u \rwLeftrightStep{\tr_1} v$, and
\item for all $\omega \in \tr_1$, $\measure{\omega} \lessRwLbl{\Pset(\M)} (p,r\sigma,\ell\sigma)$ or ($\measure{\omega} \eqRwLbl{\Pset(\M)} (p,r\sigma,\ell\sigma)$ and $\omega \in p\cdot \tr_R\sigma$)
\end{itemize}
From the last property, we deduce that $(p,s\sigma,\ell\sigma) \lessRwLbl{\Pset(\M)} (p,r\sigma,\ell\sigma)$ and $(p\cdot p',r'\sigma'\sigma,\ell\sigma'\sigma) \lessRwLbl{\Pset(\M)} (p,r\sigma,\ell\sigma)$. Hence, as $(\ell \rightarrow s) \in \R_{i+1}$, we can apply our inductive hypothesis on $u \rightrwstep{p}{\sigma}{\ell}{s} u[s\sigma]_p$ which gives us the existence of a syntactic rewrite trace $\tr_2$ from $\En \cup \R'$ such that $u \rwLeftrightStep{\tr_2} u[s\sigma]_p$, which is smaller than $\M$, and for all $\omega \in \tr_2$, $\measure{\omega} \lessRwLbl{\Pset(\M)} (q,a,b)$.

Similarly, as $\ell' \rightarrow r' \in \R_0$, we can apply our inductive hypothesis on $u[s\sigma]_p \leftrwstep{p\cdot p'}{\sigma'\sigma}{\ell'}{r'} u[t\sigma]_p$ which implies the existence of a syntactic rewrite trace $\tr_3$ from $\En \cup \R'$ such that $u[s\sigma]_p \Leftrwstep{\tr_3} u[t\sigma]_p$, which is smaller than $\M$, and for all $\omega \in \tr_3$, $\measure{\omega} \lessRwLbl{\Pset(\M)} (q,a,b)$. We conclude by taking $\tr = \tr_2 \tr_3 (p \cdot p' \cdot \tr_R\sigma)$.
\item \emph{Rule \RNormL:} In such a case, $\ell = f(\ell_1,\ldots,\ell_m)$, $\ell_j =_\En t \rwstep{\Rn} \ell'_j$, $\R_i = \R'' \cup \{ \ell \rightarrow r\}$ and $\R_{i+1} = \R'' \cup \{ \ell[\ell'_j]_{j} \rightarrow r, r \rightarrow \ell[\ell'_j]_{j}\}$. Denoting $s = \ell[\ell'_j]_j$ and applying \Cref{lem:normalisation-normL}, we deduce that there exists a syntactic rewrite trace $\tr_1 = \omega_1 \ldots \omega_m$ from $\En \cup \{ r \rightarrow s, s\rightarrow r, \ell' \rightarrow r' \}$ such that:
\[
u = u_0 \Leftrightrwstep{\omega_1} u_1 \Leftrightrwstep{\omega_2} \ldots \Leftrightrwstep{\omega_m} u_m = v
\]
and $u \Leftrightrwstep{\tr_1} v$, which is smaller than $\M$ derives, and for all $j \in \{1, \ldots, m\}$, $\measure{\omega_j} \lessRwLbl{\Pset(\M)} (p,r\sigma,\ell\sigma)$.

Take $j \in \{1, \ldots, m\}$ such that $\omega_j = \leftrightrwlabel{p''}{\alpha}{w}{w'}$. If $(w = w') \in \En$ then we define $\tr'_j = \omega_j$ else if $(w \sim w') \in \{ s \leftarrow r, s\rightarrow r\}$ then $(w \sim w') \in \R_{i+1}$ else we have $(w \sim w') = (\ell' \rightarrow r') \in \R_0$. In the last two cases, since $\measure{\omega_j} \lessRwLbl{\Pset(\M)} (p,r\sigma,\ell\sigma)$ we can apply our inductive hypothesis on $u_{j-1} \Leftrightrwstep{\omega_j} u_j$ which implies the existence of a syntactic rewrite trace $\tr'_j$ from $\En \cup \R'$ such that $u_{j-1} \Leftrightrwstep{\tr'_j} u_j$ which is smaller than $\M$ and for all $\omega \in \tr'_j$, $\measure{\omega} \lessRwLbl{\Pset(\M)} (q,a,b)$. We conclude by taking $\tr = \tr'_1 \ldots \tr'_m$
\item \emph{Rule \REq:} In such a case $\R_i = \R_{i+1} \cup \{ \ell \rightarrow r\}$ and $\ell =_\En r$. Applying \Cref{lem:normalisation-eq}, the result directly holds.
\item \emph{Rule \RSubsume:} In such a case $\R_i = \R_{i+1} \cup \{ \ell \rightarrow r\}$ and there exist $(\ell' \rightarrow r') \in \R_{i+1}$, a term context $C[\_]$ and a substitution $\sigma'$ such that $C[\ell'\sigma'] \eqES{\En} \ell$ and $C[r'\sigma'] \eqE{\En} r$. Applying \Cref{lem:normalisation-subsume}, we deduce that there exists a syntactic rewrite trace $\tr_1 = \omega_1 \ldots \omega_m$ from $\En \cup \{ \ell' \rightarrow r'\}$ such that:
\[
u = u_0 \Leftrightrwstep{\omega_1} u_1 \Leftrightrwstep{\omega_2} \ldots \Leftrightrwstep{\omega_m} u_m = v
\]
and $u \Leftrightrwstep{\tr_1} v$ which is smaller than $\M$, and for all $j \in \{1, \ldots, m\}$, $\measure{\omega_j} \leqRwLbl{\Pset(\M)} (p,r\sigma,\ell\sigma)$. Thus $\measure{\omega_j} \lessRwLbl{\Pset(\M)} (q,a,b)$ for all $j \in \{1, \ldots, m\}$. 

Take $j \in \{1, \ldots, m\}$ such that $\omega_j = \leftrightrwlabel{p''}{\alpha}{w}{t}$. If $(w = t) \in \En$ then we define $\tr'_j = \omega_j$ else $(w \sim t) = (\ell' \rightarrow r') \in \R_{i+1}$. As $\measure{\omega_j} \leqRwLbl{\Pset(\M)} (p,r\sigma,\ell\sigma)$, we deduce that $(\measure{\omega_j},n-i-1)$ is strictly smaller than $((p,r\sigma,\ell\sigma),n-i)$. 
We can therefore apply our inductive hypothesis on $u_{j-1} \Leftrightrwstep{\omega_j} u_j$ which implies the existence of a syntactic rewrite trace $\tr'_j$ from $\En \cup \R'$ such that $u_{j-1} \Leftrightrwstep{\tr'_j} u_j$ which $\M$ rewrites into, and for all $\omega \in \tr'_j$, $\measure{\omega} \lessRwLbl{\Pset(\M)} (q,a,b)$. We conclude by taking $\tr = \tr'_1 \ldots \tr'_m$.\qedhere
\end{itemize}
\end{proof}


\begin{corollary}
\label{cor:normalisation}
Let $\Rn$ and $\R$ two sets of rewrite rules such that $\Rn \subseteq \R$. Assume that $\R \normstep[\Rn,\En]^* \R'$. Let $\tr_0,\tr_1$ be two ground syntactic rewrite trace from $\En \cup \R$. Let $\M$ be a set of terms. Let $u,v$ two ground terms such that $u \Leftrightrwstep{\tr_0} v$, and $u \Leftrightrwstep{\tr_1} v$, and $\M$ rewrites into both $u \Leftrightrwstep{\tr_0} v$ and $u \Leftrightrwstep{\tr_1} v$.

There exists a ground syntactic rewrite trace $\tr_2$ from $\En \cup \R'$ such that $u \Leftrightrwstep{\tr_2} v$ which $\M$ orderly rewrites into, and if for all $\omega_1 \in \tr_1$, there exists $\omega_0 \in \tr_0$ such that $\measure{\omega_1} \lessRwLbl{\Pset(\M)} \measure{\omega_0}$ then $\measure{\tr_2} \lessRwLbl{\Pset(\M)} \measure{\tr_0}$.
\end{corollary}

\begin{proof}
Direct application of \Cref{lem:normalisation-all} on each $\omega_1 \in \tr_1$ where $(q,a,b)$ is taken as $\measure{\omega_0}$ where $\measure{\omega_1} \lessRwLbl{\Pset(\M)} \measure{\omega_0}$ and $\omega_0 \in \tr_0$ (its existence is given as hypothesis).
\end{proof}

We complete this section by proving some useful properties on the normalisation rules. Given two sets of rewrite rules $\R_1, \R_2$, we will say that $\R_1$ subsumes $\R_2$ w.r.t $\Rn$ if for all $(\ell \rightarrow r) \in \R_2$, either $\ell =_\En r$ or there exists $(\ell' \rightarrow r') \in \R_1$, a term context $C[\_]$ and a substitution such that $C[\ell'\sigma] \eqES{\En} \ell$ and $C[r'\sigma] \eqE{\En} r$ and $r' \not\rwstepC[\En]{\Rn}$ and $\forall p > \varepsilon. \ell'_{|p} \not\rwstepC[\En]{\Rn}$.


\begin{lemma}
\label{lem:properties-normalisation}
Let $\R$, $\R'$, $\R_0$ and  $\Rn$ be sets of rewrite rules.
\begin{enumerate}
\item $\R \normstep[\Rn,\En]^* \R'$ implies $\R \cup \R_0 \normstep[\Rn,\En]^* \R' \cup \R_0$\label{item:properties-normalisation-set}
\item $\R \normstep[\Rn,\En]^* \R'$ and $\R$ subsumes $\R_0$ w.r.t $\Rn$ implies $\R'$ subsumes $\R_0$ w.r.t. $\Rn$\label{item:properties-normalisation-subsume-transitive}
\item If $\R$ subsumes $\R_0$ w.r.t. $\Rn$ then $\R \cup \R_0 \normstep[\Rn,\En]^* \R$\label{item:properties-normalisation-subsume}
\item $\R \normstep[\Rn,\En]^* \R'$ and $(\ell \rightarrow r) \in \R$ implies there exist $\R_{in}, \R_{rm}$ such that $\{ \ell \rightarrow r\} \normstep[\Rn,\En]^* \R_{in} \cup \R_{rm}$, and $\R_{in} \subseteq \R$ and $\R$ subsumes $\R_{rm}$ w.r.t. $\Rn$\label{item:properties-normalisation-split}
\end{enumerate}
\end{lemma}

\begin{proof}
\Cref{item:properties-normalisation-set} is direct by definition of the normalisation rules in \Cref{fig:normalisation rules}. If $\R$ subsumes $\R_0$ then by definition for all $(\ell \rightarrow r) \in \R_0$, either $\ell =_\En r$ or there exists $(\ell' \rightarrow r') \in \R$, a term context $C[\_]$ and a substitution such that $C[\ell'\sigma] \eqES{\En} \ell$ and $C[r'\sigma] \eqE{\En} r$ and $r' \not\rwstepC[\En]{\Rn}$ and $\forall p > \varepsilon. \ell'_{|p} \not\rwstepC[\En]{\Rn}$. Hence, we can successively apply the rules \RSubsume and \REq to obtain that $\R \cup \R_0 \normstep[\Rn,\En]^* \R$, yielding \Cref{item:properties-normalisation-subsume}. 

We now look at \Cref{item:properties-normalisation-subsume-transitive}. We prove the result by induction on the length of the derivation $\R \normstep[\Rn,\En]^* \R'$. 

\paragraph{Base case $\R = \R'$:} The result trivially holds.

\paragraph{Inductive step $\R \normstep[\Rn,\En]^* \R'' \normstep[\Rn,\En] \R'$:} By inductive hypothesis, we know that $\R''$ subsumes $\R_0$ w.r.t. $\Rn$. By definition, $\R_0 \normstep[\Rn,\En]^* \R'_0$ and for all $(\ell \rightarrow r) \in \R_0$, either $\ell =_\En r$ or there exists $(\ell' \rightarrow r') \in \R''$, a term context $C[\_]$ and a substitution such that $C[\ell'\sigma] \eqES{\En} \ell$ and $C[r'\sigma] \eqE{\En} r$ and $r' \not\rwstepC[\En]{\Rn}$ and $\forall p > \varepsilon. \ell'_{|p} \not\rwstepC[\En]{\Rn}$. Hence, we only need to look at the case where $(\ell' \rightarrow r') \not\in \R'$ (in the other cases, the result would hold). We do a case analysis on the normalisation rule applied.
\begin{itemize}
\item \emph{Rule \RNormL:} Cannot be the case as $\forall p > \varepsilon. \ell'_{|p} \not\rwstepC[\En]{\Rn}$
\item \emph{Rule \RNormR:} Cannot be the case as $r' \not\rwstepC[\En]{\Rn}$
\item \emph{Rule \RSubsume:} In such a case, there exist $(\ell'' \rightarrow r'') \in \R'$, a term context $C'[\_]$ and a substitution $\sigma'$ such that $C'[r''\sigma'] \eqE{\En} r'$ and $C'[\ell''\sigma'] \eqES{\En} \ell'$ and $r'' \not\rwstepC[\En]{\Rn}$ and $\forall p > \varepsilon. \ell''_{|p} \not\rwstepC[\En]{\Rn}$. Hence $\ell \eqES{\En} C[C'\sigma[\ell''\sigma'\sigma]]$ and $r \eqE{\En} C[C'\sigma[r''\sigma'\sigma]]$. Defining $C''[\_] = C[C'\sigma[\_]]$ and $\sigma'' = \sigma'\sigma$, the result holds.
\item \emph{Rule \REq:} In such a case, $\ell' =_\En r'$. Hence $C[\ell'\sigma] =_\En C[r'\sigma]$ which implies $\ell =_\En r$. 
\end{itemize}
This concludes the proof of \Cref{item:properties-normalisation-subsume-transitive}. 

We finally focus on \Cref{item:properties-normalisation-split}. Once again, we prove the result by induction on the length of the derivation $\R \normstep[\Rn,\En]^* \R'$. 

\paragraph{Base case $\R = \R'$:} The result trivially holds by taking $\R_{in} = \{ \ell \rightarrow r\}$ and $\R_{rm} = \emptyset$.

\paragraph{Inductive step $\R \normstep[\Rn,\En]^* \R'' \normstep[\Rn,\En] \R'$:} By inductive hypothesis, we know that for all $(\ell \rightarrow r) \in \R$, there exist $\R''_{in}, \R''_{rm}$ such that $\{ \ell \rightarrow r \} \normstep[\Rn,\En]^* \R''_{in} \cup \R''_{rm}$ and $\R''_{in} \subseteq \R''$ and $\R''$ subsumes $\R_{rm}$ w.r.t. $\Rn$. If the rule applied in $\R'' \normstep[\Rn,\En] \R'$ does not affect rules in $\R''_{in}$ then it would imply that $\R''_{in} \subseteq \R'$ and by \Cref{item:properties-normalisation-subsume-transitive}, we know that $\R'$ subsumes $\R''_{rm}$, which allows us to conclude. If the rule applied in $\R'' \normstep[\Rn,\En] \R'$ affect a rule in $\R''_{in}$, say $\ell' \rightarrow r'$, i.e. $\R''_{in} = \R'_{in} \cup \{ \ell'\rightarrow r'\}$ we can do a simple case analysis:
\begin{itemize}
\item \emph{Rule \RNormL or \RNormR:} In such a case, $\R'' = \R_1 \cup \{ \ell' \rightarrow r'\}$ and $\R' = \R_1 \cup \R_{add}$ for some $\R_{add}$. We conclude by defining $\R_{rm} = \R''_{rm}$ and $\R_{in} = \R'_{in} \cup \R_{add}$ and applying \Cref{item:properties-normalisation-subsume-transitive}.
\item \emph{Rule \RSubsume or \REq:} In such a case $\R'' = \R' \cup \{ \ell' \rightarrow r'\}$. We conclude by defining $\R_{rm} = \R''_{rm} \cup \{ \ell' \rightarrow r'\}$ and $\R_{in} = \R'_{in}$ and applying \Cref{item:properties-normalisation-subsume-transitive}.\qedhere
\end{itemize}
\end{proof}


\begin{lemma}
\label{lem:normalisation-idempotent}
Let $\R_1, \ldots, \R_n$, $\R'_1, \ldots, \R'_{n+1}$ and $\Rn$ be sets of rewrite rules such that for all $i \in \{1, \ldots, n\}$, $\R'_i \cup \R_i \normstep[\Rn,\En]^* \R'_{i+1}$. Then for all $i \in \{0, \ldots, n\}$, $\R'_{i+1} \cup \bigcup_{j=1}^i \R_j \normstep[\Rn,\En]^* \R'_{i+1}$.
\end{lemma}

\begin{proof}
We prove the result by induction on $i$. 

\paragraph{Base case $i=0$:} In such a case, $\R'_{i+1} \cup \bigcup_{j=1}^i \R_j = \R'_1$. Hence we trivially have that $\R'_1 \normstep[\Rn,\En]^* \R'_1$.

\paragraph{Inductive step $i>0$:} By inductive hypothesis, $\R'_i \cup \bigcup_{j=1}^{i-1} \R_j \normstep[\Rn,\En]^* \R'_i$. Moreover, by hypothesis, $\R'_i \cup \R_i \normstep[\Rn,\En]^* \R'_{i+1}$. 
Therefore, by \Cref{lem:properties-normalisation}, for all $(\ell \rightarrow r) \in  \bigcup_{j=1}^{i-1} \R_j$, there exist $\R_{rm}, \R_{in}$ such that $\{\ell \rightarrow r\} \normstep[\Rn,\En]^* \R_{rm} \cup \R_{in}$ such that $\R_{in} \subseteq \R'_i$ and $\R'_i$ subsumes $\R_{rm}$ w.r.t. $\Rn$. 
As $\R'_i \cup \R_i \normstep[\Rn,\En]^* \R'_{i+1}$, then, denoting the rules in $\R_{in}$ by $\{ \ell_1 \rightarrow r_1, \ldots, \ell_m \rightarrow r_m\}$, we deduce that $\R'_{i+1}$ subsumes $\R_{rm}$ w.r.t. $\Rn$ and there exist $\R^1_{rm}, \R^1_{in}, \ldots, \R^m_{rm}, \R^m_{in}$ such that for all $j \in \{1, \ldots, m\}$, $\{\ell_j \rightarrow r_j\} \normstep[\Rn,\En]^* \R^j_{rm} \cup \R^j_{in}$, $\R^j_{in} \subseteq \R'_{i+1}$ and $\R^j_{rm}$ is subsumed by $\R'_{i+1}$. Hence, 
\[
\R'_{i+1} \cup \{\ell \rightarrow r\} \normstep[\Rn,\En]^* \R'_{i+1} \cup \R_{rm} \cup \bigcup_{j=1}^m \R^j_{rm} \normstep[\Rn,\En]^* \R'_{i+1}
\]
Thus, we deduce that $\R'_{i+1} \cup \bigcup_{j=1}^{i-1} \R_j \normstep[\Rn,\En]^* \R'_{i+1}$. Finally, as $\R'_i \cup \R_i \normstep[\Rn,\En]^* \R'_{i+1}$, we deduce once again by \Cref{lem:properties-normalisation} that for all $(\ell \rightarrow r) \in \R_i$, there exist $\R_{rm}, \R_{in}$ such that $\{\ell \rightarrow r\} \normstep[\Rn,\En]^* \R_{rm} \cup \R_{in}$ such that $\R_{in} \subseteq \R'_{i+1}$ and $\R'_{i+1}$ subsumes $\R_{rm}$ w.r.t. $\Rn$. Therefore, we deduce that $\R'_{i+1} \cup \{\ell \rightarrow r\} \normstep[\Rn,\En]^* \R'_{i+1}$ and so $\R'_{i+1} \cup \bigcup_{j=1}^i \R_j \normstep[\Rn,\En]^* \R'_{i+1}$.
\end{proof}


\subsubsection{Overlapping transformations}

In this section, we look at the cases of two successive rules that overlap.

\begin{lemma}
\label{lem:overlap-left-right}
Let $\R,\Rn$ be sets of rewrite rules such that $\R$ is saturated w.r.t. $\Rn$. Let $u \leftrwstep{p}{\sigma}{\ell}{r} t \rightrwstep[\En]{p'}{\sigma'}{\ell'}{r'} v$ be some ground rewrite steps. Let $\tr_0 = \leftrwlabel{p}{\sigma}{\ell}{r} \rightrwlabel{p'}{\sigma'}{\ell'}{r'}[\En]$ with $\{\ell \rightarrow r, \ell'\rightarrow r'\} \subseteq \R$. Let $\M$ be a set of terms that rewrites into $u \Leftrightrwstep{\tr_0} v$.

If $p' = p \cdot q$, and $q \in \PosX{\ell}$, and $\ell\sigma > r\sigma$, and $\ell'\sigma' > r'\sigma'$ then there exists a ground syntactic rewrite trace $\tr_1$ from $\R \cup \En$ such that: 
\begin{enumerate}[label=\textbf{P\arabic*}]
\item $u \rwLeftrightStep{\tr_1} v$, and \label{item:overlap-left-right-Leftright}
\item $\M$ is greater than $u \rwLeftrightStep{\tr_1} v$, and \label{item:overlap-left-right-derive}
\item $\measure{\tr_1} \lessRwLbl{\Pset(\M)} \measure{\tr_0}$\label{item:overlap-left-right-measure}
\end{enumerate}
\end{lemma}

\begin{proof}
By definition, there exists a term context $C[\_]$ such that $t = C[\ell\sigma]$, $u = C[r\sigma]$, $C[\_]_{|p} = \_$, $\ell_{|q}\sigma =_\En \ell'\sigma'$, $v = C[r\sigma[r'\sigma']_{q}]$. Since $\ell_{|q}\sigma =_\En \ell'\sigma'$, we can denote $\gamma = \sigma \cup \sigma'$ to obtain $\ell_{|q}\gamma =_\En \ell'\gamma$. Hence $\ell_{|q}$ and $\ell'$ are $\En$-unifiable.
Let $\theta \in \mgu{\En}{\ell_{|q},\ell'}$. There exists $\sigma_3$ such that $\gamma =_\En (\theta\sigma_3)_{\dom{\gamma}}$. Moreover, we also know that there exists $\ell_3 \rightarrow r_3 \in (\overlapset{r \rightarrow \ell}{q}{\ell' \rightarrow r'})$ such that $\ell_3 = r\theta$, $r_3 = \ell\theta[r'\theta]_q$. Hence $\ell_3\sigma_3 =_\En r\gamma = r\sigma$. We also have $r_3\sigma_3 =_\En \ell\gamma[r'\gamma]_q = \ell\sigma[r'\sigma']_q$. 

As $\gamma =_\En (\theta\sigma_3)_{\dom{\gamma}}$, we deduce that there exists a syntactic rewrite trace $\tr_L$ from $\En$ such that $r\sigma \Rightrwstep{\tr_L} \ell_3\sigma_3$. Similarly, there exists a syntactic rewrite trace $\tr_R$ from $\En$ such that $r_3\sigma_3 \Rightrwstep{\tr_R} \ell\sigma[r'\sigma']_q$. We therefore obtain the following sequence:
\[
u = C[r\sigma] \Rightrwstep{p \cdot \tr_L} C[\ell_3\sigma_3] \Leftrightrwstep{\omega_3} C[r_3\sigma_3] \Rightrwstep{p\cdot \tr_R} C[\ell\sigma[r'\sigma']_q]  = v
\]
with $\omega_3 = \rightrwlabel{p}{\sigma_3}{\ell_3}{r_3}$ if $\ell_3\sigma_3 \geq r_3\sigma_3$ and $\omega_3 = \leftrwlabel{p}{\sigma_3}{r_3}{\ell_3}$ otherwise.
By defining $\tr'_0 = (p\cdot \tr_L) \omega_3 (p\cdot \tr_R)$, we obtain that $u \rwLeftrightStep{\tr'_0} v$. Moreover, as $\M$ is greater than $u \rwLeftrightStep{\tr_0} v$, we obtain by construction of $\tr'_0$ that $\M$ is greater than $u \rwLeftrightStep{\tr'_0} v$.

We now order the labels of $\tr'_0$ w.r.t. $\tr_0$. Let us denote $\omega = \leftrwlabel{p}{\sigma}{\ell}{r}$ and $\omega' = \rightrwlabel{p'}{\sigma'}{\ell'}{r'}[\En]$.
Since $\ell'\sigma' > r'\sigma'$ and $\ell_{|q}\sigma =_\En \ell'\sigma'$, we deduce that $\ell\sigma > r_3\sigma_3$. We do a case analysis on the comparison between $\ell_3\sigma_3$ and $r_3\sigma_3$. 
\begin{itemize}
\item Case $\ell_3\sigma_3 \geq r_3\sigma_3$: Since $\ell_3\sigma_3 =_\En r\sigma$, we deduce that $\measure{\omega_3} = (p,r_3\sigma_3,\ell_3\sigma_3) \lessRwLbl{\Pset(\M)} (p,r\sigma,\ell\sigma) = \measure{\omega}$. Moreover, for all $\omega_1 = \rightrwlabel{p''}{\alpha}{a}{b} \in p\cdot \tr_L$, $p \leq p''$ and if $p = p''$ then $a\alpha =_\En b\alpha =_\En r\sigma$. Since $\ell\sigma > r\sigma$, we deduce that $\measure{\omega_1} \lessRwLbl{\Pset(\M)} (p,r\sigma,\ell\sigma) = \measure{\omega}$. Finally, for all $\omega_1 = \rightrwlabel{p''}{\alpha}{a}{b} \in p\cdot \tr_R$, $p \leq p''$ and if $p = p''$ then $a\alpha =_\En b\alpha =_\En r_3\sigma_3$. Recall that $\ell_3\sigma_3 =_\En r\sigma$ thus when $p = p''$ then $\ell\sigma > r\sigma =_\En \ell_3\sigma_3 \geq r_3\sigma_3 =_\En a\alpha =_\En b\alpha$. We deduce that $\measure{\omega_1} \lessRwLbl{\Pset(\M)} (p,r\sigma,\ell\sigma) = \measure{\omega}$. We therefore conclude that for all $\omega_1 \in \tr'_0$, there exist $\omega_0 \in \tr_0$ such that $\measure{\omega_1} \lessRwLbl{\Pset(\M)} \measure{\omega_0}$.

We can represent graphically the transformation as follows:
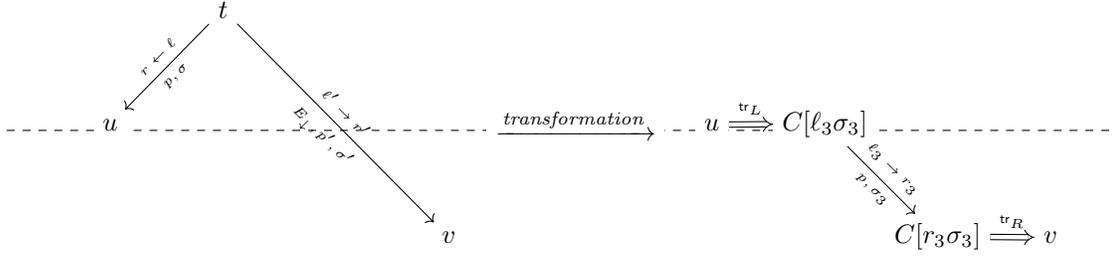
\begin{center}
    \begin{tikzpicture}[
        term/.style={}
        ]
        \def\length{1.5cm}
        \node[anchor=mid] (L) {};
        \node[right=15cm of L.mid,anchor=mid] (R) {};
        \draw[dashed] (L) -- (R);

        \node[fill=white,right=6.5cm of L.mid] (Label) {$\xrightarrow{\mathit{transformation}}$};

        \node[fill=white,term,right=\length of L.mid,anchor=mid] (T0) {$u$};
        \node[fill=white,term,above right=\length and \length of T0.mid,anchor=mid] (T1) {$t$};
        \node[fill=white,term,below right=2*\length and 2*\length of T1.mid,anchor=mid] (T2) {$v$};

        \draw[<-] (T0) edge node[auto,sloped] {\tiny$r \leftarrow \ell$} node[auto,sloped,below] {\tiny$p, \sigma$} (T1);
        \draw[->] (T1) edge node[auto,sloped] {\tiny$\ell' \rightarrow r'$} node[auto,sloped,below] {\tiny$\En, p', \sigma'$} (T2);

        \node[term,fill=white,right=8cm of T0.mid,anchor=mid] (T0') {$u$};
        \node[term,fill=white,right= \length of T0'.mid,anchor=mid] (T1') {$C[\ell_3\sigma_3]$};
        \node[term,fill=white,below right=\length and \length of T1'.mid,anchor=mid] (T2') {$C[r_3\sigma_3]$};
        \node[term,fill=white,right= \length of T2'.mid,anchor=mid] (T3') {$v$};

        \draw (T0') edge[double equal sign distance,-Implies] node[auto,sloped] {\tiny$\tr_L$} (T1');
        \draw[->] (T1') edge node[auto,sloped] {\tiny$\ell_3 \rightarrow r_3$} node[auto,sloped,below] {\tiny$p, \sigma_3$} (T2');
        \draw (T2') edge[double equal sign distance,-Implies] node[auto,sloped,] {\tiny$\tr_R$} (T3');
    \end{tikzpicture}
\end{center}
\item Case $\ell_3\sigma_3 < r_3\sigma_3$: As $r\sigma =_\En \ell_3\sigma_3$ and $\ell\sigma > r_3\sigma_3$, we deduce that $\measure{\omega_3} = (p,\ell_3\sigma_3,r_3\sigma_3) \lessRwLbl{\Pset(\M)} (p,r\sigma,\ell\sigma) = \measure{\omega}$. Moreover, for all $\omega_1 = \rightrwlabel{p''}{\alpha}{a}{b} \in p\cdot \tr_L$, $p \leq p''$ and if $p = p''$ then $a\alpha =_\En b\alpha =_\En r\sigma$. Since $\ell\sigma > r\sigma$, we deduce that $\measure{\omega_1} \lessRwLbl{\Pset(\M)} (p,r\sigma,\ell\sigma) = \measure{\omega}$. Finally, for all $\omega_1 = \rightrwlabel{p''}{\alpha}{a}{b} \in p\cdot \tr_R$, $p \leq p''$ and if $p = p''$ then $a\alpha =_\En b\alpha =_\En r_3\sigma_3$ and $q = \varepsilon$. In such a case ($p = p''$), $r'\sigma' =_\En r_3\sigma_3$ and $\ell\sigma =_\En \ell'\sigma'$. Thus, $a\alpha =_\En b\alpha =_\En r'\sigma' < \ell'\sigma'$ and so $\measure{\omega_1} \lessRwLbl{\Pset(\M)} \measure{\omega'}$. We therefore conclude that for all $\omega_1 \in \tr'_0$, there exist $\omega_0 \in \tr_0$ such that $\measure{\omega_1} \lessRwLbl{\Pset(\M)} \measure{\omega_0}$.

We can represent graphically the transformation as follows:
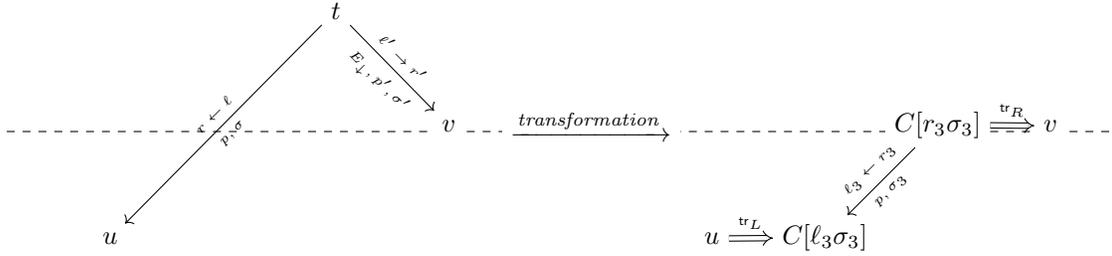
\begin{center}
    \begin{tikzpicture}[
        term/.style={}
        ]
        \def\length{1.5cm}
        \node[anchor=mid] (L) {};
        \node[right=15cm of L.mid,anchor=mid] (R) {};
        \draw[dashed] (L) -- (R);

        \node[fill=white,right=6.7cm of L.mid] (Label) {$\xrightarrow{\mathit{transformation}}$};

        \node[fill=white,term,below right=\length and \length of L.mid,anchor=mid] (T0) {$u$};
        \node[fill=white,term,above right=2*\length and 2*\length of T0.mid,anchor=mid] (T1) {$t$};
        \node[fill=white,term,below right=\length and \length of T1.mid,anchor=mid] (T2) {$v$};

        \draw[<-] (T0) edge node[auto,sloped] {\tiny$r \leftarrow \ell$} node[auto,sloped,below] {\tiny$p, \sigma$} (T1);
        \draw[->] (T1) edge node[auto,sloped] {\tiny$\ell' \rightarrow r'$} node[auto,sloped,below] {\tiny$\En, p', \sigma'$} (T2);

        \node[term,fill=white,right=8cm of T0.mid,anchor=mid] (T0') {$u$};
        \node[term,fill=white,right= \length of T0'.mid,anchor=mid] (T1') {$C[\ell_3\sigma_3]$};
        \node[term,fill=white,above right=\length and \length of T1'.mid,anchor=mid] (T2') {$C[r_3\sigma_3]$};
        \node[term,fill=white,right= \length of T2'.mid,anchor=mid] (T3') {$v$};

        \draw (T0') edge[double equal sign distance,-Implies] node[auto,sloped] {\tiny$\tr_L$} (T1');
        \draw[<-] (T1') edge node[auto,sloped] {\tiny$\ell_3 \leftarrow r_3$} node[auto,sloped,below] {\tiny$p, \sigma_3$} (T2');
        \draw (T2') edge[double equal sign distance,-Implies] node[auto,sloped,] {\tiny$\tr_R$} (T3');
    \end{tikzpicture}
\end{center}
\end{itemize}
In both cases, we have shown that for all $\omega_1 \in \tr'_0$, there exist $\omega_0 \in \tr_0$ such that $\measure{\omega_1} \lessRwLbl{\Pset(\M)} \measure{\omega_0}$.

Notice that, denoting $\R' = (\R \cup \Rn \cup (\overlapseteq{r \rightarrow \ell}{q}{\ell' \rightarrow r'}))$, we have that $\tr'_0$ is a syntactic rewrite trace from $\R' \cup \En$. Since $\tr_0$ is from $\R$, it is naturally also from $\R' \cup \En$. Note that by \Cref{def:saturated}, we know that $\R' \normstep[\Rn,\En]^* \R$ and $\Rn \subseteq \R'$. Thus, we can apply \Cref{cor:normalisation} to deduce that there exists a syntactic rewrite trace $\tr_1$ from $\R \cup \En$ such that $u \Leftrightrwstep{\tr_1} v$, which is smaller than $\M$, and $\measure{\tr_1} \lessRwLbl{\Pset(\M)} \measure{\tr_0}$. This conclude the proof.
\end{proof}


\begin{lemma}
\label{lem:overlap-En-right}
Assume that $\En$ has finite equivalence classes. Let $\R,\Rn$ be sets of rewrite rules such that $\R$ is saturated w.r.t. $\Rn$. Let $u \rightrwstep{p}{\sigma}{\ell}{r} t \rightrwstep[\En]{p'}{\sigma'}{\ell'}{r'} v$ be some ground be rewrite steps. Let $\tr_0 = \rightrwlabel{p}{\sigma}{\ell}{r} \rightrwlabel{p'}{\sigma'}{\ell'}{r'}[\En]$ with $(\ell \rightarrow r) \in \En$ and $(\ell'\rightarrow r') \in \R$. Let $\M$ be a set of terms greater than $u \Leftrightrwstep{\tr_0} v$.

If $p' = p \cdot q$, and $q \in \PosX{r}$, and $\ell'\sigma' > r'\sigma'$ then there exists a ground syntactic rewrite trace $\tr_1$ from $\En \cup \R$ such that 
\begin{enumerate}[label=\textbf{P\arabic*}]
\item $u \rwLeftrightStep{\tr_1} v$, and \label{item:overlap-En-right-Leftright}
\item $\M$ is greater than $u \rwLeftrightStep{\tr_1} v$, and \label{item:overlap-En-right-derive}
\item $\measure{\tr_1} \lessRwLbl{\Pset(\M)} \measure{\tr_0}$\label{item:overlap-En-right-measure}
\end{enumerate}
\end{lemma}

\begin{proof}
By definition, there exists a term context $C[\_]$ such that $t = C[r\sigma]$, $u = C[\ell\sigma]$, $C[\_]_{|p} = \_$, $r_{|q}\sigma =_\En \ell'\sigma'$, $v = C[\ell\sigma[r'\sigma']_{q}]$. Since $r_{|q}\sigma =_\En \ell'\sigma'$, we can denote $\gamma = \sigma \cup \sigma'$ to obtain $r_{|q}\gamma =_\En \ell'\gamma$. Hence $r_{|q}$ and $\ell'$ are $\En$-unifiable.
Let $\theta \in \mgu{\En}{r_{|q},\ell'}$. There exists $\sigma_3$ such that $\gamma =_\En (\theta\sigma_3)_{\dom{\gamma}}$. Moreover, we also know that there exists $\ell_3 \rightarrow r_3 \in (\overlapset{\ell \rightarrow r}{q}{\ell' \rightarrow r'})$ such that $\ell_3 = \ell\theta$, $r_3 = r\theta[r'\theta]_q$. Hence $\ell_3\sigma_3 =_\En \ell\gamma = \ell\sigma$. We also have $r_3\sigma_3 =_\En r\gamma[r'\gamma]_q = r\sigma[r'\sigma']_q$. 

As $\gamma =_\En (\theta\sigma_3)_{\dom{\gamma}}$, we deduce that there exists a syntactic rewrite trace $\tr_L$ from $\En$ such that $\ell\sigma \Rightrwstep{\tr_L} \ell_3\sigma_3$. Since $(\ell \rightarrow r) \in \En$ and $\En$ has finite equivalence classes, we deduce that $\ell$ and $r$ are not variables. Therefore, for all $\omega \in \tr_L$, $\pos(\omega) \neq \varepsilon$.
Similarly, there exists a syntactic rewrite trace $\tr_R$ from $\En$ such that $r_3\sigma_3 \Rightrwstep{\tr_R} r\sigma[r'\sigma']_q$ and if $q \neq \varepsilon$ then for all $\omega \in \tr_R$, $\pos(\omega) \neq \varepsilon$. As $>$ is a $\En$-strong reduction ordering and $\ell'\sigma' > r'\sigma"$, we deduce $\ell_3\sigma_3 =_\En \ell\sigma =_\En r\sigma > r\sigma[r'\sigma']_q =_\En r_3\sigma_3$. 
We therefore obtain the following sequence:
\[
u = C[\ell\sigma] \Rightrwstep{p \cdot \tr_L} C[\ell_3\sigma_3] \Rightrwstep{\omega_3} C[r_3\sigma_3] \Rightrwstep{p\cdot \tr_R} C[r\sigma[r'\sigma']_q]  = v
\]
with $\omega_3 = \rightrwlabel{p}{\sigma_3}{\ell_3}{r_3}$.
By defining $\tr'_0 = (p\cdot \tr_L) \omega_3 (p\cdot \tr_R)$, we obtain that $u \rwLeftrightStep{\tr'_0} v$. Moreover, as $\M$ rewrites into $u \rwLeftrightStep{\tr_0} v$, we obtain by construction of $\tr'_0$ that $\M$ rewrites into $u \rwLeftrightStep{\tr'_0} v$.

We now order the labels of $\tr'_0$ w.r.t. $\tr_0$. Let us denote $\omega = \rightrwlabel{p}{\sigma}{\ell}{r}$ and $\omega' = \rightrwlabel{p'}{\sigma'}{\ell'}{r'}[\En]$.
We already proved that $\ell_3\sigma_3 =_\En \ell\sigma =_\En r\sigma > r_3\sigma_3$. Thus $\measure{\omega_3} \lessRwLbl{\Pset(\M)} \measure{\omega}$. 
Moreover, we already proved that for all $\omega_1$ in $\tr_L$, $\pos(\omega_1) \neq \varepsilon$. Hence for all $\omega_1$ in $p\cdot \tr_L$, $p < \pos(\omega_1)$, implying that $\measure{\omega_1} \lessRwLbl{\Pset(\M)} \measure{\omega}$. Finally, if $q \neq \varepsilon$ then we also showed that for all $\omega_1$ in $\tr_R$, $\pos(\omega_1) \neq \varepsilon$, and so for all $\omega_1$ in $p\cdot \tr_R$, $p < \pos(\omega_1)$, implying that $\measure{\omega_1} \lessRwLbl{\Pset(\M)} \measure{\omega}$. We now focus on $q = \varepsilon$. In such a case, $p = p'$ and 
for all $\omega_1 = \rightrwlabel{p''}{\alpha}{a}{b} \in p\cdot\tr_R$, if $p'' = p$ then $a\alpha =_\En b\alpha =_\En r'\sigma'$ else $p < p''$. Hence $\measure{\omega_1} = (p'',b\alpha,a\alpha) \lessRwLbl{\Pset(\M)} (p',r'\sigma',\ell'\sigma') = \measure{\omega'}$.
We therefore conclude that for all $\omega_1 \in \tr'_0$, there exist $\omega_0 \in \tr_0$ such that $\measure{\omega_1} \lessRwLbl{\Pset(\M)} \measure{\omega_0}$.

We can represent graphically the transformation as follows:
\begin{center}
    \begin{tikzpicture}[
        term/.style={}
        ]
        \def\length{1.5cm}
        \node[anchor=mid] (L) {};
        \node[right=16cm of L.mid,anchor=mid] (R) {};
        \draw[dashed] (L) -- (R);

        \node[fill=white,right=5cm of L.mid] (Label) {$\xrightarrow{\mathit{transformation}}$};

        \node[fill=white,term,above right=\length and \length of L.mid,anchor=mid] (T0) {$u$};
        \node[fill=white,term,right=\length of T0.mid,anchor=mid] (T1) {$t$};
        \node[fill=white,term,below right=\length and \length of T1.mid,anchor=mid] (T2) {$v$};

        \draw[->] (T0) edge node[auto,sloped] {\tiny$\ell \rightarrow r$} node[auto,sloped,below] {\tiny$p, \sigma$} (T1);
        \draw[->] (T1) edge node[auto,sloped] {\tiny$\ell' \rightarrow r'$} node[auto,sloped,below] {\tiny$\En, p', \sigma'$} (T2);

        \node[term,fill=white,right=6cm of T0.mid,anchor=mid] (T0') {$u$};
        \node[term,fill=white,right=1.3*\length of T0'.mid,anchor=mid] (T1') {$C[\ell_3\sigma_3]$};
        \node[term,fill=white,below right= \length and \length of T1'.mid,anchor=mid] (T2') {$C[r_3\sigma_3]$};
        \node[term,fill=white,right= 1.3*\length of T2'.mid,anchor=mid] (T3') {$v$};

        \draw (T0') edge[double equal sign distance,-Implies] node[auto,sloped] {\tiny$p \cdot \tr_L$} (T1');
        \draw[->] (T1') edge node[auto,sloped] {\tiny$\ell_3\sigma_3 \rightarrow r_3\sigma_3$} node[auto,sloped,below] {\tiny$p, \sigma_3$} (T2');
        \draw (T2') edge[double equal sign distance,-Implies] node[auto,sloped] {\tiny$p\cdot \tr_R$} (T3');
    \end{tikzpicture}
\end{center}
Notice that, denoting $\R' = (\R \cup \Rn \cup (\overlapset{\ell \rightarrow r}{q}{\ell' \rightarrow r'}))$, we have that $\tr'_0$ is a syntactic rewrite trace from $\R' \cup \En$. Since $\tr_0$ is from $\R$, it is naturally also from $\R' \cup \En$. Note that by \Cref{def:saturated}, we know that $\R' \normstep[\Rn,\En]^* \R$ and $\Rn \subseteq \R'$. Thus, we can apply \Cref{cor:normalisation} to deduce that there exists a syntactic rewrite trace $\tr_1$ from $\R \cup \En$ such that $u \Leftrightrwstep{\tr_1} v$ which $\M$ rewrites into, and $\measure{\tr_1} \lessRwLbl{\Pset(\M)} \measure{\tr_0}$. This conclude the proof.
\end{proof}


\begin{lemma}
\label{lem:overlap-right-right}
Let $\R,\Rn$ be sets of rewrite rules such that $\R$ is saturated w.r.t. $\Rn$. Let $u \rightrwstep{p}{\sigma}{\ell}{r} t \rightrwstep[\En]{p'}{\sigma'}{\ell'}{r'} v$ be some ground rewrite steps. Let $\tr_0 = \rightrwlabel{p}{\sigma}{\ell}{r} \rightrwlabel{p'}{\sigma'}{\ell'}{r'}[\En]$ with $(\ell \rightarrow r), (\ell'\rightarrow r') \in \R$. Let $\M$ be a set of terms greater than $u \Leftrightrwstep{\tr_0} v$.

If $p' = p \cdot q$, and $q \in \PosX{r}$, and $\ell\sigma > r\sigma$, and $\ell'\sigma' > r'\sigma'$ and ($\ell \in \X$ implies $\ell \in \vars{r}$) then there exists a ground syntactic rewrite trace $\tr_1$ from $\En \cup \R$ such that 
\begin{enumerate}[label=\textbf{P\arabic*}]
\item $u \rwLeftrightStep{\tr_1} v$, and \label{item:overlap-right-right-Leftright}
\item $\M$ is greater $u \rwLeftrightStep{\tr_1} v$, and \label{item:overlap-right-right-derive}
\item $\measure{\tr_1} \lessRwLbl{\Pset(\M)} \measure{\tr_0}$\label{item:overlap-right-right-measure}
\end{enumerate}
\end{lemma}

\begin{proof}
By definition, there exists a term context $C[\_]$ such that $t = C[r\sigma]$, $u = C[\ell\sigma]$, $C[\_]_{|p} = \_$, $r_{|q}\sigma =_\En \ell'\sigma'$, $v = C[\ell\sigma[r'\sigma']_{q}]$. Since $r_{|q}\sigma =_\En \ell'\sigma'$, we can denote $\gamma = \sigma \cup \sigma'$ to obtain $r_{|q}\gamma =_\En \ell'\gamma$. Hence $r_{|q}$ and $\ell'$ are $\En$-unifiable.
Let $\theta \in \mgu{\En}{r_{|q},\ell'}$. There exists $\sigma_3$ such that $\gamma =_\En (\theta\sigma_3)_{\dom{\gamma}}$. Moreover, we also know that there exists $\ell_3 \rightarrow r_3 \in (\overlapset{\ell \rightarrow r}{q}{\ell' \rightarrow r'})$ such that $\ell_3 = \ell\theta$, $r_3 = r\theta[r'\theta]_q$. Hence $\ell_3\sigma_3 =_\En \ell\gamma = \ell\sigma$. We also have $r_3\sigma_3 =_\En r\gamma[r'\gamma]_q = r\sigma[r'\sigma']_q$. 

As $\gamma =_\En (\theta\sigma_3)_{\dom{\gamma}}$, we deduce that there exists a syntactic rewrite trace $\tr_L$ from $\En$ such that $\ell\sigma \Rightrwstep{\tr_L} \ell_3\sigma_3$. 
By hypothesis, we know that $\ell \in \X$ implies $\ell \in \vars{r}$. However,  $>$ is a reduction ordering, hence if $a$ is a strict subterm of $b$ then $a \not> b$ (otherwise we would have $b = C'[a]$ and so $a > C'[a] > C'[C'[a]] > \ldots$). Note that if $a = b$ then we also have $a \not> b$. Thus, as $\ell\sigma > r\sigma$, we deduce that $\ell\not\in \X$, meaning that for all $\omega \in \tr_L$, $\pos(\omega) \neq \varepsilon$.

Similarly, there exists a syntactic rewrite trace $\tr_R$ from $\En$ such that $r_3\sigma_3 \Rightrwstep{\tr_R} r\sigma[r'\sigma']_q$ and if $q \neq \varepsilon$ then for all $\omega \in \tr_R$, $\pos(\omega) \neq \varepsilon$. As $>$ is a $\En$-strong reduction ordering and $\ell'\sigma' > r'\sigma"$, we deduce $\ell_3\sigma_3 =_\En \ell\sigma > r\sigma > r\sigma[r'\sigma']_q =_\En r_3\sigma_3$. 
We therefore obtain the following sequence:
\[
u = C[\ell\sigma] \Rightrwstep{p \cdot \tr_L} C[\ell_3\sigma_3] \Rightrwstep{\omega_3} C[r_3\sigma_3] \Rightrwstep{p\cdot \tr_R} C[r\sigma[r'\sigma']_q]  = v
\]
with $\omega_3 = \rightrwlabel{p}{\sigma_3}{\ell_3}{r_3}$.
By defining $\tr'_0 = (p\cdot \tr_L) \omega_3 (p\cdot \tr_R)$, we obtain that $u \rwLeftrightStep{\tr'_0} v$. Moreover, as $\M$ is greater than $u \rwLeftrightStep{\tr_0} v$, we obtain by construction of $\tr'_0$ that $\M$ is greater than $u \rwLeftrightStep{\tr'_0} v$.

We now order the labels of $\tr'_0$ w.r.t. $\tr_0$. Let us denote $\omega = \rightrwlabel{p}{\sigma}{\ell}{r}$ and $\omega' = \rightrwlabel{p'}{\sigma'}{\ell'}{r'}[\En]$.
We already proved that $\ell_3\sigma_3 =_\En \ell\sigma > r\sigma > r_3\sigma_3$. Thus $\measure{\omega_3} \lessRwLbl{\Pset(\M)} \measure{\omega}$. 
Moreover, we already proved that for all $\omega_1$ in $\tr_L$, $\pos(\omega_1) \neq \varepsilon$. Hence for all $\omega_1$ in $p\cdot \tr_L$, $p < \pos(\omega_1)$, implying that $\measure{\omega_1} \lessRwLbl{\Pset(\M)} \measure{\omega}$. Finally, if $q \neq \varepsilon$ then we also showed that for all $\omega_1$ in $\tr_R$, $\pos(\omega_1) \neq \varepsilon$, and so for all $\omega_1$ in $p\cdot \tr_R$, $p < \pos(\omega_1)$, implying that $\measure{\omega_1} \lessRwLbl{\Pset(\M)} \measure{\omega}$. We now focus on $q = \varepsilon$. In such a case, $p = p'$ and 
for all $\omega_1 = \rightrwlabel{p''}{\alpha}{a}{b} \in p\cdot\tr_R$, if $p'' = p$ then $a\alpha =_\En b\alpha =_\En r'\sigma'$ else $p < p''$. Hence $\measure{\omega_1} = (p'',b\alpha,a\alpha) \lessRwLbl{\Pset(\M)} (p',r'\sigma',\ell'\sigma') = \measure{\omega'}$.
We therefore conclude that for all $\omega_1 \in \tr'_0$, there exist $\omega_0 \in \tr_0$ such that $\measure{\omega_1} \lessRwLbl{\Pset(\M)} \measure{\omega_0}$.

We can represent graphically the transformation as follows:
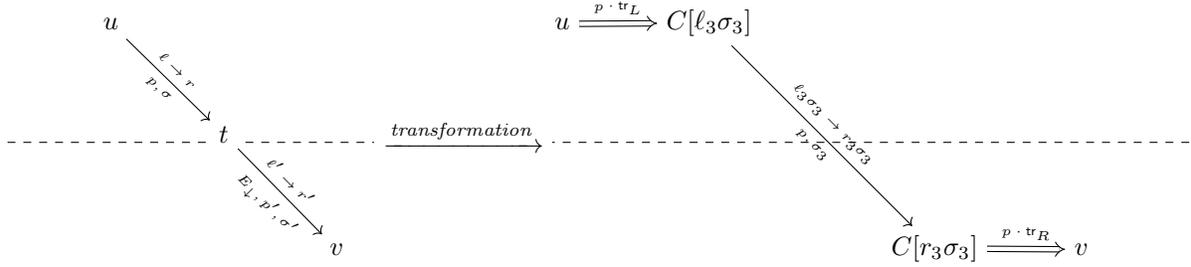
\begin{center}
    \begin{tikzpicture}[
        term/.style={}
        ]
        \def\length{1.5cm}
        \node[anchor=mid] (L) {};
        \node[right=16cm of L.mid,anchor=mid] (R) {};
        \draw[dashed] (L) -- (R);

        \node[fill=white,right=5cm of L.mid] (Label) {$\xrightarrow{\mathit{transformation}}$};

        \node[fill=white,term,above right=\length and \length of L.mid,anchor=mid] (T0) {$u$};
        \node[fill=white,term,below right=\length and \length of T0.mid,anchor=mid] (T1) {$t$};
        \node[fill=white,term,below right=\length and \length of T1.mid,anchor=mid] (T2) {$v$};

        \draw[->] (T0) edge node[auto,sloped] {\tiny$\ell \rightarrow r$} node[auto,sloped,below] {\tiny$p, \sigma$} (T1);
        \draw[->] (T1) edge node[auto,sloped] {\tiny$\ell' \rightarrow r'$} node[auto,sloped,below] {\tiny$\En, p', \sigma'$} (T2);

        \node[term,fill=white,right=6cm of T0.mid,anchor=mid] (T0') {$u$};
        \node[term,fill=white,right=1.3*\length of T0'.mid,anchor=mid] (T1') {$C[\ell_3\sigma_3]$};
        \node[term,fill=white,below right= 2*\length and 2*\length of T1'.mid,anchor=mid] (T2') {$C[r_3\sigma_3]$};
        \node[term,fill=white,right= 1.3*\length of T2'.mid,anchor=mid] (T3') {$v$};

        \draw (T0') edge[double equal sign distance,-Implies] node[auto,sloped] {\tiny$p \cdot \tr_L$} (T1');
        \draw[->] (T1') edge node[auto,sloped] {\tiny$\ell_3\sigma_3 \rightarrow r_3\sigma_3$} node[auto,sloped,below] {\tiny$p, \sigma_3$} (T2');
        \draw (T2') edge[double equal sign distance,-Implies] node[auto,sloped] {\tiny$p\cdot \tr_R$} (T3');
    \end{tikzpicture}
\end{center}
Notice that, denoting $\R' = (\R \cup \Rn \cup (\overlapset{\ell \rightarrow r}{q}{\ell' \rightarrow r'}))$, we have that $\tr'_0$ is a syntactic rewrite trace from $\R' \cup \En$. Since $\tr_0$ is from $\R$, it is naturally also from $\R' \cup \En$. Note that by \Cref{def:saturated}, we know that $\R' \normstep[\Rn,\En]^* \R$ and $\Rn \subseteq \R'$. Thus, we can apply \Cref{cor:normalisation} to deduce that there exists a syntactic rewrite trace $\tr_1$ from $\R \cup \En$ such that $u \Leftrightrwstep{\tr_1} v$ which is smaller than $\M$ , and $\measure{\tr_1} \lessRwLbl{\Pset(\M)} \measure{\tr_0}$. This conclude the proof.
\end{proof}


\subsubsection{Non-overlapping transformations}

In the previous section, we look at the cases of two successive non-overlapping rules.


\begin{lemma}
\label{lem:nonoverlap-left-right}
Let $\R,\Rn$ be sets of rewrite rules such that $\R$ is saturated w.r.t. $\Rn$. Let $u \leftrwstep{p}{\sigma}{\ell}{r} t \rightrwstep[\En]{p'}{\sigma'}{\ell'}{r'} v$ be some ground rewrite steps. Let $\tr_0 = \leftrwlabel{p}{\sigma}{\ell}{r} \rightrwlabel{p'}{\sigma'}{\ell'}{r'}[\En]$ with $(\ell \rightarrow r), (\ell'\rightarrow r') \in \R$. Let $\M$ be a set of terms greater than $u \Leftrightrwstep{\tr_0} v$.

If $p' = p \cdot q$, and $q \not\in \PosX{\ell}$, and $\ell'\sigma' > r'\sigma'$, and $\ell\sigma > r\sigma$ then there exists a ground syntactic rewrite trace $\tr_1$ from $\En \cup \R$ such that 
\begin{enumerate}[label=\textbf{P\arabic*}]
\item $u \rwLeftrightStep{\tr_1} v$, and \label{item:nonoverlap-left-right-Leftright}
\item $\M$ is greater than $u \rwLeftrightStep{\tr_1} v$, and \label{item:nonoverlap-left-right-derive}
\item $\measure{\tr_1} \lessRwLbl{\Pset(\M)} \measure{\tr_0}$\label{item:nonoverlap-left-right-measure}
\end{enumerate}
\end{lemma}

\begin{proof}
By definition, we have there exist a term context $C[\_]$, two positions $q_x,q'$ and a variable $x$ such that $u = C[r\sigma]$, $C[\_]_{|p} = \_$, $t = C[\ell\sigma]$, $q = q_x \cdot q'$,  $\ell_{|q_x} = x \in \X$, $x\sigma_{|q'} =_\En \ell'\sigma'$ and $C[\ell\sigma[r'\sigma']_q] = v$.
Since $\ell'\sigma' =_\En x\sigma_{|q'}$, there exists a syntactic rewrite trace $\tr_L$ from $\En$ such that $x\sigma_{|q'} \Rightrwstep{\tr_L} \ell'\sigma'$. Hence, we can define $\tr_M = \tr_L \rightrwlabel{\varepsilon}{\sigma'}{\ell'}{r'}$ to obtain $x\sigma_{|q'} \Rightrwstep{\tr_M} r'\sigma' = v_{|p'}$. From \Cref{lem:operation-properties}, we deduce that $v_{|p'} \Leftrwstep{\reverse{\tr_M}} x\sigma_{|q'}$.

Let $\theta$ be the substitution such that $y\theta = y\sigma$ for all $y \in \vars{r}\setminus \{x\}$ and $x\theta = x\sigma[r'\sigma']_{q'}$.
If $p_1,\ldots,p_n$ are the positions of $x$ in $r$, thus obtain that $C[r\sigma] \Rightrwstep{p \cdot p_1 \cdot q' \cdot \tr_M} \ldots  \Rightrwstep{p \cdot p_n \cdot q' \cdot \tr_M} C[r\theta]$. 
Let us denote $\omega_M = \rightrwlabel{p}{\theta}{r}{\ell}$ if $r\theta > \ell\theta$ and $\omega_M = \leftrwlabel{p}{\theta}{\ell}{r}$ otherwise. We thus have $C[r\theta] \Leftrightrwstep{\omega_M} C[\ell\theta]$.
Similarly, if $q_1,\ldots,q_m$ are the positions of $x$ in $\ell$ other than $q_x$, we obtain that $C[\ell\theta] \rwLeftStep{p\cdot q_1 \cdot q'\cdot \reverse{\tr_M}} \ldots \rwLeftStep{p\cdot q_m \cdot q'\cdot \reverse{\tr_M}} v$. Denoting $\tr'_0 = (p \cdot p_1 \cdot q' \cdot \tr_M) \ldots (p \cdot p_n \cdot q' \cdot \tr_M) \omega_M (p \cdot q_1 \cdot q' \cdot \reverse{\tr_M}) \ldots (p \cdot q_m \cdot q' \cdot \reverse{\tr_M})$, we obtain $u \Leftrightrwstep{\tr'_0} v$. As $\M$ is greater than $u \Leftrightrwstep{\tr_0} v$, we obtain by construction of $\tr_0'$ that $\M$ is greater than $u \Leftrightrwstep{\tr'_0} v$.
Additionally, by construction, $\tr'_0$ is syntactic from $\En \cup \R \cup \{ r \rightarrow \ell \}$. 

We now order the labels of $\tr'_0$ w.r.t. $\tr_0$. Let us denote $\omega = \rightrwlabel{p}{\sigma}{\ell}{r}$. As $q_x$ exists, if $m > 0$ then it implies that $q_1,\ldots,q_m$ are all different from $\varepsilon$. Thus, for all $i \in \{1,\ldots,m\}$, for all $\omega_1 \in p \cdot q_i \cdot q' \cdot \reverse{\tr_M}$ , $\measure{\omega_1} \lessRwLbl{\Pset(\M)} \measure{\omega}$. 
By definition, we know that $\ell'\sigma' > r'\sigma'$ and so $x\sigma > x\theta$, meaning that $r\sigma > r\theta$. Therefore, if $r\theta > \ell\theta$ then we directly obtain $\measure{\omega_M} = (p,\ell\theta,r\theta) \lessRwLbl{\Pset(\M)} (p,r\sigma,\ell\sigma) = \measure{\omega}$. Else we have $\measure{\omega_M} = (p,r\theta,\ell\theta)$, which also implies implies $\measure{\omega_M} \lessRwLbl{\Pset(\M)}  \measure{\omega}$ as $r\sigma > r\theta$. 

Finally, let us look $\omega_1 \in p \cdot p_i \cdot q' \cdot \tr_M$ for some $i \in \{1,\ldots,n\}$. If $p_i \neq \varepsilon$ or $q' \neq \varepsilon$ then not only we have $\measure{\omega_1} \lessRwLbl{\Pset(\M)}  \measure{\omega}$ but we also have that it holds for all $i \in \{1,\ldots,n\}$ since $p_i \neq \varepsilon$ implies for all $j \in \{1,\ldots,n\}$, $p_j \neq \varepsilon$. Thus, we now focus on $n=1$, $p_1 = \varepsilon$ and $q' = \varepsilon$: It implies that $r = x$ and $r\sigma =_\En \ell'\sigma'$. Thus, $r\sigma > r'\sigma'$ and so $\measure{\rightrwlabel{p}{\sigma'}{\ell'}{r'}} = (p,r'\sigma',\ell'\sigma') \lessRwLbl{\Pset(\M)} \measure{\omega}$. 
We therefore conclude that for all $\omega_1 \in \tr'_0$, there exist $\omega_0 \in \tr_0$ such that $\measure{\omega_1} \lessRwLbl{\Pset(\M)} \measure{\omega_0}$.

We can represent graphically the transformation as follows:
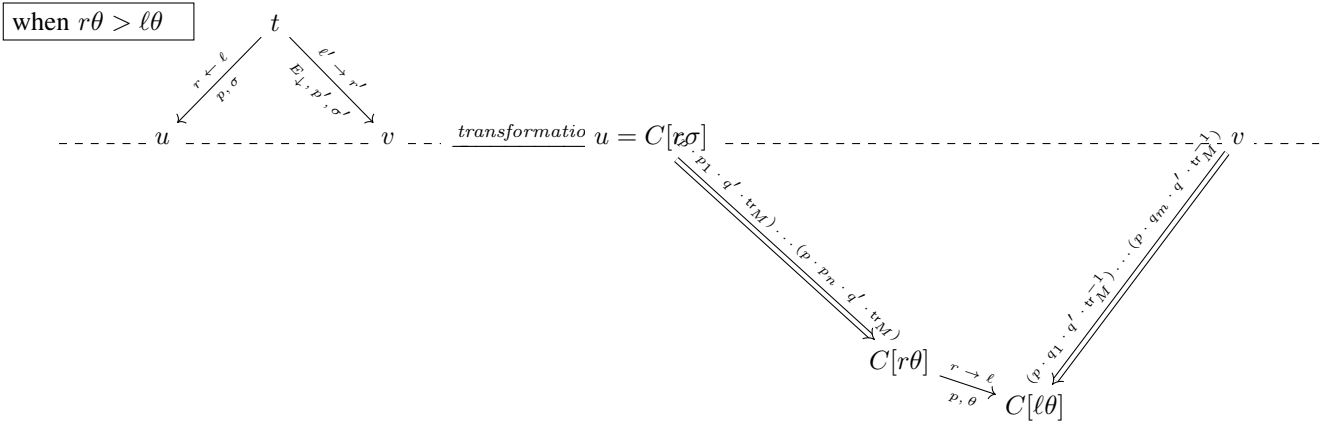
\begin{center}
    \begin{tikzpicture}[
        term/.style={}
        ]
        \def\length{1.5cm}
        \node[anchor=mid] (L) {};
        \node[right=17cm of L.mid,anchor=mid] (R) {};
        \draw[dashed] (L) -- (R);

        \node[draw,rectangle,above left= 1.5cm and 0.5cm of L,anchor=west,text width=2.3cm] (Label) {when $r\theta > \ell\theta$};

        \node[fill=white,right=5.2cm of L.mid] (Label) {$\xrightarrow{\mathit{transformation}}$};

        \node[fill=white,term,right=\length of L.mid,anchor=mid] (T0) {$u$};
        \node[fill=white,term,above right=\length and \length of T0.mid,anchor=mid] (T1) {$t$};
        \node[fill=white,term,below right=\length and \length of T1.mid,anchor=mid] (T2) {$v$};

        \draw[<-] (T0) edge node[auto,sloped] {\tiny$r \leftarrow \ell$} node[auto,sloped,below] {\tiny$p, \sigma$} (T1);
        \draw[->] (T1) edge node[auto,sloped] {\tiny$\ell' \rightarrow r'$} node[auto,sloped,below] {\tiny$\En, p', \sigma'$} (T2);

        \node[term,fill=white,right=6.5cm of T0.mid,anchor=mid] (T0') {$u = C[r\sigma]$};
        \node[term,fill=white,below right= 2*\length and 2.2*\length of T0'.mid,anchor=mid] (T1') {$C[r\theta]$};
        \node[term,fill=white,below right= 0.4*\length and 1.2*\length of T1'.mid,anchor=mid] (T2') {$C[\ell\theta]$};
        \node[term,fill=white,right= 5.2*\length of T0'.mid,anchor=mid] (T3') {$v$};

        \draw (T0') edge[double equal sign distance,-Implies] node[auto,sloped] {\tiny$(p \cdot p_1 \cdot q' \cdot \tr_M) \ldots  (p \cdot p_n \cdot q' \cdot \tr_M)$} (T1');
        \draw[->] (T1') edge node[auto,sloped] {\tiny$r \rightarrow \ell$} node[auto,sloped,below] {\tiny$p, \theta$} (T2');
        \draw (T2') edge[double equal sign distance,Implies-] node[auto,sloped] {\tiny$(p \cdot q_1 \cdot q' \cdot \tr^{-1}_M) \ldots  (p \cdot q_m \cdot q' \cdot \tr^{-1}_M)$} (T3');
        
    \end{tikzpicture}
\end{center}
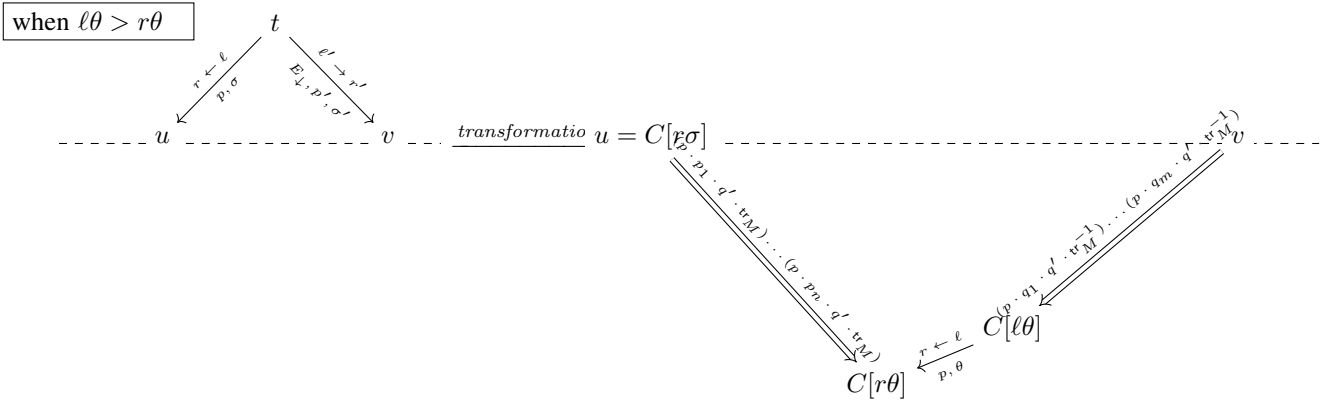
\begin{center}
    \begin{tikzpicture}[
        term/.style={}
        ]
        \def\length{1.5cm}
        \node[anchor=mid] (L) {};
        \node[right=17cm of L.mid,anchor=mid] (R) {};
        \draw[dashed] (L) -- (R);

        \node[draw,rectangle,above left= 1.5cm and 0.5cm of L,anchor=west,text width=2.3cm] (Label) {when $\ell\theta > r\theta$};

        \node[fill=white,right=5.2cm of L.mid] (Label) {$\xrightarrow{\mathit{transformation}}$};

        \node[fill=white,term,right=\length of L.mid,anchor=mid] (T0) {$u$};
        \node[fill=white,term,above right=\length and \length of T0.mid,anchor=mid] (T1) {$t$};
        \node[fill=white,term,below right=\length and \length of T1.mid,anchor=mid] (T2) {$v$};

        \draw[<-] (T0) edge node[auto,sloped] {\tiny$r \leftarrow \ell$} node[auto,sloped,below] {\tiny$p, \sigma$} (T1);
        \draw[->] (T1) edge node[auto,sloped] {\tiny$\ell' \rightarrow r'$} node[auto,sloped,below] {\tiny$\En, p', \sigma'$} (T2);

        \node[term,fill=white,right=6.5cm of T0.mid,anchor=mid] (T0') {$u = C[r\sigma]$};
        \node[term,fill=white,below right= 2.2*\length and 2*\length of T0'.mid,anchor=mid] (T1') {$C[r\theta]$};
        \node[term,fill=white,above right= 0.5*\length and 1.2*\length of T1'.mid,anchor=mid] (T2') {$C[\ell\theta]$};
        \node[term,fill=white,right= 5.2*\length of T0'.mid,anchor=mid] (T3') {$v$};

        \draw (T0') edge[double equal sign distance,-Implies] node[auto,sloped] {\tiny$(p \cdot p_1 \cdot q' \cdot \tr_M) \ldots  (p \cdot p_n \cdot q' \cdot \tr_M)$} (T1');
        \draw[<-] (T1') edge node[auto,sloped] {\tiny$r \leftarrow \ell$} node[auto,sloped,below] {\tiny$p, \theta$} (T2');
        \draw (T2') edge[double equal sign distance,Implies-] node[auto,sloped] {\tiny$(p \cdot q_1 \cdot q' \cdot \tr^{-1}_M) \ldots  (p \cdot q_m \cdot q' \cdot \tr^{-1}_M)$} (T3');
        
    \end{tikzpicture}
\end{center}
We already proved that, denoting $\R' = (\R \cup \Rn \cup \{ r \rightarrow \ell \})$, we have that $\tr'_0$ is a syntactic rewrite trace from $\R' \cup \En$. Since $\tr_0$ is from $\R$, it is naturally also from $\R' \cup \En$. Note that by \Cref{def:saturated}, we know that $\R' \normstep[\Rn,\En]^* \R$ and $\Rn \subseteq \R'$. Thus, we can apply \Cref{cor:normalisation} to deduce that there exists a syntactic rewrite trace $\tr_1$ from $\R \cup \En$ such that $u \Leftrightrwstep{\tr_1} v$ which $\M$ rewrites into, and $\measure{\tr_1} \lessRwLbl{\Pset(\M)} \measure{\tr_0}$. This conclude the proof.
\end{proof}


\begin{lemma}
\label{lem:nonoverlap-En-right}
Assume that $\En$ has finite equivalence classes. Let $\R,\Rn$ be sets of rewrite rules such that $\R$ is saturated w.r.t. $\Rn$. Let $u \rightrwstep{p}{\sigma}{\ell}{r} t \rightrwstep[\En]{p'}{\sigma'}{\ell'}{r'} v$ be some ground rewrite steps. Let $\tr_0 = \rightrwlabel{p}{\sigma}{\ell}{r} \rightrwlabel{p'}{\sigma'}{\ell'}{r'}[\En]$ with $(\ell \rightarrow r) \in \En$ and $(\ell'\rightarrow r') \in \R$. Let $\M$ be a set of terms that is greater than $u \Leftrightrwstep{\tr_0} v$.

If $p' = p \cdot q$, and $q \not\in \PosX{r}$, and $\ell'\sigma' > r'\sigma'$ then there exists a ground syntactic rewrite trace $\tr_1$ from $\En \cup \R$ such that 
\begin{enumerate}[label=\textbf{P\arabic*}]
\item $u \rwLeftrightStep{\tr_1} v$, and \label{item:nonoverlap-En-right-Leftright}
\item $\M$ is greater than $u \rwLeftrightStep{\tr_1} v$, and \label{item:nonoverlap-En-right-derive}
\item $\measure{\tr_1} \lessRwLbl{\Pset(\M)} \measure{\tr_0}$\label{item:nonoverlap-En-right-measure}
\end{enumerate}
\end{lemma}

\begin{proof}
By definition, we have there exist a term context $C[\_]$, two positions $q_x,q'$ and a variable $x$ such that $u = C[\ell\sigma]$, $C[\_]_{|p} = \_$, $t = C[r\sigma]$, $q = q_x \cdot q'$,  $\ell_{|q_x} = x \in \X$, $x\sigma_{|q'} =_\En \ell'\sigma'$ and $C[r\sigma[r'\sigma']_q] = v$.
Since $\ell'\sigma' =_\En x\sigma_{|q'}$, there exists a syntactic rewrite trace $\tr_L$ from $\En$ such that $x\sigma_{|q'} \Rightrwstep{\tr_L} \ell'\sigma'$. Hence, we can define $\tr_M = \tr_L \rightrwlabel{\varepsilon}{\sigma'}{\ell'}{r'}$ to obtain $x\sigma_{|q'} \Rightrwstep{\tr_M} r'\sigma' = v_{|p'}$. From \Cref{lem:operation-properties}, we deduce that $v_{|p'} \Leftrwstep{\reverse{\tr_M}} x\sigma_{|q'}$.

Let $\theta$ be the substitution such that $y\theta = y\sigma$ for all $y \in \vars{r}\setminus \{x\}$ and $x\theta = x\sigma[r'\sigma']_{q'}$.
If $p_1,\ldots,p_n$ are the positions of $x$ in $\ell$, thus obtain that $C[\ell\sigma] \Rightrwstep{p \cdot p_1 \cdot q' \cdot \tr_M} \ldots  \Rightrwstep{p \cdot p_n \cdot q' \cdot \tr_M} C[\ell\theta]$. As $(\ell \rightarrow r) \in \En$, we deduce that $r\theta =_\En \ell\theta$. 
Denoting $\omega_M = \rightrwlabel{p}{\theta}{\ell}{r}$, we thus have $C[\ell\theta] \Rightrwstep{\omega_M} C[r\theta]$.

Similarly, if $q_1,\ldots,q_m$ are the positions of $x$ in $r$ other than $q_x$, we obtain that $C[r\theta] \rwLeftStep{p\cdot q_1 \cdot q'\cdot \reverse{\tr_M}} \ldots \rwLeftStep{p\cdot q_m \cdot q'\cdot \reverse{\tr_M}} v$. Denoting $\tr'_1 = (p \cdot p_1 \cdot q' \cdot \tr_M) \ldots (p \cdot p_n \cdot q' \cdot \tr_M) \omega_M (p \cdot q_1 \cdot q' \cdot \reverse{\tr_M}) \ldots (p \cdot q_m \cdot q' \cdot \reverse{\tr_M})$, we obtain $u \Leftrightrwstep{\tr_1} v$. As $\M$ is greater than $u \Leftrightrwstep{\tr_0} v$, we obtain by construction of $\tr_1$ that $\M$ is greater than $u \Leftrightrwstep{\tr_1} v$.
Additionally, by construction, $\tr_1$ is syntactic from $\En \cup \R$. 

We now order the labels of $\tr_1$ w.r.t. $\tr_0$. Let us denote $\omega = \rightrwlabel{p}{\sigma}{\ell}{r}$. Since $(\ell \rightarrow r) \in \En$ and $\En$ has finite equivalence classes, we know that $\ell$ and $r$ are not variables. Hence $p_1,\ldots,p_n,q_1,\ldots,q_m$ are not $\varepsilon$. Therefore, for all $i \in \{1,\ldots,m\}$, for all $\omega_1 \in p \cdot q_i \cdot q' \cdot \reverse{\tr_M}$ , $\measure{\omega_1} \lessRwLbl{\Pset(\M)} \measure{\omega}$, and for all $i \in \{1,\ldots,n\}$, for all $\omega_1 \in p \cdot p_i \cdot q' \cdot \tr_M$ , $\measure{\omega_1} \lessRwLbl{\Pset(\M)} \measure{\omega}$. Recall that $\ell'\sigma' > r'\sigma'$, meaning that $x\sigma > x\theta$ and so $\ell\sigma =_\En r\sigma > r\theta =_\En \ell\theta$. We conclude that $\measure{\omega_M} \lessRwLbl{\Pset(\M)} \measure{\omega}$ and so, combining all previous statements, we obtain $\measure{\tr_1} \lessRwLbl{\Pset(\M)} \measure{\tr_0}$.

We can represent graphically the transformation as follows:
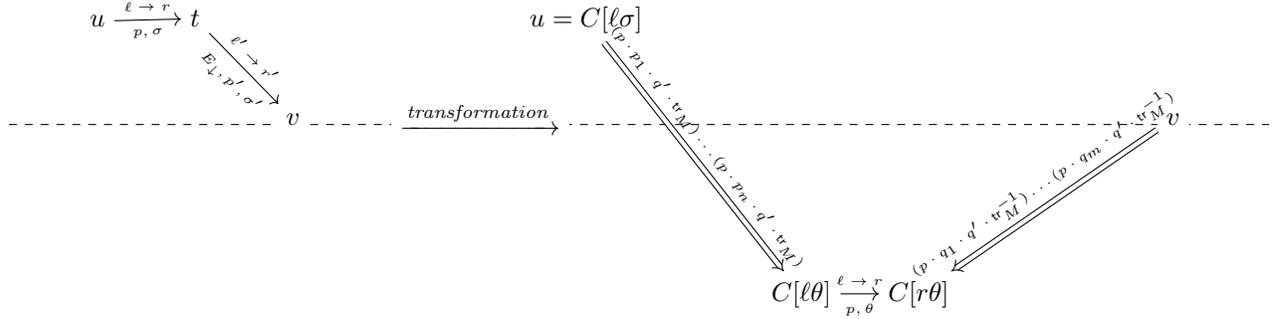
\begin{center}
    \begin{tikzpicture}[
        term/.style={}
        ]
        \def\length{1.3cm}
        \node[anchor=mid] (L) {};
        \node[right=17cm of L.mid,anchor=mid] (R) {};
        \draw[dashed] (L) -- (R);

        \node[fill=white,right=5.2cm of L.mid] (Label) {$\xrightarrow{\mathit{transformation}}$};

        \node[fill=white,term,above right=\length and \length of L.mid,anchor=mid] (T0) {$u$};
        \node[fill=white,term,right=\length of T0.mid,anchor=mid] (T1) {$t$};
        \node[fill=white,term,below right=\length and \length of T1.mid,anchor=mid] (T2) {$v$};

        \draw[->] (T0) edge node[auto,sloped] {\tiny$\ell \rightarrow r$} node[auto,sloped,below] {\tiny$p, \sigma$} (T1);
        \draw[->] (T1) edge node[auto,sloped] {\tiny$\ell' \rightarrow r'$} node[auto,sloped,below] {\tiny$\En, p', \sigma'$} (T2);

        \node[term,fill=white,right=6.5cm of T0.mid,anchor=mid] (T0') {$u = C[\ell\sigma]$};
        \node[term,fill=white,below right= 2.8*\length and 2.2*\length of T0'.mid,anchor=mid] (T1') {$C[\ell\theta]$};
        \node[term,fill=white,right= 1.2*\length of T1'.mid,anchor=mid] (T2') {$C[r\theta]$};
        \node[term,fill=white,below right= \length and 6*\length of T0'.mid,anchor=mid] (T3') {$v$};

        \draw (T0') edge[double equal sign distance,-Implies] node[auto,sloped] {\tiny$(p \cdot p_1 \cdot q' \cdot \tr_M) \ldots  (p \cdot p_n \cdot q' \cdot \tr_M)$} (T1');
        \draw[->] (T1') edge node[auto,sloped] {\tiny$\ell \rightarrow r$} node[auto,sloped,below] {\tiny$p, \theta$} (T2');
        \draw (T2') edge[double equal sign distance,Implies-] node[auto,sloped] {\tiny$(p \cdot q_1 \cdot q' \cdot \tr^{-1}_M) \ldots  (p \cdot q_m \cdot q' \cdot \tr^{-1}_M)$} (T3'); 
    \end{tikzpicture}
\end{center}
\end{proof}


\begin{lemma}
\label{lem:nonoverlap-right-right}
Let $\R,\Rn$ be sets of rewrite rules such that $\R$ is saturated w.r.t. $\Rn$. Let $u \rightrwstep{p}{\sigma}{\ell}{r} t \rightrwstep[\En]{p'}{\sigma'}{\ell'}{r'} v$ be some ground rewrite steps. Let $\tr_0 = \rightrwlabel{p}{\sigma}{\ell}{r} \rightrwlabel{p'}{\sigma'}{\ell'}{r'}[\En]$ with $(\ell \rightarrow r), (\ell'\rightarrow r') \in \R$. Let $\M$ be a set of terms that is greater than $u \Leftrightrwstep{\tr_0} v$.

If $p' = p \cdot q$, and $q \not\in \PosX{r}$, and $\ell'\sigma' > r'\sigma'$, and $\ell\sigma > r\sigma$ then there exists a ground syntactic rewrite trace $\tr_1$ from $\En \cup \R$ such that 
\begin{enumerate}[label=\textbf{P\arabic*}]
\item $u \rwLeftrightStep{\tr_1} v$, and \label{item:nonoverlap-right-right-Leftright}
\item $\M$ is greater than $u \rwLeftrightStep{\tr_1} v$, and \label{item:nonoverlap-right-right-derive}
\item $\measure{\tr_1} \lessRwLbl{\Pset(\M)} \measure{\tr_0}$\label{item:nonoverlap-right-right-measure}
\end{enumerate}
\end{lemma}

\begin{proof}
By definition, we have there exist a term context $C[\_]$, two positions $q_x,q'$ and a variable $x$ such that $u = C[\ell\sigma]$, $C[\_]_{|p} = \_$, $t = C[r\sigma]$, $q = q_x \cdot q'$,  $\ell_{|q_x} = x \in \X$, $x\sigma_{|q'} =_\En \ell'\sigma'$ and $C[r\sigma[r'\sigma']_q] = v$.
Since $\ell'\sigma' =_\En x\sigma_{|q'}$, there exists a syntactic rewrite trace $\tr_L$ from $\En$ such that $x\sigma_{|q'} \Rightrwstep{\tr_L} \ell'\sigma'$. Hence, we can define $\tr_M = \tr_L \rightrwlabel{\varepsilon}{\sigma'}{\ell'}{r'}$ to obtain $x\sigma_{|q'} \Rightrwstep{\tr_M} r'\sigma' = v_{|p'}$. From \Cref{lem:operation-properties}, we deduce that $v_{|p'} \Leftrwstep{\reverse{\tr_M}} x\sigma_{|q'}$.

Let $\theta$ be the substitution such that $y\theta = y\sigma$ for all $y \in \vars{r}\setminus \{x\}$ and $x\theta = x\sigma[r'\sigma']_{q'}$.
If $p_1,\ldots,p_n$ are the positions of $x$ in $\ell$, thus obtain that $C[\ell\sigma] \Rightrwstep{p \cdot p_1 \cdot q' \cdot \tr_M} \ldots  \Rightrwstep{p \cdot p_n \cdot q' \cdot \tr_M} C[\ell\theta]$. 

Let us denote $\omega_M = \leftrwlabel{p}{\theta}{r}{\ell}$ if $r\theta > \ell\theta$ and $\omega_M = \rightrwlabel{p}{\theta}{\ell}{r}$ otherwise. We thus have $C[\ell\theta] \Leftrightrwstep{\omega_M} C[r\theta]$.
Similarly, if $q_1,\ldots,q_m$ are the positions of $x$ in $r$ other than $q_x$, we obtain that $C[r\theta] \rwLeftStep{p\cdot q_1 \cdot q'\cdot \reverse{\tr_M}} \ldots \rwLeftStep{p\cdot q_m \cdot q'\cdot \reverse{\tr_M}} v$. Denoting $\tr'_0 = (p \cdot p_1 \cdot q' \cdot \tr_M) \ldots (p \cdot p_n \cdot q' \cdot \tr_M) \omega_M (p \cdot q_1 \cdot q' \cdot \reverse{\tr_M}) \ldots (p \cdot q_m \cdot q' \cdot \reverse{\tr_M})$, we obtain $u \Leftrightrwstep{\tr'_0} v$. As $\M$ is greater than $u \Leftrightrwstep{\tr_0} v$, we obtain by construction of $\tr_0'$ that $\M$ is greater than $u \Leftrightrwstep{\tr'_0} v$.
Additionally, by construction, $\tr'_0$ is syntactic from $\En \cup \R \cup \{ r \rightarrow \ell \}$. 

We now order the labels of $\tr'_0$ w.r.t. $\tr_0$. Let us denote $\omega = \rightrwlabel{p}{\sigma}{\ell}{r}$. As $q_x$ exists, if $m > 0$ then it implies that $q_1,\ldots,q_m$ are all different from $\varepsilon$. Thus, for all $i \in \{1,\ldots,m\}$, for all $\omega_1 \in p \cdot q_i \cdot q' \cdot \reverse{\tr_M}$ , $\measure{\omega_1} \lessRwLbl{\Pset(\M)} \measure{\omega}$. 
By definition, we know that $\ell'\sigma' > r'\sigma'$ and so $x\sigma > x\theta$, meaning that $r\sigma > r\theta$. Therefore, if $r\theta > \ell\theta$ then we directly obtain $\measure{\omega_M} = (p,\ell\theta,r\theta) \lessRwLbl{\Pset(\M)} (p,r\sigma,\ell\sigma) = \measure{\omega}$. Else we have $\measure{\omega_M} = (p,r\theta,\ell\theta)$, which also implies implies $\measure{\omega_M} \lessRwLbl{\Pset(\M)}  \measure{\omega}$ as $r\sigma > r\theta$. 

Finally, let us look $\omega_1 \in p \cdot p_i \cdot q' \cdot \tr_M$ for some $i \in \{1,\ldots,n\}$. If $p_i \neq \varepsilon$ or $q' \neq \varepsilon$ then not only we have $\measure{\omega_1} \lessRwLbl{\Pset(\M)}  \measure{\omega}$ but we also have that it holds for all $i \in \{1,\ldots,n\}$ since $p_i \neq \varepsilon$ implies for all $j \in \{1,\ldots,n\}$, $p_j \neq \varepsilon$. Thus, we now focus on $n=1$, $p_1 = \varepsilon$ and $q' = \varepsilon$: It implies that $\ell = x$ and $\ell\sigma =_\En \ell'\sigma'$. However, $\ell$ is a subterm of $r$, meaning that $\ell\sigma$ is a subterm of $r\sigma$ and so $\ell\sigma \not> r\sigma$ (as $>$ is a reduction ordering) which is a contradiction with $\ell\sigma > r\sigma$. Hence, we cannot have $p_1 = \varepsilon = q'$.
We therefore conclude that for all $\omega_1 \in \tr'_0$, there exist $\omega_0 \in \tr_0$ such that $\measure{\omega_1} \lessRwLbl{\Pset(\M)} \measure{\omega_0}$.

We can represent graphically the transformation as follows:
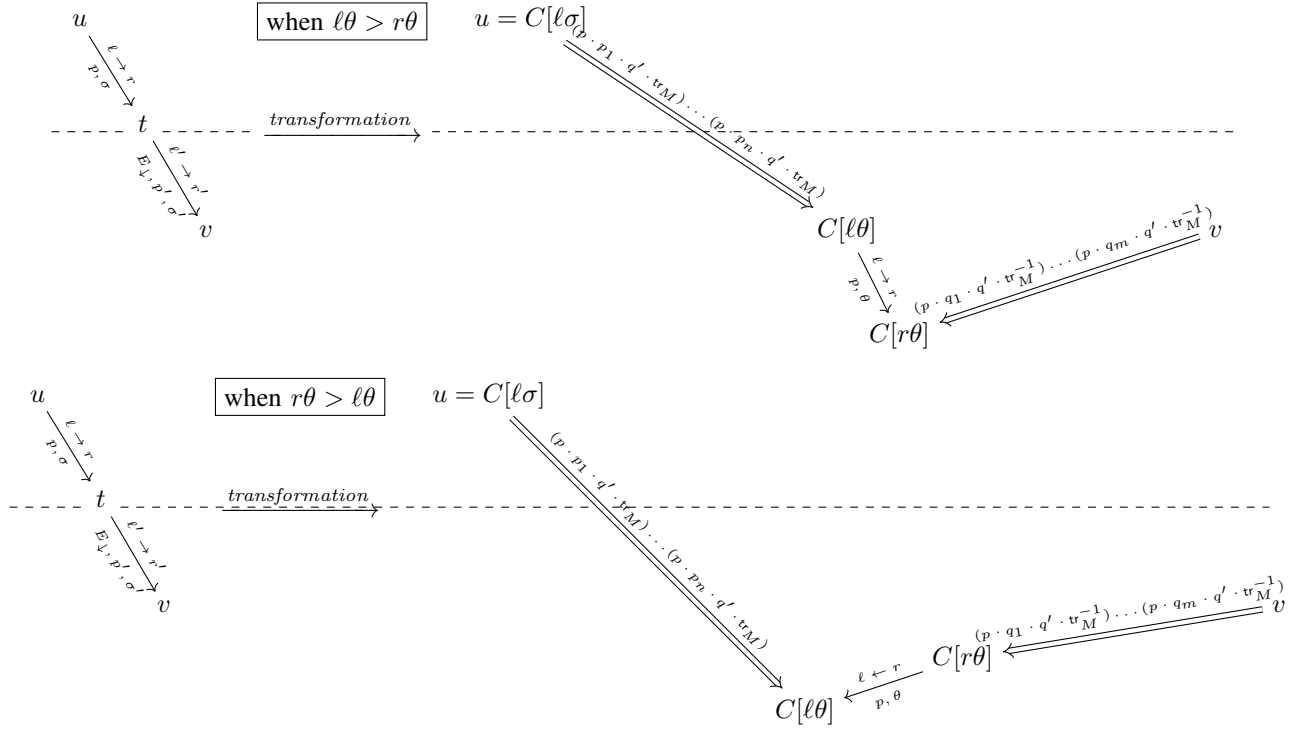
\begin{center}
    \begin{tikzpicture}[
        term/.style={}
        ]
        \def\length{1.4cm}
        \node[anchor=mid] (L) {};
        \node[right=16cm of L.mid,anchor=mid] (R) {};
        \draw[dashed] (L) -- (R);

        \node[fill=white,right=2.8cm of L.mid] (Transformation) {$\xrightarrow{\mathit{transformation}}$};
        \node[draw,rectangle,above= 0.8cm of Transformation] (Label) {when $\ell\theta > r\theta$};

        \node[fill=white,term,above right=\length and 0.5cm of L.mid,anchor=mid] (T0) {$u$};
        \node[fill=white,term,below right=\length and 0.6*\length of T0.mid,anchor=mid] (T1) {$t$};
        \node[fill=white,term,below right=\length and 0.6*\length of T1.mid,anchor=mid] (T2) {$v$};

        \draw[->] (T0) edge node[auto,sloped] {\tiny$\ell \rightarrow r$} node[auto,sloped,below] {\tiny$p, \sigma$} (T1);
        \draw[->] (T1) edge node[auto,sloped] {\tiny$\ell' \rightarrow r'$} node[auto,sloped,below] {\tiny$\En, p', \sigma'$} (T2);

        \node[term,fill=white,right=6cm of T0.mid,anchor=mid] (T0') {$u = C[\ell\sigma]$};
        \node[term,fill=white,below right= 2*\length and 3*\length of T0'.mid,anchor=mid] (T1') {$C[\ell\theta]$};
        \node[term,fill=white,below right= \length and 0.5*\length of T1'.mid,anchor=mid] (T2') {$C[r\theta]$};
        \node[term,fill=white,above right= \length and 3*\length of T2'.mid,anchor=mid] (T3') {$v$};

        \draw (T0') edge[double equal sign distance,-Implies] node[auto,sloped] {\tiny$(p \cdot p_1 \cdot q' \cdot \tr_M) \ldots  (p \cdot p_n \cdot q' \cdot \tr_M)$} (T1');
        \draw[->] (T1') edge node[auto,sloped] {\tiny$\ell \rightarrow r$} node[auto,sloped,below] {\tiny$p, \theta$} (T2');
        \draw (T2') edge[double equal sign distance,Implies-] node[auto,sloped] {\tiny$(p \cdot q_1 \cdot q' \cdot \tr^{-1}_M) \ldots  (p \cdot q_m \cdot q' \cdot \tr^{-1}_M)$} (T3');
        
    \end{tikzpicture}
\end{center}
\begin{center}
    \begin{tikzpicture}[
        term/.style={fill=white}
        ]
        \def\length{1.4cm}
        \node[anchor=mid] (L) {};
        \node[right=17cm of L.mid,anchor=mid] (R) {};
        \draw[dashed] (L) -- (R);

        \node[right=2.8cm of L.mid] (Transformation) {$\xrightarrow{\mathit{transformation}}$};
        \node[draw,rectangle,above= 0.8cm of Transformation] (Label) {when $r\theta > \ell\theta$};

        \node[term,above right=\length and 0.5cm of L.mid,anchor=mid] (T0) {$u$};
        \node[term,below right=\length and 0.6*\length of T0.mid,anchor=mid] (T1) {$t$};
        \node[term,below right=\length and 0.6*\length of T1.mid,anchor=mid] (T2) {$v$};

        \draw[->] (T0) edge node[auto,sloped] {\tiny$\ell \rightarrow r$} node[auto,sloped,below] {\tiny$p, \sigma$} (T1);
        \draw[->] (T1) edge node[auto,sloped] {\tiny$\ell' \rightarrow r'$} node[auto,sloped,below] {\tiny$\En, p', \sigma'$} (T2);

        \node[term,right=6cm of T0.mid,anchor=mid] (T0') {$u = C[\ell\sigma]$};
        \node[term,below right= 3*\length and 3*\length of T0'.mid,anchor=mid] (T1') {$C[\ell\theta]$};
        \node[term,above right= 0.5*\length and 1.5*\length of T1'.mid,anchor=mid] (T2') {$C[r\theta]$};
        \node[term,above right= 0.5*\length and 3*\length of T2'.mid,anchor=mid] (T3') {$v$};

        \draw (T0') edge[double equal sign distance,-Implies] node[auto,sloped] {\tiny$(p \cdot p_1 \cdot q' \cdot \tr_M) \ldots  (p \cdot p_n \cdot q' \cdot \tr_M)$} (T1');
        \draw[<-] (T1') edge node[auto,sloped] {\tiny$\ell \leftarrow r$} node[auto,sloped,below] {\tiny$p, \theta$} (T2');
        \draw (T2') edge[double equal sign distance,Implies-] node[auto,sloped] {\tiny$(p \cdot q_1 \cdot q' \cdot \tr^{-1}_M) \ldots  (p \cdot q_m \cdot q' \cdot \tr^{-1}_M)$} (T3');
        
    \end{tikzpicture}
\end{center}
We already proved that, denoting $\R' = (\R \cup \Rn \cup \{ r \rightarrow \ell \})$, we have that $\tr'_0$ is a syntactic rewrite trace from $\R' \cup \En$. Since $\tr_0$ is from $\R$, it is naturally also from $\R' \cup \En$. Note that by \Cref{def:saturated}, we know that $\R' \normstep[\Rn,\En]^* \R$ and $\Rn \subseteq \R'$. Thus, we can apply \Cref{cor:normalisation} to deduce that there exists a syntactic rewrite trace $\tr_1$ from $\R \cup \En$ such that $u \Leftrightrwstep{\tr_1} v$ which is smaller than $\M$, and $\measure{\tr_1} \lessRwLbl{\Pset(\M)} \measure{\tr_0}$. This conclude the proof.
\end{proof}


\subsubsection{Combining the different transformations}

We can regroup the different result from the two previous sections in the following lemma:


\begin{lemma}
\label{lem:combined-X-right}
Assume that $\En$ has finite equivalence classes. Let $\R,\Rn$ be sets of rewrite rules such that $\R$ is saturated w.r.t. $\Rn$. Assume that $\R^=$ is not trivial. Let $u \Leftrightrwstep{\omega} t \rightrwstep[\En]{p'}{\sigma'}{\ell'}{r'} v$ be some ground rewrite steps such that $\omega = \leftrightrwlabel{p}{\sigma}{\ell}{r}$ and ${\sim} \in \{ \rightarrow, \leftarrow \}$. Let $\tr_0 = \omega \rightrwlabel{p'}{\sigma'}{\ell'}{r'}[\En]$ with $(\ell \rightarrow r) \in \En \cup \R$ and $(\ell'\rightarrow r') \in \R$. Let $\M$ be a set of terms that is greater than $u \Leftrightrwstep{\tr_0} v$.

If $p \leq p'$ and $\ell'\sigma' > r'\sigma'$ then there exists a ground syntactic rewrite trace $\tr_1$ from $\En \cup \R$ such that 
\begin{enumerate}[label=\textbf{P\arabic*}]
\item $u \rwLeftrightStep{\tr_1} v$, and 
\item $\M$ is greater than $u \rwLeftrightStep{\tr_1} v$, and
\item $\measure{\tr_1} \lessRwLbl{\Pset(\M)} \measure{\tr_0}$
\end{enumerate}
\end{lemma}

\begin{proof}
As $p \leq p'$, there exists $q$ such that $p' = p \cdot q$. We do a case analysis on the shape of $\omega$ and on the value of $q$.
\begin{itemize}
    \item \emph{Case $q \in \PosX{r}$:}
        \begin{itemize}
        \item \emph{${\sim} = {\rightarrow}$ and $(\ell \rightarrow r) \in \R$:} In that case, $\ell\sigma > r\sigma$. Note that $\R^=$ is not trivial. If $\ell \in \X$ and $\ell \not\in \vars{r}$ then we would have that for all terms $t$, $t =_{\R^=} r$ which is a contradiction of $\R^=$ not being trivial. We conclude by applying \Cref{lem:overlap-right-right}.
        \item \emph{${\sim} = {\rightarrow}$ and $(\ell \rightarrow r) \in \En$:} We conclude by applying \Cref{lem:overlap-En-right}.
        \item \emph{${\sim} = {\leftarrow}$ and $(r \rightarrow \ell) \in \R$:} In that case, $\ell\sigma > r\sigma$. We conclude by applying \Cref{lem:overlap-left-right}.
        \end{itemize}
    \item \emph{Case $q \not\in \PosX{r}$:}
        \begin{itemize}
        \item \emph{${\sim} = {\rightarrow}$ and $(\ell \rightarrow r) \in \R$:} In that case, $\ell\sigma > r\sigma$. We conclude by applying \Cref{lem:nonoverlap-right-right}.
        \item \emph{${\sim} = {\rightarrow}$ and $(\ell \rightarrow r) \in \En$:} We conclude by applying \Cref{lem:nonoverlap-En-right}.
        \item \emph{${\sim} = {\leftarrow}$ and $(r \rightarrow \ell) \in \R$:} In that case, $\ell\sigma > r\sigma$. We conclude by applying \Cref{lem:nonoverlap-left-right}.\qedhere
        \end{itemize}
\end{itemize}
\end{proof}


\begin{lemma}
\label{lem:parallel}
Let $\omega = \leftrightrwlabel{p}{\sigma}{\ell}{r}[E]$ with ${\sim} \in \{ \rightarrow, \leftarrow\}$. Let $u \Leftrightrwstep{\omega} t \rightrwstep[E']{p'}{\sigma'}{\ell'}{r'} v$ be some rewrite steps. Let $\tr_0 = \omega \rightrwlabel{p'}{\sigma'}{\ell'}{r'}[E]$. Let $\M$ a set of terms that is greater than $u \Leftrightrwstep{\tr_0} v$.

If $p \para p'$ then by defining the ground rewrite trace $\tr_1 = \rightrwlabel{p'}{\sigma'}{\ell'}{r'}[E] \omega$ we have
\begin{enumerate}[label=\textbf{P\arabic*}]
\item $u \rwLeftrightStep{\tr_1} v$, and \label{item:parallel-Leftright}
\item $\M$ is greater than $u \rwLeftrightStep{\tr_1} v$, and \label{item:parallel-derive}
\item $\measure{\tr_1} \eqRwLbl{\Pset(\M)} \measure{\tr_0}$\label{item:parallel-measure}
\item if ${\sim} = {\leftarrow}$ and either $r\sigma > \ell\sigma$ or $\ell'\sigma' > r'\sigma'$ then $\measureTerm{u,\tr_0,v} < \measureTerm{u,\tr_1,v}$\label{item:parallel-measureTerm}
\end{enumerate}
\end{lemma}

\begin{proof}
Let us define $E_1 = E$, $E_2 = \emptyset$ when ${\sim} = {\rightarrow}$ and $E_1 = \emptyset$, $E_2 = E$ otherwise. Since $p \para p'$, there exist a term context $C[\__1,\__2]$ and four terms $u',v',t_1,t_2$ such that $C[\__1,\__2]_{|p} = \__1$ and $C[\__1,\__2]_{|p'} = \__2$, $t = C[t_1,t_2]$, $r\sigma =_{E_2} t_1$, $\ell'\sigma' =_{E'} t_2$, $u' =_{E_1} \ell\sigma$, $v' = r'\sigma'$, $u = C[u',t_2]$ and $v = C[t_1,v']$. Defining $t' = C[u',v']$, we deduce that $u \rightrwstep[E']{p'}{\sigma'}{\ell'}{r'} t'$ and $t' \Leftrightrwstep{\omega} v$, which gives us \Cref{item:parallel-Leftright}.
Notice that as $\M$ rewrites into $u \Leftrightrwstep{\tr_0} v$, we directly obtain by construction that $\M$ rewrites into $u \rwLeftrightStep{\tr_1} v$, i.e. \Cref{item:parallel-derive}.
Notice that $\measure{\tr_1} = \measure{\tr_0}$, hence $\measure{\tr_0} \eqRwLbl{\Pset(\M)} \measure{\tr_1}$, i.e. \Cref{item:parallel-measure}.
Finally, when ${\sim} = {\leftarrow}$, since $>$ is a $\En$-compatible reduction ordering, we know from $u \leftrwstep[E]{p}{\sigma}{r}{\ell} t$ that $t \geq u$, the inequality being strict when $r\sigma > \ell\sigma$. Moreover, as $u \rightrwstep[E']{p'}{\sigma'}{\ell'}{r'} t'$, we deduce that $u \geq t'$, the inequality being strict when $\ell'\sigma' > r'\sigma'$. Therefore, if $r\sigma > \ell\sigma$ or $\ell'\sigma' > r'\sigma'$ then we have $t > t'$ which concludes the proof of $\measureTerm{u,\tr_0,v} < \measureTerm{u,\tr_1,v}$ and thus \Cref{item:parallel-measureTerm}.

We can represent graphically the transformation as follows:
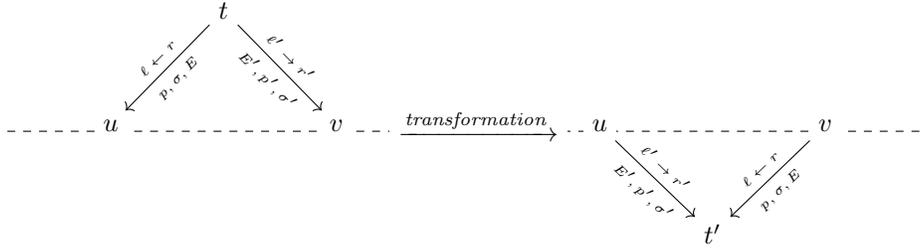
\begin{center}
    \begin{tikzpicture}[
        term/.style={}
        ]
        \def\length{1.5cm}
        \node[anchor=mid] (L) {};
        \node[right=12.5cm of L.mid,anchor=mid] (R) {};
        \draw[dashed] (L) -- (R);

        \node[fill=white,right=5.2cm of L.mid] (Label) {$\xrightarrow{\mathit{transformation}}$};

        \node[fill=white,term,right=\length of L.mid,anchor=mid] (T0) {$u$};
        \node[fill=white,term,above right=\length and \length of T0.mid,anchor=mid] (T1) {$t$};
        \node[fill=white,term,below right=\length and \length of T1.mid,anchor=mid] (T2) {$v$};

        \draw[<-] (T0) edge node[auto,sloped] {\tiny$\ell \leftarrow r$} node[auto,sloped,below] {\tiny$p, \sigma,E$} (T1);
        \draw[->] (T1) edge node[auto,sloped] {\tiny$\ell' \rightarrow r'$} node[auto,sloped,below] {\tiny$E', p', \sigma'$} (T2);

        \node[term,fill=white,right=6.5cm of T0.mid,anchor=mid] (T0') {$u$};
        \node[term,fill=white,below right= \length and \length of T0'.mid,anchor=mid] (T1') {$t'$};
        \node[term,fill=white,above right= \length and \length of T1'.mid,anchor=mid] (T2') {$v$};

        \draw[->] (T0') edge node[auto,sloped] {\tiny$\ell' \rightarrow r'$} node[auto,sloped,below] {\tiny$E',p', \sigma'$} (T1');
        \draw[<-] (T1') edge node[auto,sloped] {\tiny$\ell \leftarrow r$} node[auto,sloped,below] {\tiny$p, \sigma,E$} (T2');        
    \end{tikzpicture}
\end{center}
\end{proof}

\Cref{lem:parallel} allows us to consider an equivalence relation between rewrite traces, that allows to \emph{swap} parallel positions. 


\begin{definition}
Let $\RelPara$ the smallest equivalence relation on rewrite traces such that:
\begin{itemize}
\item $\rightrwlabel{p}{\sigma}{\ell}{r} \RelPara \leftrwlabel{p}{\sigma}{r}{\ell}$ when $(\ell \rightarrow r) \in \En$
\item $\leftrightrwlabel[\sim_1]{p_1}{\sigma_1}{\ell_1}{r_1}[E_1] \leftrightrwlabel[\sim_2]{p_2}{\sigma_2}{\ell_2}{r_2}[E_2] \RelPara  \leftrightrwlabel[\sim_2]{p_2}{\sigma_2}{\ell_2}{r_2}[E_2] \leftrightrwlabel[\sim_1]{p_1}{\sigma_1}{\ell_1}{r_1}[E_1]$ when $p_1 \para p_2$
\item $\tr \RelPara \tr'$ implies $\tr_1 \cdot \tr \cdot \tr_2 \RelPara \tr_1\cdot  \tr'\cdot  \tr_2$
\end{itemize}
\end{definition}


\begin{lemma}
\label{lem:parallel-equivalence-relation}
Let $u \Leftrightrwstep{\tr} v$. Let $\M$ a set of terms that is greater than  $u \Leftrightrwstep{\tr} v$. For all rewrite traces $\tr'$, if $\tr' \RelPara \tr$ then $u \Leftrightrwstep{\tr'} v$ which is smaller than $\M$, and $\measure{\tr} \eqRwLbl{\Pset(\M)} \measure{\tr'}$.
\end{lemma}

\begin{proof}
A direct application \Cref{lem:parallel} on the definition of the equivalence relation $\RelPara$.
\end{proof}


\subsubsection{Shape of minimal rewrite traces}


\begin{lemma}
\label{lem:shape-minimal}
Assume that $\En$ has finite equivalence classes. Let $\R,\Rn$ be sets of rewrite rules such that $\R$ is saturated w.r.t $\Rn$. Assume that $\R^=$ is not trivial. Let $u,v$ two ground terms and $\tr$ is a ground syntactic rewrite trace from $\En \cup \R$ such that $u \Leftrightrwstep{\tr} v$. Let $\M$ be a finite set of terms that is greater than $u \Leftrightrwstep{\tr} v$. 

If $u \Leftrightrwstep{\tr} v$ is minimal by $\lessRwLbl{\Pset(\M)}$ then there exist a term $t$ and three ground syntactic rewrite traces $\tr', \tr_L, \tr_R$ from $\En \cup \R$ such that:
\begin{itemize}
\item $\tr' \RelPara \tr$ and $\tr' = \tr_L \tr_R$ and $u \Rightrwstep{\tr_L} t \Leftrwstep{\tr_R} v$
\item for all $\tr'_L \RelPara \tr_L$, if $\tr'_L = \tr_1 \rightrwlabel{p}{\sigma}{\ell}{r} \tr_E \rightrwlabel{p'}{\sigma'}{\ell'}{r'} \tr_2$ and $(\ell' \rightarrow r') \in \R$ then $p \not\leq p'$
\item for all $\tr'_R \RelPara \tr_R$, if $\tr'_R = \tr_1 \leftrwlabel{p}{\sigma}{\ell}{r} \tr_E \leftrwlabel{p'}{\sigma'}{\ell'}{r'} \tr_2$ and $(\ell \rightarrow r) \in \R$ then $p' \not\leq p$
\end{itemize}
\end{lemma}

\begin{proof}
To prove the first point, we take the syntactic rewrite trace $\tr'$ such that $\tr \RelPara \tr'$ and $\tr'$ is minimal by $\measureTerm{\cdot}$, i.e. for all $\tr'' \RelPara \tr'$, $\measureTerm{u,\tr'',v} \not< \measureTerm{u,\tr',v}$. Let us look at the shape of $\tr'$. 

If $\tr' = \tr_1 \leftrwlabel{p}{\sigma}{\ell}{r} \rightrwlabel{p'}{\sigma'}{\ell'}{r'} \tr_2$ then by \Cref{lem:parallel}, we know that $\ell\sigma \leq r\sigma$ and $\ell'\sigma' \leq r'\sigma'$, i.e. $\ell\sigma =_\En r\sigma$ and $\ell'\sigma' =_\En r'\sigma'$ and so both $\ell \rightarrow r$ and $\ell' \rightarrow r'$ are in $\En$. Let us take the largest prefix $\tr''$ of $\tr'$ such that there exist two syntactic rewrite traces $\tr_L, \tr_R$ (possibly empty) and a two terms $t,s$ such that $u \Rightrwstep{\tr_L} t \Leftrwstep{\tr_R} s$ and $\tr'' \RelPara \tr_L\tr_R$ and $\measureTerm{u,\tr'',s} = \measureTerm{u,\tr_L\tr_R,s}$. 
If $\tr'' = \tr'$ then the result holds. Otherwise, $\tr' = \tr'' \omega \tr_W$ for some rewrite labels $\omega$ and a rewrite trace (possibly empty) $\tr_W$. By maximality of $\tr''$, we deduce that $\tr_R \neq \varepsilon$ and $\omega = \rightrwlabel{p}{\sigma}{\ell}{r}$.
As $\measureTerm{u,\tr'',s} = \measureTerm{u,\tr_L\tr_R,s}$ and $\tr'$ is minimal by $\measureTerm{\cdot}$, we already showed that $\ell \rightarrow r$ must be in $\En$. This is a contradiction with the maximality of $\tr''$ as, in such a case, $\rightrwlabel{p}{\sigma}{\ell}{r} \RelPara \leftrwlabel{p}{\sigma}{r}{\ell}$. This concludes the proof of the first item of the lemma. 

We now show the second and third properties of the lemma. In particular, we show the following sub-property: if $\tr_0 = \tr_1 \rightrwlabel{p}{\sigma}{\ell}{r} \tr_E \rightrwlabel{p'}{\sigma'}{\ell'}{r'} \tr_2$ and $(\ell' \rightarrow r') \in \R$ and $p \leq p'$ and $t_0 \Rightrwstep{\tr_0} t_5$ smaller than $\M$ then there exists a syntactic rewrite trace $\tr_T$ such that:
\begin{itemize}
\item $t_0 \Leftrightrwstep{\tr_T} t_5$ which is smaller than $\M$,
\item $\measure{\tr_T} \lessRwLbl{\Pset(\M)} \measure{\tr_0}$.
\end{itemize}

Let $\tr_E$ be the smallest rewrite trace (in term of $|\tr_E|$) such that $\tr'_0 \RelPara \tr_0$ and $\tr'_0 = \tr_1 \rightrwlabel{p}{\sigma}{\ell}{r} \tr_E \rightrwlabel{p'}{\sigma'}{\ell'}{r'} \tr_2$ and $(\ell' \rightarrow r') \in \R$ and $p \leq p'$. There exist $t_1,\ldots,t_4$ such that:
\[
t_0 \Rightrwstep{\tr_1} t_1 \rightrwstep{p}{\sigma}{\ell}{r} t_2 \Rightrwstep{\tr_E} t_3 \rightrwstep{p'}{\sigma'}{\ell'}{r'} t_4 \Rightrwstep{\tr_2} t_5
\]
If $\rightrwlabel{p''}{\sigma''}{\ell''}{r''} \in \tr_E$ and $(\ell'' \rightarrow r'') \in \R$ then by minimality of $\tr_E$, we know that $p \not\leq p''$ and $p'' \not\leq p'$. 

Let us show that for all $\omega'' = \rightrwlabel{p''}{\sigma''}{\ell''}{r''} \in \tr_E$, $p'' \not\RelPara p'$. By contradiction, w.l.o.g., assume that $\tr_E = \tr_A \omega'' \tr_B$, and $p'' \RelPara p'$ and for all $\omega_2 \in \tr_B$, $\pos(\omega_2) \not\RelPara p'$. We deduce that for all $\omega_2 \in \tr_B$, $\pos(\omega_2) > p'$. Moreover, as $p'' \RelPara p'$, we deduce that $\pos(\omega_2) \RelPara p''$. This would imply that $\tr_0 \RelPara \tr_1 \rightrwlabel{p}{\sigma}{\ell}{r} \tr_A \tr_B \rightrwlabel{p'}{\sigma'}{\ell'}{r'} \omega''  \tr_2$, which is a contradiction with the minimality of $\tr_E$. 

As we shown that $\omega'' = \rightrwlabel{p''}{\sigma''}{\ell''}{r''} \in \tr_E$, $p'' \not\RelPara p'$ and $p'' \not\leq p'$, we deduce that $p' < p''$ and so $p < p''$. Once again by minimality of $\tr_E$, we deduce that $(\ell'' \rightarrow r'') \not\in \R$, i.e. $(\ell'' \rightarrow r'') \in \En$. It implies that $t_1 \rightrwstep{p}{\sigma}{\ell}{r} t_2 \rightrwstep[\En]{p'}{\sigma'}{\ell'}{r'} t_4$. Hence by \Cref{lem:combined-X-right}, we deduce that there exist a syntactic trace $\tr_S$ from $\En \cup \R$ such that:
\begin{itemize}
\item $t_1 \Leftrightrwstep{\tr_S} t_4$ which is smaller than $\M$,
\item $\measure{\tr_S} \lessRwLbl{\Pset(\M)} \measure{\rightrwlabel{p}{\sigma}{\ell}{r}\rightrwlabel{p'}{\sigma'}{\ell'}{r'}[\En]}$.
\end{itemize}
As $\measure{\rightrwlabel{p}{\sigma}{\ell}{r}\rightrwlabel{p'}{\sigma'}{\ell'}{r'}[\En]} = \measure{\rightrwlabel{p}{\sigma}{\ell}{r}\rightrwlabel{p'}{\sigma'}{\ell'}{r'}} \leqRwLbl{\Pset(\M)} \measure{\rightrwlabel{p}{\sigma}{\ell}{r}\tr_E\rightrwlabel{p'}{\sigma'}{\ell'}{r'}}$, we obtain $\measure{\tr_1 \tr_S \tr_2} \lessRwLbl{\Pset(\M)} \measure{\tr'_0}$. We conclude by taking $\tr_T = \tr_1 \tr_S \tr_2$.

The second property of the lemma directly holds by applying the sub-property on $\tr_0 = \tr_L$ leading to $\measure{\tr_1 \tr_S \tr_2} \lessRwLbl{\Pset(\M)} \measure{\tr'_0} \eqRwLbl{\Pset(\M)} \measure{\tr_0}$ and so $\measure{\tr_1 \tr_S \tr_2 \tr_R} \lessRwLbl{\Pset(\M)} \measure{\tr}$ yielding a contradiction with the minimality of $\tr$ by $\lessRwLbl{\Pset(\M)}$. The third property of the lemma holds by applying the sub-property on $\tr_0 = \reverse{\tr_R}$ which will also lead to a contradiction with the minimality of $\tr$ by $\lessRwLbl{\Pset(\M)}$.
\end{proof}

Recall we placed ourselves within the hypotheses of \Cref{th:generation of rewrite theory}. Hence \GenExtended{$E',\Rn,\En$} terminates and returns $\R$ such that for all $(\ell \rightarrow r) \in \R$, $\vars{r} \subseteq \vars{\ell}$. We also denote $E = E' \cup \Rn^= \cup \En$ and $T = (>,\R,\Rn,\En,\En)$. 


\begin{restatable}{lemma}{lemfinalshape}
\label{lem:final_shape}
For all ground terms $u,v$, if $u =_E v$ then there exist a ground term $w$ and two rewrite traces $\tr_L$ and $\tr_R$ such that: 
\begin{itemize}
\item rules in $\tr_L$ and $\tr_R$ are from $\R \cup \En$, and
\item $u \Rightrwstep{\tr_L} w \Leftrwstep{\tr_R} v$, and 
\item $\tr_L,\tr_R$ are both ordered by decreasing position. 
\end{itemize}
\end{restatable}

\begin{proof}
As $u,v$ are ground and $u =_E v$, we deduce that $u = t_0 \rwsteps{\R} t_1 \rwsteps{\R} t_2 \rwsteps{\R} \ldots \rwsteps{\R} t_n = v$ where $\R$ is the rewrite system composed of either rules from $\Rn$ or rules $\ell \rightarrow r$ where either $(\ell = r) \in \En \cup E'$ or $(r = \ell) \in \En \cup E'$. Note that some of the $t_i$s could potentially contain variables, i.e. not be ground. As rewriting is stable by application of substitutions, we can take $\sigma$ with $\dom{\sigma} = \vars{t_1,\ldots,t_n}$ and $\img{\sigma} \subseteq \N$. Note that the image of $\sigma$ does not really matter other than being ground. As $u,v$ are ground, we thus obtain that $u = t_0 = t_0\sigma \rwsteps{\R} t_1\sigma \rwsteps{\R} t_2\sigma \rwsteps{\R} \ldots \rwsteps{\R} t_n\sigma = t_n = v$. Since $>$ is $\En$-total, we can thus order each pair $(t_i\sigma,t_{i+1}\sigma)$ with $>$, i.e. either $t_i\sigma > t_{i+1}\sigma$ or $t_i\sigma < t_{i+1}\sigma$ or $t_i\sigma =_\En t_{i+1}\sigma$. Moreover, as $>$ is $\En$-compatible, we deduce the existence of $n+1$ ground syntactic rewrite labels $\omega_1,\ldots,\omega_{n+1}$ such that $u = t_0 \Leftrightrwstep{\omega_1} t_1 \Leftrightrwstep{\omega_2} \ldots \Leftrightrwstep{\omega_n} t_n = v$. Denoting $\tr = \omega_1 \ldots \omega_n$ and $\M = \multiset{t_0,\ldots,t_n}$, we deduce that $\tr$ is a ground syntactic trace from $\R$ and $u \Leftrightrwstep{\tr} v$ is smaller than $\M$.

Let us denote by $\R_0,\R_1, \ldots$ the successive values of $\R$ in \Cref{alg:generation}. Similarly, we denote by $\R'_0, \R'_1, \ldots$ the successive values of $\R'$ in \Cref{alg:generation}. In particular, from \Cref{alg:initial value}, we have $\R_0 = \Rn \cup \{\ell \rightarrow r, r\rightarrow \ell \mid (\ell = r) \in E'\}$. By denoting $\R'_0 = \R_0$, we obtain that for all $i \geq 1$, $\R_i =$ \Normalize{$\R_{i-1} \cup \R'_{i-1},\Rn,\En$}, that is $\R_{i-1} \cup \R'_{i-1} \normstep[\Rn,\En]^* \R_i$. Hence the rewrite trace $\tr$ is from $\R_0 \cup \En$. By \Cref{lem:normalisation-idempotent}, we know that for all $i \geq 0$, $\R_{i+1} \cup \bigcup_{j=0}^i \R'_j \normstep[\Rn,\En]^* \R_{i+1}$.

We prove by induction that for all $i \geq 1$, there exists $\tr_i$ a ground syntactic rewrite trace from $\R_i \cup \En$ such that $u \Leftrightrwstep{\tr_i} v$ which is smaller than $\M$. In the base case $i = 1$, on \Cref{alg:first normalisation} of \Cref{alg:generation}, we have $\R_1 =$ \Normalize{$\R_0,\Rn,\En$}, i.e. $\R_0 \normstep[\Rn,\En]^* \R_1$. By \Cref{cor:normalisation}, we deduce that there exists $\tr_1$ a ground syntactic rewrite trace from $\R_1 \cup \En$ such that $u \Leftrightrwstep{\tr_1} v$ which is smaller than $\M$.
In the inductive step $i > 1$, by inductive hypothesis, we know that there exists $\tr_{i-1}$ a ground syntactic rewrite trace from $\R_{i-1} \cup \En$ such that $u \Leftrightrwstep{\tr_{i-1}} v$ which is smaller than $\M$.
By definition, $\R_i =$ \Normalize{$\R_{i-1} \cup \R'_{i-1},\Rn,\En$}, i.e. $\R_{i-1} \cup \R'_{i-1} \normstep[\Rn,\En]^* \R_i$. By \Cref{lem:properties-normalisation}, $\R_{i-1} \cup \R'_{i-1} \cup \bigcup_{j=0}^{i-1} \R'_j \normstep[\Rn,\En]^* \R_i \cup \bigcup_{j=0}^{i-1} \R'_j \normstep[\Rn,\En]^* \R_i$. 
Note that $\tr_{i-1}$ being a ground syntactic rewrite trace from $\R_{i-1}$ implies that $\tr_{i-1}$ is also from $\R_{i-1} \cup \bigcup_{j=0}^{i-1} \R'_j \cup \En$. As $\R'_0$ contains $\Rn$, we can apply \Cref{cor:normalisation} to deduce that there exists $\tr_i$ a ground syntactic rewrite trace from $\R_i \cup \En$ such that $u \Leftrightrwstep{\tr_i} v$ which is smaller than $\M$.

We now complete the proof our the lemma.
As the algorithm terminates, there exists $N \in \mathbb{N}$ such that $\R_N \cup \R'_N \normstep[\Rn,\En]^* \R_N$, i.e. it reached a fix-point. From \Cref{alg:generation-forall,alg:generation-right-right,alg:generation-left-right,alg:generation-inverse} of \Cref{alg:generation}, we deduce that $\R_N$ is saturated w.r.t. $\Rn$. Moreover, we already proved that for all $(\ell \rightarrow r) \in \R_N$, $\ell =_E r$ and by hypothesis, $E$ is not trivial. It implies that $\R_N^=$ is also not trivial. Recall that $\M$ is finite and ground. Since $\tr_N$ is a ground syntactic rewrite trace from $\R_N \cup \En$ such that $u \Leftrightrwstep{\tr_i} v$ smaller than $\M$, we can take $\tr_{min}$ to be a ground syntactic rewrite trace from $\R_N \cup \En$ such that $u \Leftrightrwstep{\tr_i} v$ which is smaller than $\M$ and $\tr_{min}$ is minimal by $\lessRwLbl{\Pset(\M)}$.

By \Cref{lem:shape-minimal}, we deduce that there exist a term $w$ and three ground syntactic well-ordered rewrite traces $\tr'_{min}$, $\tr_L$, $\tr_R$ from $\En \cup \R_N$ such that:
\begin{enumerate}
\item $\tr'_{min} \RelPara \tr_{min}$ and $\tr'_{min} = \tr_L \tr_R$ and $u \Rightrwstep{\tr_L} w \Leftrwstep{\tr_R} v$\label{item-theorem-shape}
\item for all $\tr'_L \RelPara \tr_L$, if $\tr'_L = \tr_A \rightrwlabel{p}{\alpha}{\ell}{r} \tr_E \rightrwlabel{p'}{\alpha'}{\ell'}{r'} \tr_B$ and $(\ell' \rightarrow r') \in \R$ then $p \not\leq p'$\label{item-theorem-L}
\item for all $\tr'_R \RelPara \tr_R$, if $\tr'_R = \tr_A \leftrwlabel{p}{\alpha}{\ell}{r} \tr_E \leftrwlabel{p'}{\alpha'}{\ell'}{r'} \tr_B$ and $(\ell \rightarrow r) \in \R$ then $p' \not\leq p$\label{item-theorem-R}
\end{enumerate}
This allows us to conclude.
\end{proof}


\thgenerationsignature*

\begin{proof}
We first show that $T = (>,\R,\Rn,\En,\En)$ is a rewrite theory. This is a simple matter as \Cref{S:subset,S:order} are given as assumptions. Finally, as we assumed that for all $(\ell \rightarrow r) \in \R$, we have $\vars{r} \subseteq \vars{\ell}$, and since \Cref{alg:final_modification} of \Cref{alg:generation} ensures that $\R$ contains $f(x_1,\ldots,x_n) \rightarrow f(x_1,\ldots,x_n)$ for all $f \in \F$ with arity $n$, \Cref{S:std} is guaranteed.

Let us now prove that $T$ mimics $E$. Amongst the three properties required to show that $T$ mimics $E$, only the last one, i.e. \Cref{M:ind1}, is difficult. \Cref{M:EaE} is directly obtained, since $\En = \Ea$ and $\En \cup \Rn^= \cup E' = E$. The proof of \Cref{M:eqcomplete} is done a simple induction on the number of loop the algorithm went through in \Cref{alg:loop}. Indeed, all rules $\ell \rightarrow r$ in the initial value $\R$ on \Cref{alg:initial value} of \Cref{alg:generation} satisfy $\ell =_E r$. By applying \Cref{lem:overlapset_complete} and noticing that the normalisation rules \RNormL and \RNormR preserve this invariant, since $\Rn^= \cup \En \subseteq E$, we obtain that $\R$ satisfies \Cref{M:eqcomplete}.

We now focus on the proof of \Cref{M:ind1}. Consider $f(t_1,\ldots,t_n) =_E t$ and $\M = \{t_1,\ldots,t_n,t\}$. If $\nf{T}{E}{\M}$ holds then by \Cref{def:normal form} we know that $u\theta =_\En \minOrd[>]{E}(u\theta)$ for all $u \in \M$ where $\theta$ is an injective substitution closing $f(t_1,\ldots,t_n)$ and $t$. Hence, by \Cref{lem:final_shape}, we have 
\[
    f(t_1\theta,\ldots,t_n\theta) \Rightrwstep{\tr_L} w \Leftrwstep{\tr_R} t\theta
\]
However, $w \Leftrwstep{\tr_R} t\theta$ implies that $w < t\theta$ or $w =_\En t\theta$. Since $t\theta =_\En \minOrd[>]{E}(t\theta)$, we deduce that $w =_\En t\theta$. 

As $\tr_L$ is ordered by decreasing position, we know that there can only be one rule applied at root position, i.e. $\varepsilon$, that can strictly decrease the terms. Moreover, as $t_i\theta =_\En \minOrd[>]{E}(t_i\theta)$ for all $i$, we also know that any rule application in $\tr_L$ before the one at root position cannot strictly decrease the terms (otherwise $\minOrd[>]{E}(t_i\theta)$ would not be minimal). In other words, either 
\begin{inparaenum}[(a)]
\item $f(t_1\theta,\ldots,t_n\theta) =_\En w$; or
\item we have:
\[
    f(t_1\theta,\ldots,t_n\theta) \Rightrwstep{\tr_1} s \rightrwstep{\varepsilon}{\alpha}{\ell}{r} w' \Rightrwstep{\tr_2} w
\]
where all rules in $\tr_1$ and $\tr_2$ are in $\En$, and $(\ell \rightarrow r) \in \R$ with $\ell\alpha > r\alpha$. Additionally, we also know that the positions in $\tr_1$ are different from $\varepsilon$, as $\tr_L$ is ordered by decreasing position. Case (a) implies $f(t_1\theta,\ldots,t_n\theta) =_\En t\theta$ and so $f(t_1,\ldots,t_n) =_\En t$. We conclude by taking the rule $f(x_1,\ldots,x_n) \rightarrow f(x_1,\ldots,x_n)$ and $\sigma = \{ x_1 \mapsto t_1, \ldots, x_n \mapsto t_n \}$. Case (b) implies that $t_i\theta =_\En \ell_{|i}\alpha$ and so $t_i =_\En \ell_{|i}\alpha\theta^{-1}$ for all $i$, and $r\alpha\theta^{-1} =_\En t$. We thus conclude with the rule $\ell \rightarrow r$ and $\sigma = \alpha\theta^{-1}$.
\end{inparaenum}
\end{proof}


\subsection{Proof of Section~\ref{sec:optimisations}}
\label{sec:app-optimisation}

\lemcleanup*

\begin{proof}
Let $\R_0, \R_1, \ldots$ the successive values of $\R$ in \Cref{alg:superfluous}. We prove by induction on $n$ that $(>,\R_n,\Rn,\En,\Ea)$ mimics $E$. The base case $n=0$ is given by hypothesis. In the inductive step, we know that $(>,\R_{n-1},\Rn,\En,\Ea)$ mimics $E$. The only non-trivial property to prove in order to show that $(>,\R_n,\Rn,\En,\Ea)$ mimics $E$ is \Cref{M:ind1}. 
Take $f(t_1,\ldots,t_n) =_E t$ and $\nf{T}{E}{\{t_1,\ldots,t_n,t\}}$. We know that there exist $\sigma$ and $f(s_1,\ldots,s_n) \rightarrow s$ in $R_{n-1}$ such that $t =_\Ea s\sigma$ and for all $i \in \{1, \ldots, n\}$, $t_i =_\Ea s_i\sigma$. If $f(s_1,\ldots,s_n) \rightarrow s$ is in $R_n$ then the result trivially holds. 
Else, from \Cref{line:superfluous-while}, there exist a substitution $\alpha$ and a rule $f(\ell_1,\ldots,\ell_n) \rightarrow \ell$ in $\R_n$ such that $\ell\alpha \eqE{\Ea} s$ and for all $i \in \{1, \ldots, n\}$, $\ell_i\alpha \eqE{\Ea} s_i$. Therefore, it implies $\ell\alpha\sigma =_\Ea t$ and for all $i \in \{1, \ldots, n\}$, $\ell_i\alpha\sigma \eqE{\Ea} t_i$. This allows us to conclude. 
\end{proof}

\subsection{Order \texorpdfstring{$E$}{E}-compatible}
\label{sec:app-order}


\lembuildingEstrongorder*

\begin{proof}
We define $>'_2$ as the following relation: $u >'_2 v$ if and only if 
\begin{itemize}
\item either $u > v$ 
\item or $u \equiv v$ and there exist $s,t$ ground terms and a context $C[\_]$ such that $u =_E C[s]$, $v =_E C[t]$ and $s >_1 t$
\end{itemize}
Finally, we define the $>_2$ as being the transitive closure of $>'_2$.

Let us prove that $>_2$ is a $E$-strong reduction order compatible with $\R$. 

\paragraph{Closed by application of contexts and substitutions:}

Assume $u >'_2 v$ and let $D[\_]$ be a context and $\sigma$ be a substitution. If $u > v$ then as $>$ is a reduction order, we have that $D[u\sigma] > D[v\sigma]$ and so $D[u\sigma] >'_2 D[v\sigma]$.  Otherwise, $u \equiv v$ and there exists $s,t$ ground terms and a context $C[\_]$ such that $u =_E C[s]$, $v =_E C[t]$ and $s >_1 t$. Since $s$ and $t$ are ground, there exists a context $D'[\_] = D[C\sigma[\_]]$ such that $D[u\sigma] =_E D'[s]$ and $D[v\sigma] =_E D'[s]$. Since $\equiv$ is close by application of contexts and substitutions, we deduce that $D[u\sigma] \equiv D[v\sigma]$. This allows us to deduce that $D[u\sigma] >'_2 D[v\sigma]$ meaning that $>'_2$ is closed by application of contexts and substitutions.
Since $>_2$ is the transitive closure of $>'_2$, we conclude that $>_2$ is closed by application of contexts and substitutions.

\paragraph{Transitive:} By definition of $>_2$.

\paragraph{Asymmetric and irreflexive:} if $a >_2 b$ and $b >_2 a$ then there exists $a >'_2 u_1 >'_2 \ldots >'_2 u_n >'_2 b >'_2 v_1 >'_2 \ldots >'_2 v_m >'_2 a$. Since $>$ is $\equiv$-compatible, we deduce from definition of $>'_2$ that $a >'_2 u_1 >'_2 \ldots >'_2 u_n >'_2 b$ implies that either $a > b$ or $a \equiv b$ (and $a \equiv u_1 \equiv \ldots u_n \equiv b$). Similarly, $b >'_2 v_1 >'_2 \ldots >'_2 v_m >'_2 a$ implies that either $b > a$ or $a \equiv b$. If $a > b$ then we cannot have $b > a$ since $>$ is a strict order and we also cannot have $a \equiv b$ since \emph{$\equiv$-compatibility} would imply that $a > a$ which is prevented by $>$ being a strict order. With the same reasoning, we can show that $b > a$. It entails that $a \equiv b$. Note that by definition of $>'_2$ and since $>_1$ is closed by application of contexts and $E$-compatible, we deduce that $a >_1 u_1 >_1 \ldots u_n >_1 b >_1 v_1 >_1 \ldots >_1 v_m >_1 a$. This is in contradiction with $>_1$ being a strict order.

\paragraph{Well founded:} We only need to show that the relation $>'_2$ is well founded since $>_2$ is the transitive closure of $>'_2$. Consider an infinite sequence $u_1 >'_2 u_2 >'_2 u_3 >'_2 \ldots$. As $>$ is well founded and by $\equiv$-compatibility we deduce that we cannot have an infinite sub-sequence $u_{i} > u_{i+1} > u_{i+2} > \ldots$ and we cannot have infinite alternation of $\equiv$ and $>$. This allows us to deduce that there exists an infinite sub-sequence $u_{j} >_1 u_{j+1} >_1 u_{j+2} >_1 \ldots$ (recall that $>_1$ is closed by application of contexts). This contradicts the fact that $>_1$ is well-founded.

\paragraph{$E$-compatible:} Take $s =_E u >'_2 v =_E t$. By definition, $s =_E u$ implies $s \equiv u$, and $v =_E t$ implies $v \equiv t$. If $u > v$ then we deduce that $s > t$ (by $\equiv$-compatibility) and so $s >'_2 t$. Otherwise $u \equiv v$ and there exist $a,b$ ground terms and a context $C[\_]$ such that $u =_E C[a]$, $v =_E C[b]$ and $a >_1 b$. Hence $s =_E C[a]$ and $t =_E C[b]$ and $a >_1 b$. Therefore $s >'_2 t$. To complete the proof, we can notice that if $s =_E u >_2 v =_E t$ then $s =_E u >'_2 a_1 >'_2 \ldots >'_2 a_n >'_2 v =_E t$. As $s =_E u >'_2 a_1$ implies $s >'_2 a_1$, and $a_n >'_2 v =_E t$ implies $a_n >'_2 t$, we conclude that $s >'_2 a_1 >'_2 \ldots >'_2 a_n >'_2 t$ and so $s >_2 t$.

\paragraph{$E$-total:} Let $a,b$ be ground terms. By $\equiv$-total property, we know that either $a > b$ or $b > a$ or $a \equiv b$. In the two first case, we obtain $a >_2 b$ and $b >_2 a$ respectively. In the last case, as $a$ and $b$ are ground, we know that either $a =_E b$ or $a >_1 b$ or $b >_1 a$. If $a >_1 b$ (resp. $b >_1 a$) then by taking the empty context $C[\_] = \_$, we obtain that $a =_E C[a]$ and $b =_E C[b]$ which leads to $a >'_2 b$ (resp. $b >'_2 a$) and so $a >_2 b$ (resp. $b >_2 a$).

\paragraph{Compatible with $\R$:} By hypothesis, for all $(\ell \rightarrow r) \in \R$, we have $\ell > r$ and so $\ell >_2 r$.

\paragraph{Stable by renaming:} Recall that due to $\equiv$-compatibility of $>$ , we know that if $u \sim_1 \ldots \sim_i s > t \sim_{i+1} \ldots \sim_n v$ with $\sim_i \in \{ {\equiv}, {>}\}$ for all $i$ then $u > v$. Furthermore, by definition of $>_2$, if $u >_2 v$ and $u \not> v$ then $u \equiv v$ and $u >_1 v$. Hence, we obtain that $u >_2 v$ implies that $u > v$ or $u >_1 v$. 

Let us now look at the case where we order names: If $a >_2 b$ then $a > b$ or $a >_1 b$. By hypothesis, $a > b$ implies $a >_1 b$. Therefore $a >_2 b$ implies $a >_1 b$. Assume now that $a>_1 b$. By $\equiv$-totality, we know that either $a \equiv b$, $a > b$ or $b > a$. The latter case is in contradiction with the fact that it would imply $b >_1 a$. Hence, either $a > b$ or $a \equiv b$. In both cases, as $a >_1 b$, we deduce that $a >_2 b$. This conclude the proof of $a >_2 b$ iff $a >_1 b$.
This property allows us to prove that a renaming $\rho$ preserves $>_2$ iff $\rho$ preserves $>_1$.

Let us now show that if $\rho$ preserves $>_2$ then $\rho$ preserves $>$. Take $a,b \in \dom{\rho}$. If $a > b$ then $a >_2 b$ which implies $a\rho >_2 b\rho$. By definition, either $a\rho > b\rho$ or $a\rho \equiv b\rho$. In the latter case, since $\equiv$ is stable by application of renaming, we deduce that $a \equiv b$. However, by $\equiv$-compatibility, it contradicts our hypothesis $a > b$. Thus $a\rho > b\rho$. Let us assume that $a\rho > b\rho$. We can in fact apply the same reasoning: $a\rho > b\rho$ implies $a\rho >_2 b\rho$ and $a >_2 b$. By definition, either $a > b$ or $a \equiv b$. Once again the latter case is in contradiction with $a\rho >_2 b\rho$ by $\equiv$-compatibility and the stability of $\equiv$ by application of renamings.

We can now complete the proof of stability by renaming: Consider a renaming $\rho$ preserving $>_2$ and two terms $u$ and $v$ such that $\names{u,v} \subseteq \dom{\rho}$. We already proved that $\rho$ preserves both $>_1$ and $>$. We first prove that $u >'_2 v$ if and only if $u\rho >'_2 v\rho$. Assume that $u >'_2 v$. By definition we have:
\begin{itemize}
\item either $u > v$: In such a case, as $>$ is stable by renaming, we deduce that $u\rho > v\rho$ and so $u\rho >_2 v\rho$.
\item or $u \equiv v$ and there exist ground terms $s,t$ and a context $C[\_]$ such that $u =_E C[s]$, $v =_E C[t]$ and $s >_1 t$: In such a case, as $\equiv$ is stable by application of renamings, $u\rho \equiv v\rho$. Moreover, as $\rho$ preserves $>_1$ and $>_1$ is stable by renaming, we deduce that $u\rho =_E C\rho[s\rho]$, $v\rho =_E C\rho[t\rho]$ and $s\rho >_1 t\rho$. Therefore $u\rho >_2 v\rho$.
\end{itemize}
The other implication works in a similar fashion with a small difference when $u\rho >'_2 v\rho$ due to $u\rho =_E C[s]$, $v\rho =_E C[t]$ and $s >_1 t$: Recall that $\names{u,v} \subseteq \dom{\rho}$. Hence $u =_E C\rho^{-1}[s\rho^{-1}]$ and $v =_E C\rho^{-1}[t\rho^{-1}]$. Moreover, $\rho$ preserves $>_1$ implies $\rho^{-1}$ preserves $>_1$, meaning that $s\rho^{-1} >_1 t\rho^{-1}$. This allows us to conclude that $u>_2 v$.
\end{proof}


\begin{corollary}
Consider $\R_{\mathcal{AG}}$ the rewrite system $AC$-convergent for $\mathcal{AG}$, first proposed by Lankford~\cite{hullot1980catalogue}, defined below.
\[
\R_{\mathcal{AG}} = 
\left\{ 
\begin{array}{r@{\ \rightarrow\ }l r@{\ \rightarrow\ }l}
x * 1  & x   & {x^{-1}}^{-1} & x \\
1^{-1} & 1   & (x^{-1} * y)^{-1} & x * y^{-1}\\ 
x * x^{-1} & 1 & x * (x^{-1} * y) & y\\
x^{-1} * y^{-1} & (x * y)^{-1} & x^{-1} * (y^{-1} * z) & (x * y)^{-1} * z\\
(x * y)^{-1} * y & x^{-1} & (x * y)^{-1} * (y * z) & x^{-1} * z\\
\end{array}
\right. 
\]
There exists a $E$-strong reduction order compatible with $\R_{\mathcal{AG}}$.
\end{corollary}

\begin{proof}
For a term $u$, let us note $\#(u)$ the number of function symbols in $u$.

We define $\equiv$ the smallest equivalence relation such that $u \equiv v$ implies $u =_{AC} v$ or there exist $a,b$ ground terms and a term context $C[\_]$ such that $u = C[a]$ and $C[b] = v$ and $\#(a) = \#(b)$.

Let us now define $>'$ be the order defined as $u >' v$ if and only if either $u \equiv \circ \rwstep{\R_{\mathcal{AG}}} \circ \equiv v$ or else there exist $a,b$ ground terms and a term context $C[\_]$ such that $u \equiv C[a]$ and $C[b] \equiv v$ and $\#(a) > \#(b)$. Finally, the relation $>_\mathcal{AG}$ is the transitive closure of $>'$.

By definition of $AC$, we can easily notice that $u \equiv v$ implies that $\#(u) = \#(v)$. Additionally, by definition of $\R_{\mathcal{AG}}$, we also have that for $\ell \rightarrow r \in \R_{\mathcal{AG}}$, for all substitutions $\sigma$, $\#(\ell\sigma) > \#(r\sigma)$. Hence $u \rwstepC[AC]{\R_{\mathcal{AG}}} v$ implies $\#(u) > \#(v)$. Thus, we obtain that $>'$ is well founded and so does $>_\mathcal{AG}$. 
It can also easily be shown that $>_\mathcal{AG}$ is closed by application of substitutions and contexts. Hence $>_\mathcal{AG}$ is a reduction order. 

Take $a,b$ two ground terms. Assume that $a \equiv \circ \rwstepC[AC]{\R_{\mathcal{AG}}} \circ \equiv b$ or $b \equiv \circ \rwstepC[AC]{\R_{\mathcal{AG}}} \circ \equiv a$ , we directly have that $a >_\mathcal{AG} b$ or $b >_\mathcal{AG} a$. Else, we naturally have either $\#(a) > \#(b)$ or $\#(a) = \#(b)$ or $\#(a) < \#(b)$. In the first case, we have $a >_\mathcal{AG} b$. In the second case, $a \equiv b$. In the third case, $b >_\mathcal{AG} a$. This prove $\equiv$-totality of $>_\mathcal{AG}$.

Finally, assume $s \equiv u >' v \equiv t$. If $u >' v$  because $u \equiv \circ \rwstepC[AC]{\R_{\mathcal{AG}}} \circ \equiv v$ then we naturally have $s \equiv \circ \rwstepC[AC]{\R_{\mathcal{AG}}} \circ \equiv t$ and so $s >' t$. Otherwise, $u \equiv C[a]$ and $C[b] \equiv v$ and $\#(a) > \#(b)$ for some ground terms $a,b$ and context $C[\_]$. Hence, we also naturally have $s \equiv C[a]$ and $C[b] \equiv t$ which allows us to conclude that $s >' t$. As $>_\mathcal{AG}$ is the transitive closure of $>'$, it concludes the proof of $\equiv$-compatibility.

We conclude the existence of a $E$-strong reduction order compatible with $\R_{\mathcal{AG}}$ by applying \Cref{lem:building E-strong reduction order} with the order $>_\mathcal{AG}$ and the order in~\cite{DBLP:conf/rta/Rubio99} that was shown to be a $AC$-strong reduction order (\Cref{sec:discussion}).
\end{proof}

\end{document}